\documentclass[a4paper,12pt,oneside]{amsart}

\usepackage{epsf,color}
\usepackage{graphicx,mathrsfs,phonetic,amsmath,mathabx,stmaryrd}
\usepackage[bookmarks=true,%
    colorlinks=true,%
    linkcolor=blue,%
   citecolor=blue,%
   filecolor=blue,%
    menucolor=blue,%
    urlcolor=blue,%
    breaklinks=true]{hyperref}

\newcommand{\bs}[1]{\ensuremath{\boldsymbol{#1}}}
\newcommand{\dd}{\mathrm{d}}

\def\Res{\mathop{\mathrm{Res}}}
\def\Aut{\mathop{\mathrm{Aut}}}

\def\DP{\widehat{\phi}}
\def\d{\partial}
\def\oM{\overline{\mathcal{M}}}
\newcommand{\beq}{\begin{equation}}
\newcommand{\eeq}{\end{equation}}
\newcommand{\bea}{\begin{eqnarray}}
\newcommand{\eea}{\end{eqnarray}}
\newcommand{\Tr}{{\rm Tr}\,}

\newtheorem{theorem}{Theorem}[section]

\newtheorem{proposition}[theorem]{Proposition}
\newtheorem{lemma}[theorem]{Lemma}
\newtheorem{corollary}[theorem]{Corollary}

\theoremstyle{remark}

\newtheorem{example}{Example}[section]
\newtheorem{definition}[example]{Definition}
\newtheorem{notation}[example]{Notation}
\newtheorem{remark}[example]{Remark}

\newcommand*\MapsTo{%
  \xrightarrow[\raisebox{0.25 em}{\smash{\ensuremath{\sim}}}]{}%
}

\begin{document}

\title[Blobbed topological recursion]{Blobbed topological recursion: \\ properties and applications}
\author{Ga\"etan Borot}
\address{Max Planck Institut f\"ur Mathematik, Vivatsgasse 7, 53111 Bonn, Germany.}
\email{gborot@mpim-bonn.mpg.de}
\author{Sergey Shadrin}
\address{Korteweg de Vries Instituut voor Wiskunde, Universiteit van Amsterdam, P.O. Box 94248, 
1090 GE Amsterdam, Netherlands.}
\email{s.shadrin@uva.nl}

\thanks{The authors would like to thank the organizers of the workshop ``Moduli Spaces and Integrable Systems", held in Banff in August 2013, where this project was initiated. G.B. is grateful to the University of Amsterdam for hospitality during the conduct of this project.}

\begin{abstract}
We study the set of solutions $(\omega_{g,n})_{g \geq 0,n \geq 1}$ of abstract loop equations. We prove that $\omega_{g,n}$ is determined by its purely holomorphic part: this results in a decomposition that we call ``blobbed topological recursion". This is a generalization of the theory of the topological recursion, in which the initial data $(\omega_{0,1},\omega_{0,2})$ is enriched by non-zero symmetric holomorphic forms in $n$ variables $(\phi_{g,n})_{2g - 2 + n > 0}$. In particular, we establish for any solution of abstract loop equations: (1) a graphical representation of $\omega_{g,n}$ in terms of $\phi_{g,n}$; (2) a graphical representation of $\omega_{g,n}$ in terms of intersection numbers on the moduli space of curves; (3) variational formulae under infinitesimal transformation of $\phi_{g,n}$ ; (4) a definition for the free energies $\omega_{g,0} = F_g$ respecting the variational formulae. We discuss in detail the application to the multi-trace matrix model and enumeration of stuffed maps.
\end{abstract}

\maketitle

\tableofcontents

% % % % % % % % % % % % % % % %
% % % % % % % % % % % % % % % %
% % % % % % % % % % % % % % % %
% % % % % % % % % % % % % % % %
% % % % % % % % % % % % % % % %
% % % % % % % % % % % % % % % %

\section{Introduction}

Loop equations are ubiquitous in mathematical physics, and their solution compute interesting quantities in random matrix theory, combinatorics of surfaces, enumerative geometry (Hurwitz and Gromov-Witten theory), topological quantum field theories, etc.  Specifically in each of these problems, one can often derive functional relations between quantities (that bear the name of Schwinger-Dyson equations, Virasoro constraints, Tutte's equations, cut-and-join relations, Ward identities, etc.) and then show that they imply loop equations. The latter take a universal form, and here we refer to the framework of ``abstract loop equations" formulated in \cite{BEO}. The quantity subject to these abstract loop equations is a collection $(\omega_{g,n})_{g,n}$ of meromorphic forms in $n$ variables indexed by two integers $n \geq 1$ and $g \geq 0$. In this article, we propose a general study the set $\mathcal{M}$ of solutions of abstract loop equations. The definition of $\mathcal{M}$ depends on a few data that are made explicit in \S~\ref{extdsn}, and $\mathcal{M}$ itself is in any case infinite dimensional.

The topological recursion of \cite{EOFg} already constructed a subset $\mathcal{M}^{0}$ of solutions to abstract loop equations, by a recursion on $2g - 2 + n > 0$ starting from the initial data $\omega_{0,1}$ and $\omega_{0,2}$. The solutions in $\mathcal{M}^0$ have the property that the holomorphic part of $\omega_{g,n}$ is normalized uniformly for all $n$ and $g$. The topological recursion can be studied \textit{per se} and enjoys beautiful properties: Seiberg-Witten like formulas for infinitesimal variations of the initial data \cite{EORev}, symplectic invariance \cite{EO2MM}, representation in terms of intersection numbers on $\overline{\mathcal{M}}_{g,n}$ \cite{Ekappa,Einter}, etc. It has received many applications in algebraic geometry \cite{EOwp,Norbu,DBOSS,BKMP,Zhou2,EOBKMP}, in relation with integrable systems \cite{BEInt,MDHitchin}, in knot theory \cite{DiFuji2,BEMknots,BEknots,BESeifert}, and we also refer to the examples given in \cite{EORev,BEO} for applications to matrix models and statistical physics.

As we shall review in Section~\ref{motiv}, the necessity of considering the full $\mathcal{M}$ instead of its subset $\mathcal{M}^{0}$ came matrix models with multi-trace interactions \cite{Bstuff} and combinatorics of surfaces obtained by gluing elementary $2$-cells of any topology. We find in Theorem~\ref{bllb} that $\mathcal{M}$ is an extension of $\mathcal{M}^0$ by the adjunction of extra initial data $(\varphi_{g,n})_{g,n}$ for each $g$ and $n$ with $2g - 2 + n > 0$. As a matter of fact, $\varphi_{g,n}$ is the ``purely holomorphic part" of $\omega_{g,n}$, Theorem~\ref{bllb} shows that any $(\omega_{g,n})_{g,n} \in \mathcal{M}$ is determined by $\omega_{0,1}$, $\omega_{0,2}$ and the $(\varphi_{g,n})_{2g - 2 + n > 0}$, explicitly in Equations~\eqref{P1wgn} and \eqref{HAPB}. We decide to call $\varphi_{g,n}$ the blobs, and this construction is called the ``blobbed topological recursion". 

We prove that many of the interesting properties of the topological recursion -- maybe to the exception of the relations with integrability -- extend to its blobbed counterpart. In other words, they are structural properties of $\mathcal{M}$. The blobbed topological recursion has several diagrammatic representations that are very helpful in proving general statements about $\mathcal{M}$. We then prove:
\begin{itemize}
\item[$\bullet$] Theorems~\ref{thm:DP} and~\ref{thm:omega-graphs} expressing $\omega_{g,n}$ in terms of intersection numbers. For this purpose, we use a different parametrization of $\mathcal{M}$, denoted $(\phi_{g,n})_{g,n}$, and called ``KdV blobs".
\item[$\bullet$] Theorem~\ref{decun} showing that the even part of the $\omega_{g,n}$'s are decoupled from the odd part.
\item[$\bullet$] Theorem~\ref{vary} computing the infinitesimal generator of the flows $\mathcal{M}$ corresponding to variations of the initial data $\phi_{h,k}$. For $h,k \neq (0,1)$ and $(0,2)$, this only makes sense in $\mathcal{M}$ -- and not in $\mathcal{M}^{0}$. These formulae involve $k$-linear combinations of $\omega$'s, and contain all the necessary information to study solutions of loop equations on families of spectral curves with varying complex structures. 
\item[$\bullet$] Theorem~\ref{popore} associating numbers $F_{g}:= \omega_{g,0}$ (the free energies) to any point in $\mathcal{M} \times \bigoplus_{g \geq 0} \mathbb{C}$, in such a way that the variational equations are respected. The extra data in $\mathbb{C}$ accounts for ``integration constants".
\end{itemize}
We hope that our theory will shed some light on the geometric meaning of the topological recursion, and we think it is a good starting point to investigate further properties (in particular symplectic invariance, quantum curves), generalizations (non-commutative version, base field $\neq \mathbb{C}$), and relations to BV algebras and to homological algebra.

% % % % % % % % % % % % % % % %
% % % % % % % % % % % % % % % %
% % % % % % % % % % % % % % % %
% % % % % % % % % % % % % % % %
% % % % % % % % % % % % % % % %
% % % % % % % % % % % % % % % %

\section{Loop equations and blobbed topological recursion}
\label{S2}
\subsection{Loop equations}
\label{extdsn}

\subsubsection{Local setup} \label{localste} In order to fix the notation, let us recall that we consider a Riemann surface $\Sigma$ which is a disjoint union of finitely many open disks $U_i$, $i \in \llbracket 1,s \rrbracket$. We assume the data of a cover $x\,:\Sigma \rightarrow \,V$. By taking $U_i$ small enough, it is always possible to assume that it contains a single zero $p_i$ of $\dd x$. Throughout the paper we assume that all zeroes of $\dd x$ are simple. Then, again by taking $U_i$ small enough, the deck transformation of the cover $x\,:\,\Sigma \rightarrow V$ defines a holomorphic involution $\sigma_i\,:\,U_i \rightarrow U_i$ so that $x \circ \sigma_i = x$ and $\sigma_i \neq \mathrm{id}$. The $p_i$ are the fixed points of the $\sigma_i$. We find convenient to call this ``local spectral curve''; $\mathcal{O}_{\Sigma}$ is the vector space of holomorphic functions, and $K_{\Sigma}$ the canonical bundle. They are both infinite dimensional since $\Sigma$ is a collection of disks. Nevertheless, it is clear by going in local coordinates that their tensor product is well-defined. $P = \sqcup_{i = 1}^s \{p_i\}$ is the divisor of zeroes of $\dd x$, and $\Delta$ is the diagonal divisor in $\Sigma^2$.

A set of admissible correlators is a sequence $(\omega_{g,n})_{g,n}$ indexed by two integers $g \geq 0$ and $n \geq 1$. We ask that $\omega_{0,1} = y\dd x$ and $\omega_{0,2} = B$ with:
\beq
y \in \mathcal{O}_{\Sigma,P}^{*},\qquad B \in H^0\big(\Sigma^{2},K_{\Sigma}^{\boxtimes\,2}(-2\Delta)\big)^{\mathfrak{S}_{2}}\big|_{1}
\eeq
$\mathcal{O}_{\Sigma,P}^*$ stands for the set of holomorphic functions on $\Sigma$ that do not vanish on $P$. The extra subscript $1$ means that we restrict to bidifferential forms with leading coefficient $1$ on the diagonal, i.e.
$$
B(z_1,z_2) = \frac{\dd\xi(z_1)\dd \xi(z_2)}{(\xi(z_1) - \xi(z_2))^2} + O(1),\qquad z_1 \rightarrow z_2
$$
in any local coordinate $\xi$ around $z_2$. And, for $2g - 2 + n > 0$, we ask:
$$
\omega_{g,n} \in H^0\big(\Sigma^n,K_{\Sigma}^{\boxtimes\,n}(* P)\big)^{\mathfrak{S}_{n}}
$$
As the notation indicates, admissible correlators are invariant under the action of $\mathfrak{S}_{n}$ permuting the $n$ factors of $\Sigma^{n}$.

The abstract loop equations form a list of constraints on admissible correlators. In each degree $(g,n)$ with $g \geq 0$ and $n \geq 1$, it consists of the linear loop equations:
$$
\omega_{g,n}(z,I) + \omega_{g,n}(\sigma_{i}(z),I)
$$
is holomorphic when $z \rightarrow p_i$ (since it is $\sigma$-invariant, it must have at least a simple zero at $p_i$); together with the quadratic loop equations, enforcing that:
\beq
\label{Qng}Q_{g,n}(z,I) := \omega_{g - 1,n + 1}(z,\sigma_{i}(z),I) + \sum_{\substack{J \sqcup J' = I \\ h + h' = g}} \omega_{h,1 + |J|}(z,J) \otimes \omega_{h,1 + |J'|}(\sigma_{i}(z),J')
\eeq
is a holomorphic quadratic form in $z$, with at least a double zero at $z \rightarrow p_i$. Here $I$ is a generic $(n - 1)$-uple of variables in $\Sigma$.
For a given local spectral curve $\mathcal{C}$, we shall study the set $\mathcal{M}_{\mathcal{C}}$ of solutions of the abstract loop equations. It carries a filtration by $n$ and $g$. $(g,n) = (0,1)$ and $(0,2)$ are called ``unstable''. The other degrees satisfy $2g - 2 + n > 0$ and are called ``stable''.

\subsubsection{Remark on global setup}
\label{global}
We call ``spectral curve'' a surjective morphism of Riemann surfaces $x\,:\,\underline{\Sigma} \rightarrow \underline{V}$ which, after restriction to suitable open neighborhoods of the ramification divisor in $\underline{\Sigma}$ and of the branching divisor in $\underline{V}$, defines a local spectral curve in the sense of \S~\ref{localste}. We do not put any assumption on the topology (connectedness, compactness, etc.) of $\underline{\Sigma}$ and $\underline{V}$. 

The notion of admissible correlators in \S~\ref{localste} only depend on the data $\Sigma$ and a divisor $P \subseteq \Sigma$. So, we also have notion of admissible correlators $(\omega_{g,n})_{g,n}$ on a spectral curve: they are meromorphic forms globally defined on $n$ copies of $\overline{\Sigma}$ and respecting the axioms.  Let $\overline{\mathcal{C}}$ be a spectral curve, and $\mathcal{C}$ the corresponding local spectral curve. We define $\mathcal{M}_{\overline{\mathcal{C}}}$ to be the sequences of admissible correlators in $\overline{\mathcal{C}}$, whose restriction to $\mathcal{C}$ is in $\mathcal{M}_{\mathcal{C}}$. For obvious reasons, the results of this Section hold as well for $\mathcal{M}_{\overline{\mathcal{C}}}$.

\subsection{Description of the solution set}

\subsubsection{Polar part and holomorphic part}

\label{sec:PolarPart}

Since $\omega_{0,2}$ has a double pole with leading coefficient $1$, for any meromorphic $1$-form $\lambda$ in $\Sigma$:
$$
\mathcal{P}\lambda(z_0) := \sum_{i = 1}^s \Res_{z \rightarrow p_i} \lambda(z)G_i(z,z_0),\qquad G_{i}(z,z_0)  := \int_{p_i}^{z} \omega_{0,2}(\cdot,z_0)
$$
is a meromorphic $1$-form in $\Sigma$ whose divergent part at $P$ coincide with that of $\lambda$. We call it ``the polar part" of $\lambda$. We can decompose:
$$
\lambda = \mathcal{P}\lambda + \mathcal{H}\lambda
$$
where $\mathcal{H}\lambda$ is now a holomorphic form. We call it the ``holomorphic part" of $\lambda$. $\mathcal{P}$ and $\mathcal{H}$ are projectors, and $\mathcal{P}\circ \mathcal{H} = \mathcal{H}\circ \mathcal{P} = 0$. For meromorphic $1$-forms of $n$ variables, we define $\mathcal{P}_{j}$ and $\mathcal{H}_{j}$ the projectors acting on the $j$-th variable. We stress that the notion of polar part and holomorphic part depend on the data of $\omega_{0,2}$.

\begin{definition}
\label{def0} If $(\omega_{g,n})_{g,n}$ are admissible correlators, we introduce the ``purely polar part" $\omega_{g,n}^{\mathcal{P}} = \mathcal{P}_{1}\cdots\mathcal{P}_{n}\omega_{g,n}$ and the ``purely holomorphic part" $\varphi_{g,n}:= \mathcal{H}_1\cdots\mathcal{H}_{n}\omega_{g,n}$.
\end{definition}
Since $\omega_{0,1}$ and $\omega_{0,2}$ do not have pole at $p_i$'s, projecting them to the polar part with respect to one variable gives zero. For $2g - 2 + n > 0$, since $\omega_{g,n}$ has no pole on diagonals, the order to which the operations $\mathcal{P}_i$ or $\mathcal{H}_i$ are applied is indifferent. So, in both stable and unstable cases, $\varphi_{g,n}$ and $\omega_{g,n}^{\mathcal{P}}$ are symmetric in their $n$ variables.

\subsubsection{Normalized solutions}

We say that a solution of abstract loop equations is \emph{normalized} when, for any $(g,n)$ such that $2g - 2 + n > 0$, we have $\mathcal{P}_1\cdots\mathcal{P}_{n}\omega_{g,n} = \omega_{g,n}$. Then, the purely holomorphic part of $\omega_{g,n}$ vanishes for $2g - 2 + n > 0$. We denote $\mathcal{M}^0_{\mathcal{C}} \subseteq \mathcal{M}_{\mathcal{C}}$ the subset of normalized solutions. This is the framework of the usual topological recursion of Eynard and Orantin \cite{EOFg}. We define the local recursion kernel $K_i \in H^0(U_i \times \Sigma,K_{U_i}^{-1}\boxtimes K_{\Sigma})$ by the formula:
\beq
\label{recK} K_i(z,z_1) = \frac{1}{2}\,\frac{\int_{\sigma_i(z)}^{z} B(\cdot,z_1)}{\omega_{0,1}(z) - \omega_{0,1}(\sigma_i(z))}
\eeq

\begin{theorem}
\label{normje} $\omega^0$ is a normalized solution of abstract loop equations iff for any $2g - 2 + n > 0$ we have:
\bea
 \label{btoporec} & & \omega_{g,n}^0(z_1,I) \\
 & & = \sum_{i=1}^{s} \Res_{z \rightarrow p_i} K_i(z,z_1)
\bigg\{
\omega_{g-1,n+1}^0(z,\sigma_i(z),I) 
+ \!\!\!\sum'_{
\begin{smallmatrix}
h+h'=g\\
J\sqcup J'= I 
\end{smallmatrix}
}
\!\!\omega_{h,1+|J|}^0(z,J)\omega_{h',1 + |J'|}^0(\sigma_{i}(z),J')
\bigg\} \nonumber
\eea
Here, $I = \{z_2,\ldots,z_n\}$ is a generic $n$-uple of variables on $\Sigma$, and the symbol $\sum'$ means that we assume $(h,1 + |J|)$ and $(h,1 + |J'|) \neq (0,1)$. 
\end{theorem}
The result is proved in \cite{EOFg,BEO}, but we shall give a short and self-contained proof in \S~\ref{tuna1}. It says that normalized solutions are determined by their $(0,1)$ and $(0,2)$ parts with formula~\eqref{btoporec}. We have a surjective map:
$$
{\rm pr}^0\,:\,\mathcal{M} \rightarrow \mathcal{M}^0
$$
which associates to a solution $(\omega_{g,n})_{g,n}$ of abstract loop equations, the unique normalized solution $\omega^{0}_{g,n}$ constructed from $(\omega_{0,1}^0,\omega_{0,2}^0) = (\omega_{0,1},\omega_{0,2})$.

\subsubsection{General solutions}
\label{soliset}
One of our main result is an explicit description of $\mathcal{M}_{\mathcal{C}}$.

\begin{theorem}
\label{bllb}The projection to the ``purely holomorphic part" defines a bijection:
$$
{\rm BTR}_{{\mathcal{C}}}^{-1}\,:\,\mathcal{M}_{\mathcal{C}} \MapsTo \mathcal{V}_{\mathcal{C}} := \mathcal{O}_{\Sigma,P}^{*} \oplus H^0\big(\Sigma^{2},K_{\Sigma}^{\boxtimes\,2}(-2\Delta)\big)^{\mathfrak{S}_{2}}\big|_{1} \oplus \bigoplus_{\substack{g \geq 0,\,\,n \geq 1 \\ 2g - 2 + n > 0}} H^0(\Sigma^{n},K_{\Sigma}^{\boxtimes n})^{\mathfrak{S}_{n}}
$$
We give the name ``blobbed topological recursion" to the inverse bijection ${\rm BTR}_{\mathcal{C}}$.
\end{theorem}
By Definition~\ref{def0} and its following remark, the unstable part in the right-hand side is $(\omega_{0,1},\omega_{0,2})$, and the stable part is $(\varphi_{g,n})_{2g - 2 + n > 0}$ and is indeed symmetric in its $n$ variables. The result can be reformulated as an exact sequence:
$$
\mathcal{M}^0 \longrightarrow \mathcal{M} \longrightarrow \bigoplus_{2g - 2 + n > 0} H^0(\Sigma^n,K_{\Sigma}^{\boxtimes n})^{\mathfrak{S}_{n}} \longrightarrow 0
$$

We prove this theorem in by constructing ${\rm BTR}_{\mathcal{C}}\,:\,\varphi \mapsto \omega$ by necessary conditions, and checking in \S~\ref{converse} it indeed defines a solution of abstract loop equations. To describe this inverse map, we assume that $\omega$ is a solution of loop equations, and we decompose it:
\beq
\label{demufm}\omega_{g,n}(z_1,\ldots,z_n) = \mathcal{H}_{1}\omega_{g,n} + \mathcal{P}_{1}\omega_{g,n}
\eeq
When $\omega_{g,n}$ is a solution of abstract loop equations, we show in \S~\ref{tuna1} that the first term is given by a formula of type \eqref{btoporec}:
\begin{equation}
\mathcal{P}_{1}\omega_{g,n}(z_1,I) = \sum_{i=1}^{s} \Res_{z \rightarrow p_i} K_i(z,z_1)
\bigg\{
\omega_{g-1,n+1}(z,\sigma_i(z),I) 
+ \!\!\!\sum'_{
	\begin{smallmatrix}
	h+h'=g\\
	J\sqcup J'= I 
	\end{smallmatrix}
}
\!\!\omega_{h,1+|J|}(z,J)\omega_{h',1 + |J'|}(\sigma_{i}(z),J')
\bigg\} 
\label{P1wgn}
\end{equation}
In the case of normalized solutions, the second term $\mathcal{H}_{1}\omega_{g,n}$ vanishes, and this achieves the proof of Theorem~\ref{normje}. But in general, the second term is non-zero, and can be computed as a special case of a more general formula. We show in \S~\ref{tuna2} that, for any partition $A \sqcup B= \llbracket 1,n \rrbracket$, we have:
\beq
\label{HAPB}\mathcal{H}_{A}\mathcal{P}_{B}\omega_{g,n} := \Big(\bigotimes_{a \in A} \mathcal{H}_{a}\Big)\Big(\bigotimes_{b \in B} \mathcal{H}_{b}\Big) \omega_{g,n} = \sum_{\Gamma \in {\rm Bip}^0_{g,n}(A,B)} \frac{\varpi_{\Gamma}^0}{|{\rm Aut}\,\Gamma|}
\eeq
We underline again that, since $\omega_{g,n}$ has no pole on the diagonals for $2g - 2 + n > 0$, the order of the operations $\mathcal{H}_{a}$ and $\mathcal{P}_{b}$ does not matter: the left-hand side is well-defined and $\mathcal{H}_{A}\mathcal{P}_{B}\omega_{g,n} = \mathcal{H}_{B}\mathcal{P}_{A}\omega_{g,n}$. We define ${\rm Bip}^0(A,B)$ as the set of bipartite graphs $\Gamma$ with the properties:
\begin{itemize}
\item[$\bullet$] Vertices $v$ are either of type $\varphi$ or $\omega^0$, and carry an integer label $h(v)$ (the genus), such that the valency satisfies $2h(v) - 2 + d(v) > 0$.
\item[$\bullet$] Edges can only connect $\varphi$ to $\omega^0$ vertices.
\item[$\bullet$] There are $n$ unbounded edges (=leaves) labeled from $1$ to $n$. The leaves with label $a \in A$ must be incident to $\varphi$-vertices, and the leaves with label $b \in B$ to $\omega^0$-vertices.
\item[$\bullet$] Each $\omega^0$ vertex must be incident to at least a leaf (so, if $B =\emptyset$, then we have no $\omega^0$ vertices).
\item[$\bullet$] $\Gamma$ is connected and $b_1(\Gamma) + \sum_{v} h(v) = g$.
\end{itemize}
Since all vertices are stable, the last constraint on the topology selects only finitely many graphs. The automorphisms of $\Gamma$ are the permutation of the edges preserving the graphs and the leaf labels. The weight $\varpi_{\Gamma}^0(z_1,\ldots,z_n)$ is computed as follows. We assign variables $z_1,\ldots,z_n \in \Sigma$ to the leaves, and integration variables $z_e$ to the edges. To each vertex $v$ whose set of variables on incident edges is $Z(v)$, we assign a local weight $\omega^0_{h(v),d(v)}(Z(v))$ or $\varphi^0_{h(v),d(v)}$ depending on its type. This does not depend on any order given to the variables in $Z(v)$ since $\varphi$'s and $\omega^0$'s are symmetric in their variables. We then multiply all local weights: each edge variable $z_e$ appears twice in some combination $\omega^0(z_e,\ldots)\varphi(z_e,\ldots)$ and we integrate it out with the pairing:
\beq
\label{pairing0}\big\langle \omega(z_e,\ldots)\varphi(z_e,\ldots)\big\rangle := \sum_{i = 1}^s \Res_{z_e \rightarrow p_i} \omega(z_e,\ldots)\int^{z_e}_{p_i} \varphi(z_e',\ldots)
\eeq
Since all vertices are stable, the integrand does not contain a $\omega_{0,2}$ and thus has poles only when one of the variables goes to the divisor $P$. Therefore, the final result does not depend on the order of integration of the edge variables: we obtain a meromorphic form in $n$ variables $\varpi_{\Gamma}(z_1,\ldots,z_n)$.

The second term in \eqref{demufm} is then computed from:
$$
\mathcal{H}_{1}\omega_{g,n}(z_1,I) = \sum_{A' \sqcup B = I} \mathcal{H}_{\{1\}\sqcup A'}\mathcal{P}_{B}\omega_{g,n}(z_1,I)
$$
Since the right-hand side contains the projection on the holomorphic part for at least with respect to one variable, it is expressed in terms of $\omega_{g',n'}^0$ with $2g' - 2 + n' < 0$ and $\varphi_{g',n'}$ with $2g' - 2 + n' \leq 2g - 2 + n$. We underline that the variable $z_1$ (chosen arbitrarily) plays a special role in the decomposition \eqref{demufm}, and in particular the two terms $\mathcal{P}_{1}\omega_{g,n}$ and $\mathcal{H}_{1}\omega_{g,n}$ are in general not symmetric in their $n$ variables.

\begin{remark} Formula~\ref{HAPB} computes in particular $\omega_{g,n}^{\mathcal{P}}:= \mathcal{P}_{1}\cdots\mathcal{P}_{n}\omega_{g,n}$ as a sum over ${\rm Bip}_{g,n}^0(\emptyset,\llbracket 1,n \rrbracket$). Since these graphs may have internal $\varphi$-vertices, $\omega_{g,n}^{\mathcal{P}}$ is in general not equal to $\omega_{g,n}^0$, although we have $\mathcal{P}_{1}\cdots\mathcal{P}_{n}\omega_{g,n}^{\mathcal{P}} = \omega_{g,n}^{\mathcal{P}}$. This implies that in general, the sequence of correlators $\omega_{0,1},\omega_{0,2},(\omega_{g,n}^{\mathcal{P}})_{2g - 2 + n > 0}$ is not a solution of loop equations.
\end{remark}

\begin{figure}[h!]
\includegraphics[width=0.45\textwidth]{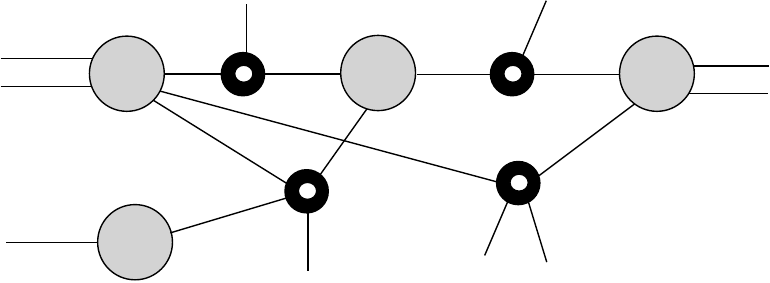}
\caption{\label{FBip0} Example of a graph in ${\rm Bip}^0_{g,n}$. The $\omega^0$ vertices appear as black-white vertices.}
\end{figure}

\subsubsection{$\omega$ from its purely polar and purely holomorphic part}

We can rewrite $\mathcal{H}_{A}\mathcal{P}_{B}\omega_{g,n}$ solely in terms of $\omega^{\mathcal{P}}$'s and $\varphi$'s.
Let ${\rm Bip}^{\mathcal{P}}_{g,n}(A,B)$ be the subset of graphs in ${\rm Bip}^0_{g,n}(A,B)$ that do not contain internal $\varphi$-vertices, and we define a new weight $\varpi_{\Gamma}^{\mathcal{P}}$ for such graphs: we rename $\omega^{\mathcal{P}}$ and $\varphi$ the type of vertices in ${\rm Bip}^{\mathcal{P}}(A,B)$, and assign local weights $\omega^{\mathcal{P}}_{h(v),d(v)}$ or $\varphi_{h(v),d(v)}$ to a vertex $v$ according to its type. The total weight $\varpi_{\Gamma}^{\mathcal{P}}$ is then computed as in \S~\ref{soliset}, by integrating out all edge variables with the pairing \eqref{pairing0}.

\begin{proposition}
\label{comuadfd}If $(\omega_{g,n})_{g,n}$ is a solution of abstract loop equations, we have for $2g - 2 + n > 0$ and partition $A \sqcup B = \llbracket 1,n \rrbracket$:
$$
\mathcal{H}_{A}\mathcal{P}_{B}\omega_{g,n} = \sum_{\Gamma \in {\rm Bip}^{\mathcal{P}}_{g,n}(A,B)} \frac{\varpi_{\Gamma}^{\mathcal{P}}}{|{\rm Aut}\,\Gamma|}
$$
\end{proposition}

This is an easy consequence of \eqref{HAPB} proved in \S~\ref{cororoo}  by a resummation of internal $\varphi$-vertices in ${\rm Bip}^0_{g,n}(A,B)$.

\subsection{Diagrammatics and examples}

\subsubsection{Normalized solutions: skeleton graphs}
\label{Skekek}
We first review the diagrammatics of the usual topological recursion, i.e. the computation of any normalized solution of abstract loop equations:
$$
\omega^0 := {\rm BTR}_{\mathcal{C}}[\omega_{0,1},\omega_{0,2},(0)_{2g - 2 + n > 0}]
$$
as a sum over skeleton graphs. Indeed, by applying repeatedly the residue formula \eqref{btoporec}, we arrive in $2g - 2 + n$ steps to an expression of $\omega_{g,n}$ involving only the recursion kernel $K_i$ and $\omega_{0,2}$ (Proposition~\ref{pmm1} below). At each step of the recursion, so there are many ways to go further, since we need to choose a leg to apply the residue formula. We will explain that, if we choose an initial leg $i_0 \in \llbracket 1,n \rrbracket$, there is a canonical way to make further choices. Though this approach breaks the symmetry, it has the advantage to restrict the number of terms.

\begin{figure}[h!]
\includegraphics[width=0.9\textwidth]{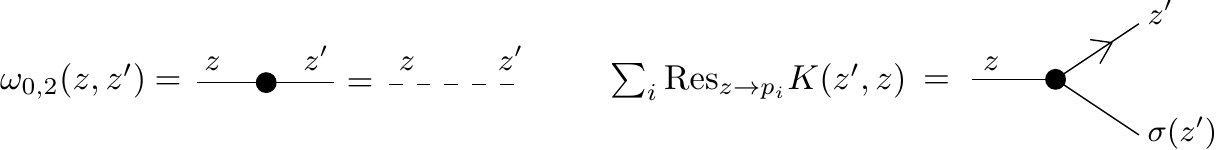}
\caption{\label{FKB} Building elements of the skeleton graphs.}
\end{figure}

Let $i_0 \in \llbracket 1,n \rrbracket$. For $2g - 2 + n > 0$, we define a set ${\rm Skel}_{g,n}(i_0)$ of graphs $G$ such that:
\begin{itemize}
\item[$\bullet$] $G$ has $n$ leaves labeled from $1$ to $n$, trivalent vertices, bivalent vertices, and its first Betti number is $g$. 
\item[$\bullet$] the trivalent vertices have a cyclic order of their incident half-edges.
\item[$\bullet$] $G$ is equipped with a spanning tree $T$, going through all trivalent vertices, and rooted at the leaf $i_0$. All the other leaves are incident to a bivalent vertex.
\end{itemize}
We shall impose an extra constraint on the graphs, but first need some vocabulary. 
\begin{definition}
Two trivalent vertices $v$ and $v'$ are \emph{parent} if the shortest path $i_0 \rightarrow v$ along $T$ contains (or is contained) in the shortest path $i_0 \rightarrow v'$ along $T$.
\end{definition}
\noindent We require the following property for our graphs:
\begin{itemize}
\item[$\bullet$] If the two edges incident to a bivalent vertex are separating, they should be incident to a leaf and a trivalent vertex. If they are non separating, they must be incident to two parent trivalent vertices.
\end{itemize}
Alternatively, one can erase the black bivalent vertices and replace their two incident edges with a single dashed edge.
\begin{definition}
There is a unique way to arrive at a trivalent vertex $v$ following a path in $T$ from the root. The edge on which we arrive to $v$ is called the source, and we say that the next (resp. previous) edge according to the cyclic order at $v$ is ``left'' (resp. ``right'').
\end{definition}
\noindent We give a few example of skeleton graphs in $\mathscr{G}_{g,n}(i_0)$, and counterexamples in Figure~\ref{FSkel-forbid}.

The weight of $G$ is obtained by: (1) assigning points $z_1,\ldots,z_n \in \Sigma$ to the leaves; (2) browsing $T$, at each trivalent vertex $v$ met, assigning an integration variable $z_{v}$ to the left edge, and $\sigma(z_{v})$ to the right edge. Now, all edges carry a variable. (3) Assigning a local weight $K(z'',z') := \sum_{i = 1}^s \mathbf{1}_{U_i}(z'')\,K_i(z'',z')$ to a vertex $v$ in $T$ whose source edge carries $z'$ and left edge carries $z''$; (4) assigning a local weight $\omega_{0,2}(z_{e_1},z_{e_2})$ to a bivalent vertex with incident edge variables $\{z_{e_1},z_{e_2}\}$; (5) multiplying all local weights; (6) for each vertex $v$ in $T$ (starting with the last one in the exploration of $T$ from the root) apply the operation $\sum_{i = 1}^s \Res_{z_{v} \rightarrow p_i}$ to the expression obtained so far. In the equivalent representation of $G$ where bivalent vertices were replaced by dashed edges, each half-dashed edge carries a variable prescribed by the assignment made at leaves and then by browsing $T$.

\begin{figure}[h!]
\includegraphics[width=0.6\textwidth]{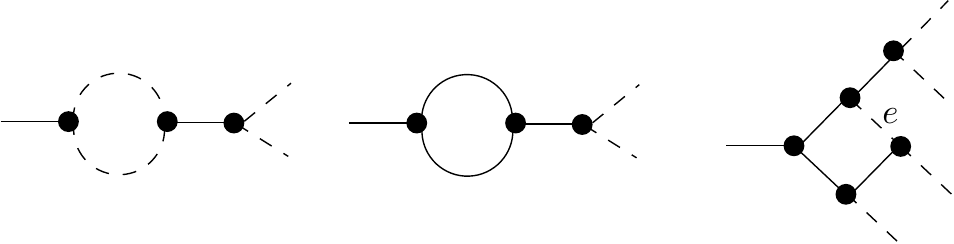}
\caption{\label{FSkel-forbid} Examples of graphs that do not belong to ${\rm Skel}_{g,n}(i_0)$. In the first graph, $K$ vertices do not form a tree. In the second graph. In the third graph, the edge $e$ connects vertices that are not parent.}
\end{figure}

\begin{proposition}[\cite{EOFg}]
\label{pmm1} Denote $\omega_{g,n}^{0}$ the normalized solution associated to $\omega_{0,1}$ and $\omega_{0,2}$ -- see Theorem~\ref{normje}. For any $2g - 2 + n > 0$ and $i_0 \in \llbracket 1,n \rrbracket$, we have:
\beq
\omega_{g,n}^{0}(z_1,\ldots,z_n) = \sum_{G\,\in\,{\rm Skel}_{g,n}(i_0)} \varpi_{G}^{{\rm Skel}}(z_1,\ldots,z_n)
\eeq
\end{proposition}
\begin{proof} If we choose the variable attached to an edge $i_0$ to apply a step of the recursion, one produces a trivalent vertex with source edge $i_0$. We declare one of the two other edges to be the left one $e_{L}$ -- this is an arbitrary choice of cyclic order -- and always decide to apply the recursion at the next step to the edge $e_{L}$. In this way, one easily proves that $\omega_{g,n}^0$ is a sum over graphs which have the properties announced.
\end{proof}

\vspace{0.2cm}

\begin{remark} Since $K(z,z_1)$ has a pole at $z_1 = z,\sigma(z)$ and $\omega_{0,2}(z_1,z_2)$ has a pole at $z_1 = z_2$, the order of taking the residues does matter. However, two graphs that differ by their cyclic ordering at the vertices have the same weight, since it merely correspond to a change of variable of integration $z \rightarrow \sigma(z)$.
\end{remark}

\begin{remark}
By consistency, the sum over skeleton graphs must be symmetric in the $n$ variables, although $i_0$ seem to play a special role. The symmetry can be checked by direct computation, see \cite{EOFg}.
\end{remark}

\begin{figure}[h!]
\includegraphics[width=\textwidth]{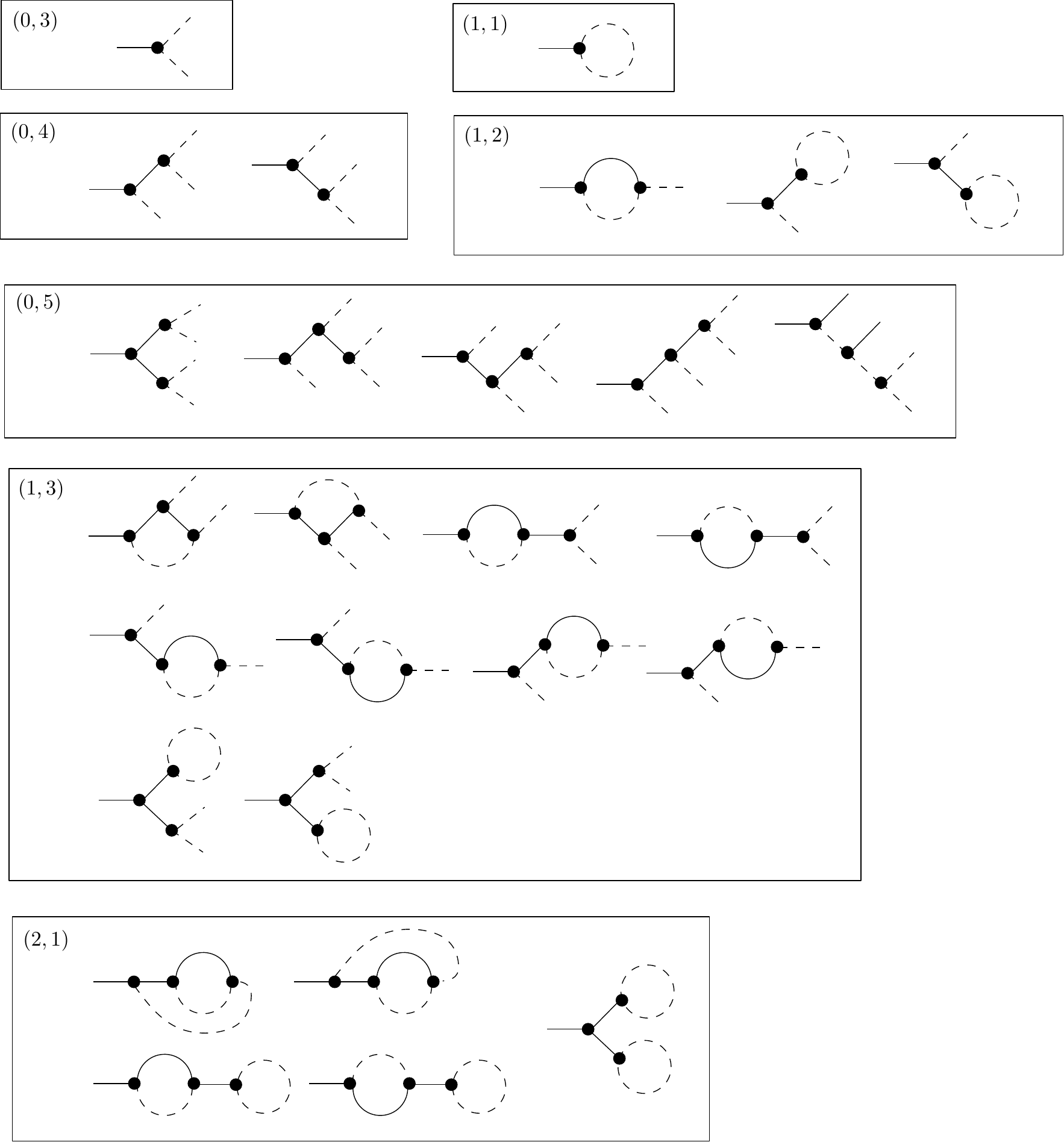}
\caption{\label{FSkel-gn} List of skeleton graphs. It remains to distribute the leaf labels $2,\ldots,n$ in all possible inequivalent ways. For instance, the graph for $(0,3)$ gives rise to 2 terms, whether we find the leaf $2$ on the right on the left.}
\end{figure}

\subsubsection{General solution: bipartite skeleton graphs}
\label{Skekek1}
To compute a general solution $(\omega_{g,n})_{g,n}$ of the loop equations, we can project with respect to each variable either to holomorphic or polar part and use formula~\eqref{HAPB}:
\begin{corollary}
For any $2g - 2 + n > 0$, we have:
$$
\omega_{g,n}(z_1,\ldots,z_n) = \sum_{\Gamma \in {\rm Bip}^{0}_{g,n}} \frac{\varpi_{\Gamma}^0(z_1,\ldots,z_n)}{|{\rm Aut}\,\Gamma|}
$$
where ${\rm Bip}^0_{g,n} = \bigsqcup_{A \sqcup B = \llbracket 1,n \rrbracket} {\rm Bip}^0_{g,n}(A,B)$, and the local weight of $\omega^0$-vertices is itself computed as a sum over skeleton graphs of Theorem~\ref{pmm1}.
\end{corollary}

\subsection{Computation of $\mathcal{P}_{1}\omega_{g,n}$ (proof of Equation~\eqref{btoporec})}
\label{tuna1}

The proof already appears in \cite[Section 1]{BEO}, but we give here a short and self-contained proof. This formula is important for it will also be used in \S~\ref{Hignproo} to compute the second term $\mathcal{H}_{1}\omega_{g,n}$.

\label{proofTBR}

If $\lambda$ is a $1$-form, we denote:
$$
\Delta\lambda(z) = \lambda(z) - \lambda(\sigma(z)),\qquad \mathcal{S}\lambda(z) = \lambda(z) + \lambda(\sigma(z))
$$
in terms of the local involution $\sigma$ near the $p_i$'s. For $\lambda,\mu$ two $1$-forms, we have:
\bea
\label{mumu1}\lambda(z)\mu(z) + \lambda(\sigma(z))\mu(\sigma(z)) & = & \frac{1}{2}\big(\Delta\lambda(z)\cdot\Delta\mu(z) + \mathcal{S}\lambda(z)\cdot\mathcal{S}\mu(z)\big) \\
\label{mumu2} \lambda(z)\mu(\sigma(z)) + \lambda(\sigma(z))\mu(z) & = & \frac{1}{2}\big(\Delta \lambda(z)\cdot\Delta \mu(z) - \mathcal{S}\lambda(z)\cdot \mathcal{S}\mu(z)\big)
\eea
If $\lambda(z,\ldots)$ depends on many variables $\Delta_{z}$ or $\mathcal{S}_{z}$ the action of these operators act on the first variable.

Let $(g,n)$ such that $2g - 2 + n > 0$. We compute:
\bea
\mathcal{P}_{1}\omega_{g,n}(z_1,I) & = & \sum_{i = 1}^s \Res_{z \rightarrow p_i} G_i(z,z_1)\,\omega_{g,n}(z,I) \nonumber \\
& = & \frac{1}{2} \sum_{i = 1}^s \Res_{z \rightarrow p_i} \big\{G_i(z,z_0)\omega_{g,n}(z,I) + G_i(\sigma_i(z),z_0)\omega_{g,n}(\sigma_i(z),I)\big\} \nonumber \\
& = & \frac{1}{4} \sum_{i = 1}^s \Res_{z \rightarrow p_i} \big\{\Delta_{1}G_i(z,z_0)\cdot\Delta_{1}\omega_{g,n}(z,I) + \mathcal{S}_{1}G_i(z,z_0)\cdot\mathcal{S}_{1}\omega_{g,n}(z,I)\big\} \nonumber
\eea
The linear loop equation tells us that $\mathcal{S}_{1}\omega_{g,n}(z,I)$ is holomorphic when $z \rightarrow P$, hence the residue of the second term vanishes. Let us rearrange the expression $Q_{g,n}$ involved in the quadratic loop equations \eqref{Qng},  by writing apart $\omega_{g,n}$ and using \eqref{mumu2}:
\beq
\label{mnun}Q_{g,n}(z,I) = \frac{1}{2}\big(-\Delta_{1}\omega_{g,n}(z,I)\cdot\Delta\omega_{1,0}(z) + \mathcal{S}_{1}\omega_{g,n}(z,I) \cdot \mathcal{S}\omega_{0,1}(z)\big) + \widetilde{Q}_{g,n}(z,I)
\eeq
The remainder is:
$$
\widetilde{Q}_{g,n}(z,I) = \omega_{g - 1,n + 1}(z,\sigma_i(z),I) + \sum_{\substack{h + h' = g \\ J \sqcup J' = I}}^{'} \omega_{h,1 + |J|}(z,J)\omega_{h',1 + |J'|}(\sigma_i(z),J')
$$
where $\sum^{'}$ means that the terms containing $\omega_{0,1}$ were excluded. Therefore:
$$
\mathcal{P}_{1}\omega_{g,n}(z_0,I) = \sum_{i = 1}^s \Res_{z \rightarrow p_i} \frac{\Delta_{1}G_i(z,z_0)}{2\,\Delta \omega_{0,1}(z)}\Big(\widetilde{Q}_{g,n}(z,I) - Q_{g,n}(z,I) + 2\,\mathcal{S}_{1}\omega_{g,n}(z,I)\cdot\mathcal{S}\omega_{0,1}(z)\Big)
$$
The assumption $y \in \mathcal{O}_{\Sigma,P}^*$ implies that $\Delta_{z}\omega_{0,1}$ has exactly a double zero when $z \rightarrow p_i$, so the prefactor $\Delta_{1}G_i(z,z_0)/\Delta\omega_{0,1}$ has exactly a simple pole at $z \rightarrow p_i$. We conclude by observing that the two last terms do not contribute to the residue since: $Q_{g,n}(z,I)$ has a double zero when $z \rightarrow p_i$ according to the quadratic loop equations ; and  $\mathcal{S}_{1}\omega_{g,n}(z,I)$ and $\mathcal{S}\omega_{0,1}(z)$ both have at least a simple zero according to the linear loop equations.

\subsection{Computation of $\mathcal{H}_{A}\mathcal{P}_{B}\omega_{g,n}$ (proof of Equation~\eqref{HAPB})}
\label{tuna2}
\label{Hignproo}

If $B = \emptyset$, there is only one graph in ${\rm Bip}_{g,n}^{0}(A,B)$, namely the $\varphi$-vertex by definition of the purely holomorphic part, so formula~\eqref{HAPB} holds. We shall prove it in general by induction on increasing values of $2g - 2 + n > 0$ and $|B|$.

We first consider $2g - 2 + n = 1$, that is $(g,n) = (0,3)$ and $(g,n) = (1,1)$. For the case $(0,3)$, we have from the residue formula~\eqref{P1wgn}:
\bea
\mathcal{P}_{3}\omega_{0,3}(z_1,z_2,z_3) & \!\!\!\!= &\!\!\! \sum_{i = 1}^s \Res_{z \rightarrow p_i} K_i(z,z_3)\big\{\omega_{0,2}(z,z_1)\omega_{0,2}(\sigma_i(z),z_2) + \omega_{0,2}(\sigma_i(z),z_1)\omega_{0,2}(z,z_2)\big\}  \nonumber \\
& \!\!\!\!= &\!\!\! \omega_{0,3}^{0}(z_1,z_2,z_3) \nonumber
\eea
The right-hand side is already of the form $\mathcal{P}_{2}\mathcal{P}_{3}\lambda(z_1,z_2,z_3)$, thus:
$$
\mathcal{H}_{1}\mathcal{H}_{2}\mathcal{P}_{3}\omega_{0,3} = 0,\qquad \mathcal{H}_{1}\mathcal{P}_{2}\mathcal{P}_{3}\omega_{0,3} = 0,\qquad \mathcal{P}_{1}\mathcal{P}_{2}\mathcal{P}_{3}\omega_{0,3} = \omega_{0,3}^0
$$
This agrees with formula \eqref{HAPB}, since ${\rm Bip}_{0,3}^0(A,B)$ is empty unless $A$ or $B$ is empty, and ${\rm Bip}_{0,3}^0(\emptyset,\{1,2,3\})$ contains only the graph made of a $\omega_{0,3}^0$ vertex.

For the case $(g,n) = (1,1)$, the residue formula \eqref{P1wgn} gives:
$$
\mathcal{P}_{1}\omega_{1,1}(z_1) = \sum_{i = 1}^s \Res_{z \rightarrow p_i} K_i(z,z_1)\,\omega_{0,2}(z,\sigma_i(z)) = \omega_{1,1}^0
$$
and this agrees again with formula \eqref{HAPB} for $A = \emptyset$ and $B = \{1\}$.

Fix $(g,n)$  with $2g - 2 + n \geq 2$ and $\overline{B}$ a non-empty subset of $\overline{N} = \llbracket 1,n \rrbracket$. Assume the formula~\eqref{HAPB} is proved for all $(g',n')$ such that $2g' - 2 + n' < 2g - 2 + n$, and for $|B'| < |\overline{B}|$. Pick an arbitrary element $b_0 \notin B$. We denote:
$$
\qquad A = \overline{N}\setminus \overline{B}, \qquad  B = \overline{B} \setminus \{b_0\},\qquad N = \overline{N}\setminus\{b_0\},\qquad N[j] = N\setminus\{j\}
$$
We first write the residue formula \eqref{P1wgn} with respect to the variable $z_{b_0}$ for $\mathcal{P}_{b_0}\omega_{g,n}$:
 \bea
& & \mathcal{P}_{b_0}\omega_{g,n}(z_{\overline{N}}) \nonumber \\
& = & \sum_{i = 1}^{s} \Res_{z \rightarrow p_i} K_i(z,z_{b_0})\Big\{\omega_{g - 1,n + 1}(z,\sigma_i(z),z_{N}) + \sum_{\substack{J \sqcup J' = N \\ h + h' = g}}^{''} \!\!\!\! \omega_{h,|J| + 1}(z,z_J)\omega_{h',|J'| + 1}(\sigma_i(z),z_{J'}) \nonumber \\
& & + \sum_{j \in N} \omega_{0,2}(z,z_j)\omega_{g,n - 1}(\sigma_i(z),z_{N[j]}) + \omega_{0,2}(\sigma_i(z),z_{j})\omega_{g,n - 1}(z,z_{N[j]}) \Big\} \nonumber
\eea
The $\sum^{''}$ means that we exclude the terms containing $\omega_{0,1}$, or $\omega_{0,2}$ -- which was written in the last line. Note that, since we assumed $2g -2 + n \geq 2$, $\omega_{g - 1,n + 1} \neq \omega_{0,2}$ and the last line does not contain products of two $\omega_{0,2}$. We would like to apply $\mathcal{H}_{A}\mathcal{P}_{B}$ to the right-hand side. We can commute these operations with the residue in $z$ when the integrand has no pole when $z_{a},z_{b}$ approaches $z$ or $\sigma_i(z)$. This is the case for all variables in the second line, but not for the variable $j$ in the last line. However, after taking the residue in $z$ the $j$-th term in the last line is already in ${\rm Im}\,\mathcal{P}_{b_0}\mathcal{P}_{j}$. Hence:
\bea
& & \mathcal{H}_{A}\mathcal{P}_{\overline{B}}\,\omega_{g,n}(z_{\overline{N}})  \nonumber \\
& & = \sum_{i = 1}^{s} \Res_{z \rightarrow p_i} K_i(z,z_{b_0})\bigg\{\mathcal{H}_{A}\mathcal{P}_{B}\omega_{g - 1,n + 1}(z,\sigma_i(z),z_{N}) \nonumber \\
 & &  + \sum_{\substack{A' \sqcup A'' = A \\ B' \sqcup B'' = B \\ h' + h'' = g}}^{''} \!\!\!\! \mathcal{H}_{A'}\mathcal{P}_{B'}\omega_{h',1 + |A'| + |B'|}(z,z_{A'\sqcup B'})\cdot\mathcal{H}_{A''}\mathcal{P}_{B''}\omega_{h'',1 + |A''| + |B''|}(\sigma_i(z),z_{A''\sqcup B''}) \nonumber \\
& &  + \sum_{j \in B} \omega_{0,2}(z,z_j)\mathcal{H}_{A}\mathcal{P}_{B[j]}\omega_{g,n - 1}(\sigma_i(z),z_{N[j]}) + \omega_{0,2}(\sigma_i(z),z_{j})\mathcal{H}_{A}\mathcal{P}_{B[j]}\omega_{g,n - 1}(z,z_{N[j]}) \bigg\} \nonumber \\
\label{HAPBint} 
\eea
Then, we project the integrand in each variable $z$ and $\sigma_i(z)$ either to the holomorphic part or to the polar part, and we can replace the integrand with sums over bipartite graphs using the induction hypothesis. We would like to push the residue in $z$ inside the weights of the graphs, and we need to discuss the type of terms that appear.

\begin{figure}[h!]
\includegraphics[width=0.95\textwidth]{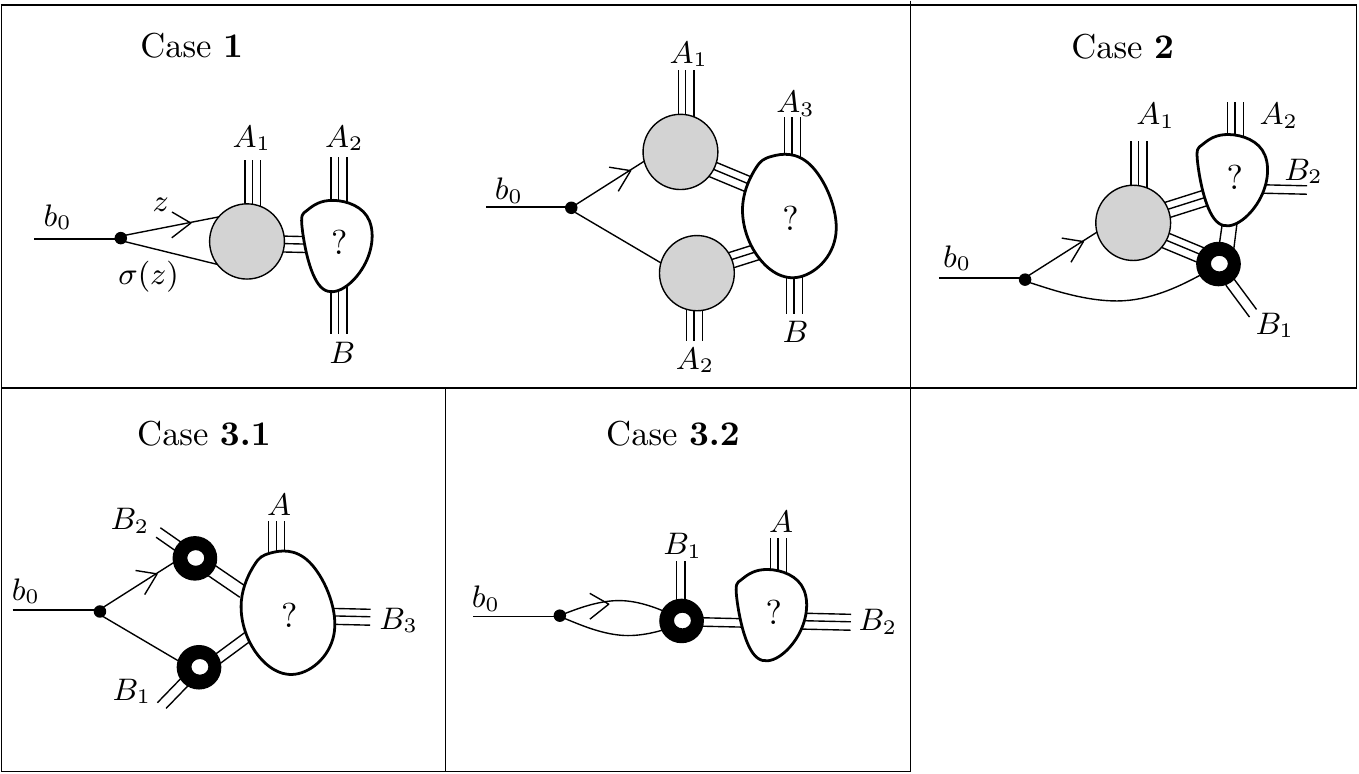}
\caption{\label{FProofPlant} Terms appearing in Equation~\eqref{HAPBint}. The areas marked by ``?'' are subgraphs that need not be connected. $(B_j)_j$ (resp., $(A_j)_{j}$) form a partition of $A$ (resp., $B$).}
\end{figure}

\vspace{0.2cm}

\textbf{1} \noindent $\bullet$ If we used $\mathcal{H}_{z}\mathcal{H}_{\sigma_i(z)}$, we get a term of the form:
$$
\sum_{i = 1}^s \Res_{z \rightarrow p_i} K_i(z,z_{b_0})\,\mathcal{R}\big[U(z,\sigma_i(z),z_{M})\big]
$$
for some $U$ which is a holomorphic form in $z$ and in $\sigma_i(z)$. $z_{M}$ represents other external variables or integration variables (appearing in the weight of a graph) that the integrand may involve. The operation $\mathcal{R}$ stands for the residue at $P$ taken on the integration variables. The only poles of the integrand in those variables occur at $P$, so $\mathcal{R}$ can be put in front of the residue on $z$. Then, it can be written:
$$
\mathcal{R}\bigg[\sum_{i = 1}^s \sum_{z \rightarrow p_i} K_i(z,z_{b_0})\,\Res_{z' \rightarrow z} \Res_{z'' \rightarrow \sigma_i(z)} \omega_{0,2}(z,z')\omega_{0,2}(\sigma_i(z),z'')\,\int^{z'}\int^{z''} U(\cdot,\cdot,z_{M})\bigg]
$$
Noting the double pole at $z = z'$ in $\omega_{0,2}$, we commute the residue in $z$ and $z'$ with the formula:
$$
\Res_{z \rightarrow p_i} \Res_{z' \rightarrow z} = \Res_{z' \rightarrow p_i} \Res_{z \rightarrow p_i} - \Res_{z \rightarrow p_i} \Res_{z \rightarrow p_i} = \sum_{i' = 1}^s \Res_{z' \rightarrow p_{i'}} \Res_{z \rightarrow p_i} - \Res_{z \rightarrow p_i} \Res_{z' \rightarrow p_{i'}}
$$
Since the integrand is regular when $z' \rightarrow p_{i'}$, only remains the first term. We commute likewise the residue in $z''$ and $z$ and recognize the expression of $\mathcal{P}_{1}\omega_{0,3} = \omega_{0,3}^0$:
\bea
& & \mathcal{R}\bigg[\sum_{i',i'' = 1}^s \Res_{z' \rightarrow p_{i'}} \Res_{z'' \rightarrow p_{i''}} \Big(\int^{z'}\int^{z''} U(\cdot,\cdot,z_{M})\Big) \sum_{i = 1}^{s} \Res_{z \rightarrow a_i} K_i(z,z_{b_0}) \omega_{0,2}(z,z') \omega_{0,2}(\sigma_i(z),z'')\bigg] \nonumber \\
& & = \mathcal{R}\bigg[\frac{1}{2} \sum_{i',i'' = 1}^s \Res_{z' \rightarrow p_{i'}} \Res_{z'' \rightarrow p_{i''}} \Big(\int^{z'}\int^{z''} U(\cdot,\cdot,z_{M})\Big) \omega_{0,3}^0(z',z'',z_{b_0})\bigg]  \nonumber
\eea
We recognize the pairing \eqref{pairing0}:
\bea
& & \mathcal{R}\bigg[\frac{1}{2}\sum_{i',i'' = 1}^s \Res_{z' \rightarrow p_{i'}} \Res_{z'' \rightarrow p_{i''}} \Big(\int^{z'}\int^{z''} U(\cdot,\cdot,z_{M})\Big)\omega_{0,3}^0(z',z'',z_{b_0})\bigg] \nonumber \\
\label{yyy1}  & & = \frac{1}{2}\,\big\langle \omega_{0,3}^0(z_{b_0},z',z'')\,U(z',z'',z_{M}) \big\rangle 
\eea
The brackets represent integration over $z',z''$, and over all other integration variables present in $z_{M}$.

\vspace{0.2cm}

\textbf{2} \noindent $\bullet$  If we use $\mathcal{H}_{z}\mathcal{P}_{\sigma_i(z)}$, we get a term of the form:
$$
\sum_{i = 1}^s \Res_{z \rightarrow p_i} K_i(z,z_{b_0})\,\mathcal{R}\Big[U(z,z_{M})\,\omega_{h,k}^0(\sigma_i(z),z_{M'})\Big]
$$
for some $U(z,z_{M})$ which is a holomorphic form in $z$, and some $(h,k) \neq (0,1)$ and $(0,2)$. $z_{M}$ and $z_{M'}$ represents other external or integration variables (appearing in the weight of a graph) that the integrand may involve. Since $(h,k) \neq (0,2)$, we can pull $\mathcal{R}$ in front of the residue in $z$. In what concerns the $U$, we can write again:
\bea
& & \mathcal{R}\bigg[\sum_{i = 1}^s \Res_{z \rightarrow p_i} K_i(z,z_{b_0})\,\Res_{z' \rightarrow z} \Big(\int^{z'} U(\cdot,z_{M})\Big)\omega_{0,2}(z,z')\,\omega_{h,k}^0(\sigma_i(z),z_{M'})\bigg]  \nonumber \\
\label{yyy2}&  & = \bigg\langle \Big[\sum_{i = 1}^s \Res_{z \rightarrow p_i} K_i(z,z_{b_0})\,\omega_{0,2}(z,z')\,\omega_{h,k}^0(\sigma_i(z),z_{M'})\Big] U(z',z_{M})\bigg\rangle 
\eea
We also encounter in the initial sum the term where the role of $z$ and $\sigma_i(z)$ is interchanged.

\vspace{0.2cm}

\textbf{3} \noindent $\bullet$ Using $\mathcal{P}_{z}\mathcal{P}_{\sigma_i(z)}$ yields a term of the form:
\bea
{\rm case}\,\,\textbf{3.1} & & \sum_{i = 1}^{s} \Res_{z \rightarrow p_i} K_i(z,z_{b_0})\,\mathcal{R}\Big[\omega_{h,k}^0(z,\sigma_i(z),z_{M})\,U(z_{M'})\Big] \nonumber \\
{\rm or}\,\,{\rm case}\,\,\textbf{3.2} & & \sum_{i = 1}^s \Res_{z \rightarrow p_i} K_i(z,z_{b_0})\,\mathcal{R}\Big[\omega_{h,k}^0(z,z_{M})\omega_{h',k'}^0(\sigma_i(z),z_{M'})\,U(z_{M''})\Big] \nonumber
\eea
for some $(h,k) \neq (0,1)$ and $(0,2)$. Since $\omega^0_{0,2}$ vertices do not appear in any graph of ${\rm Bip}^0$, we can pull $\mathcal{R}$ before the residue in $z$. The result is the pairing:
\bea
\label{yyy31}& & \bigg\langle \Big[\sum_{i = 1}^{s} \Res_{z \rightarrow p_i} K_i(z,z_{b_0})\omega_{h,k}^0(z,\sigma_i(z),z_{M})\Big]\,U(z_{M''})\bigg\rangle \\
\label{yyy32} {\rm or} & & \bigg\langle\Big[\sum_{i = 1}^s \Res_{z \rightarrow p_i} K_i(z,z_{b_0})\omega_{h,k}^0(z,z_{M})\omega_{h',k'}^0(\sigma_i(z),z_{M'})\Big]\,U(z_{M''}) \bigg\rangle 
\eea

\vspace{0.2cm}

When we replace the $\mathcal{H}\omega$'s in the right-hand side of \eqref{HAPBint} with sums over bipartite graphs, each term comes from a (maybe disconnected) bipartite graph $\Gamma$. We shall associate to each term a new bipartite graph $\overline{\Gamma} \in {\rm Bip}_{g,n}^0(A,\overline{B})$. It remains to check that all graphs of ${\rm Bip}_{g,n}^0(A,\overline{B})$ appear with a total weight that matches \eqref{HAPB}. Here is how we build $\overline{\Gamma}$, with notations taken from the previous list:
\begin{itemize}
\item[$\bullet$] In case $\mathbf{1}$, we connect the leaves $e$ and $\sigma(e)$ of $\Gamma$ carrying the variables $z$ and $\sigma(z)$, to a $\omega_{0,3}^0$-vertex, and add a leaf labeled $b_0$ to it. Note that $e$ and $\sigma(e)$ in $\Gamma$ become (internal) edges in $\overline{\Gamma}$. This graph has a symmetry factor of $1/2$ due to permutation of $z$ and $\sigma_i(z)$.
\item[$\bullet$] In case $\mathbf{2}$, denote $e'$ the leaf of $\Gamma$ carrying the variable $z'$, and decompose $M'$ into the set of leaves $M_{L}'$ and edges $M_{E}'$. We erase the $\omega_{h,k}^0$-vertex incident to $\sigma(e)$, and replace it with an $\omega_{h,k + 1}^0$-vertex with leaves labeled by $M_{L}'$ and $b_0$, and incident to edges labeled by $M_{E}'$ as well as the two edges $e'$ and $\sigma(e)$.
\item[$\bullet$] In case $\mathbf{3.1}$, decompose again $M$ into $M_{L}$ and $M_{E}$. We erase the $\omega_{h,k}^0$-vertex incident to the leaves carrying $z$ and $\sigma(z)$, and replace it with a $\omega_{h + 1,k - 1}^{0}$-vertex with leaves $M_{L}$ and $b_0$, and incident to edges $e$ and $\sigma(e)$.
\item[$\bullet$] In case $\mathbf{3.2}$, decompose $M$ into $M_{L}$ and $M_{E}$ -- resp. $M'$ in $M_{L}'$ and $M_{E}'$. We erase the $\omega_{h,k}^0$- and $\omega_{h',k'}^0$-vertices incident to $e$ and $\sigma(e)$, and replace them with a $\omega_{h + h',k + k' + 1}^0$-vertex with leaves labeled by $M_{L}\sqcup M_{L}'$ and $b_0$, and incident to edges $M_{E}$ and $M_{E}'$ as well as $e$ and $\sigma(e)$.
\end{itemize}

Now, we collect all the terms in \eqref{HAPBint} which are associated to the same graph $\overline{\Gamma} \in {\rm Bip}_{g,n}^0(A,\overline{B})$. We classify them according to the typology of the $\omega_{h_0,k_0}^0$-vertex $v$ to which $b_0$ is incident. Let us denote $E$ the set of edges (resp. $L$ the set of leaves excluding $b_0$) to which $v$ is incident, $M$ the set of all other edges and leaves in $\overline{\Gamma}$, and $U_{v}^{\overline{\Gamma}}(z_{E},z_{M})$ the product of local weights over all vertices in $\overline{\Gamma}$ except $v$. We recognize from \eqref{yyy1}-\eqref{yyy32} that the terms associated to $\overline{\Gamma}$ is:
\bea
& & \frac{1}{|{\rm Aut}\,\Gamma|}\,\bigg\langle \Big[\sum_{i = 1}^s \Res_{z \rightarrow a_i} K_i(z,z_{b_0})\Big\{\omega_{h_0 - 1,k_0 + 1}^0(z,\sigma_i(z),z_{E},z_{L}) \nonumber \\
& & + \sum_{\substack{J \sqcup J' = E \sqcup L \\ f + f' = h_0}}^{'} \omega_{f,|J| + 1}^0(z,z_{J})\omega_{f',|J'| + 1}^0(\sigma_i(z),z_{J'})\Big\}\Big]\,U(z_{E},z_{M})\bigg\rangle \nonumber \\
& & = \frac{\big\langle \mathcal{P}_{b_0}\omega_{h_0,k_0}^0(z_{b_0},z_{E},z_{L})\,U_{v}^{\overline{\Gamma}}(z_{E},z_{M}) \big\rangle}{|{\rm Aut}\,\Gamma|} = \frac{\big\langle \omega_{h,k}^0(z_{b_0},z_{E},z_{L})\,U_{v}^{\overline{\Gamma}}(z_{E},z_{M})\big\rangle}{|{\rm Aut}\,\Gamma|}  \nonumber
\eea
This expression coincides with $\varpi_{\overline{\Gamma}}$, and we have proved the formula \eqref{HAPB} for $\mathcal{H}_{A}\mathcal{P}_{\overline{B}}\omega_{g,n}$. We conclude to the general case by induction.

\subsection{End of proof of Theorem~\ref{bllb}}
\label{converse}

Now, we have to show that for arbitrary sequence of holomorphic symmetric $1$-form in $n$ variable $(\varphi_{g,n})_{2g - 2 + n > 0}$, the correlators $\omega_{g,n}$ defined by the sum of the right-hand sides of formulas~\eqref{P1wgn} and \eqref{HAPB} satisfy abstract loop equations. We first observe that, since the right-hand side of \eqref{HAPB} is holomorphic in its first variable, the right-hand side of \eqref{btoporec} can be identified with $\mathcal{P}_{1}\omega_{g,n}$. The linear loop equation is clearly equivalent to:
$$
0 = \mathcal{P}_{1}\mathcal{S}_{1}\omega_{g,n}(z_1,I)
$$
and since $\Delta\omega_{0,1}$ has exactly a double zero at $p_i$, the quadratic loop equation is equivalent to:
$$
\mathcal{P}_{1}\Big[\frac{Q_{g,n}(z_1;I)}{\Delta\omega_{0,1}(z_1)}\Big] = 0
$$
where $\mathcal{P}_{1}$ acts on the $1$-form in the variable $z_1$ to its right. The ratio $Q_{g,n}(z_1,I)/\Delta\omega_{0,1}(z_1)$ is odd in the variable $z_1$ with respect to $\sigma$ -- we can symmetrize and obtain the equivalent equation:
\beq
\label{lefnu}\sum_{i = 1}^s \Res_{z \rightarrow p_i} K_i(z,z_1)\,Q_{g,n}(z;I) = 0
\eeq
where we remind the expression of the recursion kernel \eqref{recK}.

Assume $2g - 2 + n > 0$. The right-hand side of \eqref{btoporec} gives a formula:
\bea
\mathcal{P}_{1}\omega_{g,n}(z_1,I) & = & \sum_{i = 1}^s \Res_{z \rightarrow p_i} K_i(z,z_1)\,\widetilde{Q}_{g,n}(z,I) \nonumber \\
& = & \sum_{i = 1}^s \Res_{z \rightarrow p_i} \Big(\int_{\sigma_i(z)}^{z} \!\!\!\!\!\!\omega_{0,2}(\cdot,z_1)\Big)\frac{\widetilde{Q}_{g,n}(z,I)}{2 \Delta\omega_{0,1}(z)}
\eea
Since the ratio $\widetilde{Q}_{g,n}(z,I)/\Delta\omega_{0,1}(z)$ is odd in the variable $z$ with respect to $\sigma$, we can desymmetrize:
$$
\mathcal{P}_{1}\omega_{g,n}(z_1,I) = \sum_{i = 1}^s \Res_{z \rightarrow p_i} \Big(\int_{p_i}^{z} \!\!\omega_{0,2}(\cdot,z_1)\Big)\frac{\widetilde{Q}_{g,n}(z,I)}{\Delta\omega_{0,1}(z)} = \mathcal{P}_{1}\Big[\frac{\widetilde{Q}_{g,n}(z_1,I)}{\Delta \omega_{0,1}(z_1)}\Big]
$$
The operators $\mathcal{S}_{1}$ and $\mathcal{P}_{1}$ obviously commute, thus:
$$
\mathcal{S}_{1}\mathcal{P}_{1}\omega_{g,n}(z_1,I) = \mathcal{P}_{1}\mathcal{S}_{1}\Big[\frac{\widetilde{Q}_{g,n}(z_1,I)}{\Delta \omega_{0,1}(z_1)}\Big] = 0
$$
This justifies the linear loop equation.

In the left-hand side of \eqref{lefnu} and using the expression \eqref{mnun} for $Q_{g,n}$, we recognize a piece coinciding with the definition of $\mathcal{P}_{1}\omega_{g,n}$:
\bea
& & \sum_{i = 1}^s \Res_{z \rightarrow p_i} K_i(z,z_1)\,Q_{g,n}(z;I) = -\mathcal{P}_{1}\omega_{g,n}(z_1,I) \nonumber \\
& & + \sum_{i = 1}^{s} \Res_{z \rightarrow p_i} \frac{\int_{\sigma_i(z)}^{z}\omega_{0,2}(\cdot,z_1)}{2\,\Delta\omega_{0,1}(z)}\,\frac{1}{2}\Big[\mathcal{S}_{1}\omega_{g,n}(z,I)\cdot\mathcal{S}\omega_{0,1}(z) -\Delta_{1}\omega_{g,n}(z,I)\cdot\Delta \omega_{0,1}(z) \Big] \nonumber
\eea
Thanks to the linear loop equation, the first term in the bracket has a double zero at $z \rightarrow p_i$, hence does not contribute to the residue. The last term in the right-hand side is thus equal to:
$$
\frac{1}{4}\,\sum_{i = 1}^s \Res_{z \rightarrow p_i} \Big(\int_{\sigma_i(z)}^{z} \omega_{0,2}(\cdot,z_1)\Big) \Delta_{1}\omega_{g,n}(z,I)
$$
Since $\mathcal{S}_{1}\omega_{g,n}(z,I)$ is holomorphic, it can be added to the integrand without changing the result. Then, we can desymmetrize the formula and recognize $\mathcal{P}_{1}\omega_{g,n}(z_1,I)$. This justifies the quadratic loop equation.

\subsection{Proof of Proposition~\ref{comuadfd}}
\label{cororoo}
Equation~\ref{HAPB} expresses the left-hand side as a sum over graphs  in ${\rm Bip}^0_{g,n}(A,B)$ with weight involving $\varphi$'s and $\omega^0$. Those graphs may contain internal $\varphi$-vertices, i.e. that do not have any incident leaves. In particular, ${\rm Bip}^0_{g,n}(\emptyset,\llbracket 1,n \rrbracket)$ is the set of graphs where all $\varphi$-vertices are internal and according to Formula~\ref{HAPB}
\beq
\label{WPgn}\omega^{\mathcal{P}}_{g,n} = \sum_{\Gamma \in {\rm Bip}^{0}_{g,n}(\emptyset,\llbracket 1,n \rrbracket)} \frac{\varpi_{\Gamma}^0}{|{\rm Aut}\,\Gamma|}
\eeq
If $\overline{\Gamma}$ is a graph in ${\rm Bip}^0_{g,n}(A,B)$, let us erase all its $\varphi$-vertices which are incident to at least a leaf. We obtain a number of connected components $\overline{\Gamma}_{C}$ of total genus $h_C$ and with $k_C$ new leaves (which were initially connected to some external $\varphi$-vertices) and $k_C'$ old leaves (those incident to the $\omega^0$-vertices in $\Gamma_{C}$). The $k_{C}$ leaves are not ordered, but the connected components $\Gamma_{C}$ can be grouped in clusters according to the leaf labels incident to the $\varphi$-vertices they were incident to in $\overline{\Gamma}$. This is taken into account by the automorphism factor in $\overline{\Gamma}$. Each $\overline{\Gamma}_{C}$ can be seen as a graph in ${\rm Bip}^{\mathcal{P}}_{h_{C},k_{C}}$, and all such graph appears.  Including the symmetry factor, the contribution of the connected component $C$ to the weight $\varpi_{\overline{\Gamma}}^{0}$ (before integrating out the variables carried by edges incident to external $\varphi$-vertices) is precisely $\omega^{\mathcal{P}}_{h_C,k_C + k_{C'}}(Z_C)$ given by \eqref{WPgn}. This entails Corollary~\ref{comuadfd}.

\section{Relation to intersection theory on the moduli space}
\label{S3} 

The goal of this section is to relate the formulas~\eqref{P1wgn} and~\eqref{HAPB} for general solutions of the abstract loop equations to the intersection theory of the moduli spaces of curves $\oM_{g,n}$, namely, we express these solutions in terms of the intersection indices of $\psi$-classes. In order to do this, we need a different parametrization of the space of solutions. The parameters are again symmetric holomorphic forms that we denote by $\phi_{g,n}$, but they in general differ from $\varphi_{g,n}$ of Section~\ref{S2}. We give below both the formulas for the expansions of $\omega_{g,n}$'s and $\varphi_{g,n}$'s in terms of these new parameters.

\subsection{Expansions and differential operators}

We are interested in the expansions of the differential forms $\omega_{g,n}$ near the points $p_1,\ldots,p_s$ in the local coordinates $\zeta_1,\ldots,\zeta_s$ defined by (up to an arbitrarily chosen sign) $x=\zeta_i^2/2 + a_i$, $i \in \llbracket 1,s \rrbracket$, where $a_i=x(p_i)$, $i \in \llbracket 1,s \rrbracket$. Let $\tilde{U}_i$ the domain of definition of $\zeta_i$, and $\tilde{\Sigma}$ the disjoint union of $\tilde{U}_i$. In presence of many variables, $\zeta_{i,\ell}$ denotes this coordinate for the $\ell$-th variable. We introduce the ``standard bidifferential":
\beq
\label{wo2sk} \omega_{0,2}|_{{\rm KdV}}(z_1,z_2) = \delta_{ij}\,\frac{\dd\zeta_{i,1}\otimes\dd\zeta_{j,2}}{(\zeta_{i,1} - \zeta_{j,2})^2},\quad {\rm for}\,\,(z_1,z_2) \in U_i \times U_j
\eeq
and $\mathcal{H}^{{\rm KdV}}$ the projection to the holomorphic part (Definition~\ref{sec:PolarPart}) using $\omega_{0,2}|^{{\rm KdV}}$ instead of $\omega_{0,2}$. For $2g - 2 + n > 0$, we define:
$$
\phi_{g,n}:= \mathcal{H}_{1}^{{\rm KdV}}\cdots\mathcal{H}_{n}^{{\rm KdV}}\omega_{g,n}
$$
and we take the conventions:
$$
\phi_{0,1}:= \omega_{0,1} = y\dd x,\qquad \phi_{0,2} = \omega_{0,2} - \omega_{0,2}|_{{\rm KdV}}
$$
The $\phi$'s are symmetric, holomorphic forms\footnote{Note that the ``standard bidifferential" is only defined in a local spectral curve, and cannot in general be extended to the global case discussed in \S~\ref{global}. In this respect, $\varphi$'s can be defined globally, and $\phi$'s only locally at ramification points}. in their $n$ variables in $\Sigma = \sqcup_{i = 1}^s U_i$. We may call them ``KdV-blobs". 

Let us first fix the notation for the expansions. For any $(g,n)$, we write the Taylor expansions at the ramification points:
$$
\phi_{g,n}(z_1,\ldots,z_n) \sim \phi_{g,n}^{(i_1,\ldots,i_n)} := \sum_{d_1,\ldots,d_n \geq 0} \phi_{g,n}\left[\begin{smallmatrix}
i_1 & \ldots & i_n \\
d_1 & \ldots & d_n
\end{smallmatrix}\right] 
\bigotimes_{\ell=1}^n \zeta_{i_\ell,\ell}^{d_\ell} \dd\zeta_{i_\ell,\ell}
$$
near the point $(p_{i_1},\ldots,p_{i_n})\in\Sigma^n$. 

Now we introduce differential operators that use the coefficients of these expansions. For any $(g,n)\neq(0,1)$ we denote:
\begin{equation}
\label{phihatgn} \DP_{g,n}^{(i_1,\ldots,i_n)} := \hbar^{g+n-1}\sum_{d_1,\ldots,d_n \geq 0} \phi_{g,n}\left[\begin{smallmatrix}
i_1 & \ldots & i_n \\
2d_1 & \ldots & 2d_n
\end{smallmatrix}\right] 
\prod_{\ell=1}^n (2d_\ell-1)!! \frac{\d}{\d t_{i_\ell,d_\ell}},
\end{equation}
where $t_{i,d}$ indexed by integers $i \in \llbracket 1,s \rrbracket$ and $d \geq 0$ is some set of formal variables. The case of $(g,n)=(0,1)$ is exceptional, in this case we denote:
$$
\DP_{0,1}^{(i)} :=\sum_{d \geq 2} \phi_{0,1}\left[\begin{smallmatrix}
i \\
2d
\end{smallmatrix}\right] 
(2d-1)!! \frac{\d}{\d t_{i,d}}.
$$
The two terms $d=0$ and $d=1$ are missing in this sum compared to \eqref{phihatgn}. Let us comment on that. 

The coefficients $\phi_{0,1}\left[\begin{smallmatrix}
i \\
0
\end{smallmatrix}\right] = y\dd x|_{p_i} $ are equal to zero since we assumed that $y$ is holomorphic (recall that $p_i$ is a critical point of function $x$).
The coefficients $\phi_{g,n}\left[\begin{smallmatrix}
i \\
2
\end{smallmatrix}\right]$ were excluded because they play a special role; we denote them by $\alpha_i$ indexed by $i \in \llbracket 1,s \rrbracket$ and they were assumed non-zero.  For any $(g,n)$ we denote by $\DP_{g,n}$ the sum $\sum_{1 \leq i_1,\ldots,i_n \leq s} \DP_{g,n}^{(i_1,\ldots,i_n)}$.

\begin{remark} We insist that $\DP_{g,n}$ only depends on the coefficients of the purely odd part of the $\phi_{g,n}$. In particular, $\DP_{g,n} = 0$ whenever $\phi_{g,n}$ is even in at least one of its variables. 
\end{remark}

\subsection{Explicit formula}

We consider the Witten-Kontsevich partition function of intersection numbers of $\psi$-classes on $\overline{\mathcal{M}}_{g,n}$, which is a Tau function of the KdV hierarchy \cite{Konts,Witten}:
$$
\tau_{{\rm KdV}}(\hbar,\{t_{d}\})=\exp\bigg(\sum_{g \geq 0} \hbar^{g - 1}\,F_g (\{t_{d}\})\bigg)
$$
in the variables $(t_{d})_{d \geq 0}$. For $2g - 2 + n > 0$:
 $$
 F_{g}(\{t_d\}) = \sum_{\substack{n \geq 1 \\ 2g - 2 + n > 0}} \frac{1}{n!} \sum_{d_1,\ldots,d_n \geq 0} \Big(\int_{\overline{\mathcal{M}}_{g,n}}
 \prod_{\ell = 1}^n \psi_{\ell}^{d_{\ell}}\Big)\prod_{\ell = 1}^{n} t_{d_{\ell}}
 $$
The integrals over $\overline{\mathcal{M}}_{g,n}$ in this expression are zero when $d_1 + \ldots + d_n \neq 3g - 3 + n$. It is convenient to modify this partition function to include the intersection numbers $\oM_{0,2}$ and $\oM_{0,1}$. By convention, they extend  to $n = 1$ and $2$ the formula
$$
\forall n \geq 3,\qquad \int_{\overline{\mathcal{M}}_{0,n}} \prod_{\ell = 1}^{n} \psi_{\ell}^{d_{\ell}} = {n - 3 \choose d_1,\ldots,d_n }
$$
To do so, we introduce extra variables $t_{-1},t_{-2},t_{-3},\ldots$, and define the corresponding generating function as 
\begin{align}
\label{unstablem}\mbox{for } \oM_{0,2}:\quad & \sum_{d \geq 0} \binom{-1}{d} t_{d}t_{-1-d} = \sum_{d \geq 0} (-1)^d t_{d}t_{-1-d} \\
\mbox{for } \oM_{0,1}:\quad & \binom{-2}{-2} t_{-2} = t_{-2}
\end{align}
Finally, we denote by $Z_i$ the following modification of $\tau_{{\rm KdV}}$:
\begin{equation}
Z_i:= \exp\bigg(\frac{t_{-2}}{\hbar}+\frac{1}{\hbar} \sum_{d \geq 0} (-1)^d t_{d}t_{-1-d} \bigg) \tau_{{\rm KdV}} \bigg|_{
\begin{smallmatrix}
t_d \rightarrow -t_{i,d}/\alpha_i & (d\in\mathbb{Z}) \\
\hbar  \rightarrow  \hbar/\alpha_i^2. &
\end{smallmatrix}}
\end{equation}
This defines $Z_i$ to be a series in $\hbar$ and $(t_{i,d})_{d \in \mathbb{Z}}$, and the rescaling by $\alpha_i$ implies that each monomial $\hbar^{g - 1} \prod_{\ell=1}^n t_{i,d_\ell}$ in $\ln Z_i$ comes with a prefactor $(-\alpha_i)^{2-2g-n}$.

We define a new partition function $Z$ by the following formula:
\begin{equation}
\label{eq:defZ} Z:=\exp\Big(\sum_{g \geq 0} \sum_{n \geq 1}  \DP_{g,n}\Big) \prod_{i=1}^s Z_i.
\end{equation}
The fact that it is well-defined will be explained in the next paragraph. Since the $\DP_{g,n}$ are pure differential operators, they all commute and we can apply each $e^{\DP_{g,n}}$ separately. We expand $\ln Z$ in the variables $\hbar$ and $t_{i,d}$:
\begin{equation}
\label{eq:logZ}\ln Z=\sum_{g \geq 0} \hbar^{g-1}\sum_{n \geq 1} \frac{1}{n!} \sum_{i_1,\ldots,i_n=1}^s \sum_{d_1,\ldots, d_n\in \mathbb{Z}} \langle \begin{smallmatrix}
i_1 & \ldots & i_n \\
d_1 & \ldots & d_n
\end{smallmatrix} \rangle_g \prod_{\ell=1}^n t_{i_{\ell},d_{\ell}}.
\end{equation}
We call $(g,n)$-part of $\ln Z$ the sum of terms coming with $\hbar^{g - 1}/n!$ in \eqref{eq:logZ}. 
Our main result is:
\begin{theorem}\label{thm:DP} There exists a solution of the abstract loop equations characterized by the following property: for any $(g,n)$ the odd part (in each variable) of the expansion of $\omega_{g,n}$ near $(p_{i_1},\ldots,p_{i_n})$ in the local coordinates $\zeta_{i_1},\ldots,\zeta_{i_n}$ is given by
\begin{equation}
\label{eq:ome}\omega_{g,n}^{({\rm odd})}\sim \omega_{g,n}^{(i_1,\ldots,i_n)}:=
\sum_{d_1,\ldots, d_n\in \mathbb{Z}} \langle \begin{smallmatrix}
i_1 & \ldots & i_n \\
d_1 & \ldots & d_n
\end{smallmatrix} \rangle_g 
\bigotimes_{\ell=1}^n \frac{(2d_\ell+1)!!\,\dd\zeta_{i_\ell}}{\zeta_{i_\ell}^{2d_\ell+2}}.
\end{equation}
\end{theorem}

\begin{remark}\label{rem:02conv} The statement of the theorem requires one convention. In the case $(g,n)=(0,2)$ and $i_1=i_2$, we must assume the convergence of the series in the right hand side of Equation~\eqref{eq:ome}, that is, we assume that $|\zeta_{i_1}|>|\zeta_{i_2}|$ for the sum over $d_1\geq 0$ and $d_2<0$, 
and $|\zeta_{i_1}|<|\zeta_{i_2}|$ in the opposite range $d_1<0$ and $d_2\geq 0$.
\end{remark}

\begin{remark} As we see, the odd part of the expansions of a solution of the abstract loop equations is characterized in terms of the new parameters $\phi_{g,n}$. The odd parts of the expansions of the standard parameters $\varphi_{g,n}$ can then be reconstructed using Definition~\ref{def0}. There is no unicity because the property \eqref{eq:ome} does not fixes the even part, but we show in \S~\ref{sec:even-odd-parts} that the even part decouples. The fact that Theorem~\ref{thm:DP} allows a representation is explained in its proof in \S~\ref{S37}. 
\end{remark}

This relates the computation with the odd part of the blobbed topological recursion \eqref{btoporec} to a computation whose elementary blocks are the intersection of $\psi$-classes. We show below in  \S~\ref{sec:even-odd-parts} that the even part of the $\omega_{g,n}$ is somewhat decoupled. The key substitution relation is:
\beq
\label{substim}\boxed{t_{i,d} \,\,\longleftrightarrow\,\, \frac{(2d+1)!!\,\dd\zeta_i}{\zeta_i^{2d+2}}}
\eeq
valid for any $d \in \mathbb{Z}$. In particular, for negative indices:
\beq
\label{substin}t_{i,-d - 1} \rightarrow \frac{(-1)^{d}\zeta_{i}^{2d}\,\dd\zeta_{i}}{(2d - 1)!!}, \qquad d\geq 0.
\eeq
We will prove Theorem~\ref{thm:DP} by diagrammatic techniques in \S~\ref{S36}-\ref{S37}, but we first give some illustration on how it works in some special cases that were already known and that will be used in the general proof.

\subsection{Structure of the partition function and exceptional correlators} It is not obvious that expansion given by Equation~\eqref{eq:logZ} is well-defined. Namely, from the general shape of Equation~\eqref{eq:defZ} one could expect that there is an infinite summation involved in the definition of each particular correlator in~\eqref{eq:logZ}. In this section we first prove that the expansion in~\eqref{eq:logZ} is well-defined, and then show how this formula works in two examples. 

\subsubsection{Structure of partition function $Z$} \label{sec:structureZ}

In order to prove that Equation~\eqref{eq:logZ} is well-defined we should find a way to eliminate the $(0,1)$ and $(0,2)$-terms in $Z_i$, $i \in \llbracket 1,s \rrbracket$, and the operator $\DP_{0,1}$. 

First we can observe that the action of $\DP_{0,1}$ is well-defined. Indeed, it is just a shift of the variables $t_{i,d}$ for $d\geq 2$. Since $\ln Z_i$ can be considered as a formal power series in the variables $t_{i,d}$, $d\leq 1$, whose coefficients are polynomials in $t_{i,d}$, $d\geq 2$, this shift of variables is well-defined. More precisely, the expansion of $\ln\big(\exp(\DP_{0,1})Z_i\big)$ is given by
\begin{equation} \label{eq:DP01action}
\sum_{g \geq 0} \Big(\frac{\hbar}{\alpha_i^2}\Big)^{g-1} \sum_{\substack{n \geq 1\\ k\geq 0}} 
\frac{1}{n!k!} \sum_{\substack{d_1,\ldots, d_n\in \mathbb{Z} \\ b_1,\ldots,b_k \geq 2}}
\int_{\oM_{g,n+k}} \prod_{\ell=1}^n \psi_\ell^{d_\ell} \prod_{m=1}^k \psi_{n+m}^{b_\ell} 
\prod_{\ell=1}^n \frac{t_{i,d_{\ell}}}{-\alpha_i} \prod_{\ell=1}^k 
\frac{\phi_{0,1}\left[\begin{smallmatrix}
i \\
2b_\ell
\end{smallmatrix}\right]}{-\alpha_i}
(2b_\ell-1)!!
\end{equation} 
Since $b_\ell\geq 2$, and $\sum_{\ell=1}^n d_\ell+\sum_{m=1}^k b_{m} = 3g-3+n+k$, the coefficient of each monomial $\hbar^{g-1}t_{i,d_1}\cdots t_{i,d_n}$ is a finite sum in this expansion. Furthermore, it is easy to see, for the same dimensional reason, that the $(g,n)=(0,1)$, $(0,2)$, $(0,3)$ and $(1,1)$ terms of this expansion coincide with those of $\ln Z_i$ (the argument will be revisited in Lemma~\ref{renom} below). 

Since the operators  $\DP_{g,n}$ never contain a differentiation with respect to a variable with negative index, we can commute them with the $(0,1)$-terms in $\prod_{i=1}^s\exp(\DP_{0,1})Z_i$ (which are the same as those in $\prod_{i=1}^s Z_i$):
$$
\exp\Big(\sum_{\substack{ g \geq 0, n \geq 1 \\ (g,n)\neq(0,1)} } \DP_{g,n}\Big) \prod_{i=1}^s\exp\Big(-\frac{t_{i,-2}}{\hbar \alpha_i} \Big)=\prod_{i=1}^s\exp\Big(-\frac{t_{i,-2}}{\hbar \alpha_i} \Big)
\exp\Big(\sum_{\substack{ g \geq 0, n \geq 1 \\ (g,n)\neq(0,1)} }   \DP_{g,n}\Big)
$$
For the same reason, the conjugation with the $(0,2)$-terms of $\prod_{i=1}^s\exp(\DP_{0,1})Z_i$ (or, equivalently, $\prod_{i=1}^s Z_i$) gives an operator that is well-defined as a formal power series in $t_{i,-d-1}$ indexed by $i \in \llbracket 1,s \rrbracket$ and $d\geq 0$:
$$
\exp\Big(-\sum_{i=1}^s\sum_{d \geq 0} (-1)^d t_{i,d}t_{i,-1-d}\Big)
\exp\Big(\sum_{\substack{ g \geq 0, n \geq 1 \\ (g,n)\neq(0,1)} } \DP_{g,n}\Big)
\exp\Big(\sum_{i=1}^s\sum_{d \geq 0} (-1)^d t_{i,d}t_{i,-1-d}\Big)
$$
So, we can consider Equation~\eqref{eq:defZ} as an action of this operator on a product of usual KdV tau-functions with rescaled and shifted (for $d\geq 2$) variables, which is well-defined. 

\subsubsection{Case $(g,n)=(1,0)$}

Let us discuss a simple example in order to see how we get expansion of $\omega_{0,1}$ via Theorem~\ref{thm:DP}. Indeed, from the definition we have:
\begin{equation}\label{eq:omega01}
\omega_{0,1}^{{\rm (odd)}} \sim \alpha_i \zeta_i^2\dd\zeta_i +  
\sum_{d \geq 2} \phi_{0,1}\left[\begin{smallmatrix}
i \\
2d
\end{smallmatrix}\right] 
\zeta_i^{2d} \dd\zeta_i
\end{equation}

On the other hand, the $(g,n)=(0,1)$-part of $\ln Z$ is given by the sum over $i$ of the following expression:
\begin{align}
& (-\alpha_i)^{1} t_{i,-2}+
\sum_{d \geq 2} \phi_{0,1}\left[\begin{smallmatrix}
i \\
2d
\end{smallmatrix}\right] 
(2d-1)!! \frac{\d}{\d t_{i,d}} \left(
(-\alpha_i)^{0} \sum_{d \geq 0} (-1)^d t_{i,d}t_{i,-1-d} 
\right)
\nonumber \\ \notag
& =
-\alpha_i t_{i,-2}+ 
\sum_{d \geq 2} \phi_{0,1}\left[\begin{smallmatrix}
i \\
2d
\end{smallmatrix}\right] 
(2d-1)!!  (-1)^d t_{i,-1-d}  \nonumber
\end{align}
After the changes \eqref{substim}-\eqref{substin}, we get exactly the right hand side of Equation~\eqref{eq:omega01}.

\subsubsection{Case $(g,n) = (0,2)$}

In this case, from the definition we have the following expansion in $\tilde{U}_i\times \tilde{U}_j$:

\begin{equation} \label{eq:omega02}
\omega_{0,2}^{{\rm (odd)}} \sim \delta_{ij} \sum_{d=0}^\infty \frac{\zeta_{i,1}^{2d}}{\zeta_{j,2}^{2d+2}} (2d+1) \dd\zeta_{i,1}\otimes\dd\zeta_{j,2} + \sum_{d_1,d_2 \geq 0} 
\phi_{0,2}\left[\begin{smallmatrix}
i & j \\
2d_1 & 2d_2
\end{smallmatrix}\right] \zeta_{i,1}^{2d_1} \zeta_{j,2}^{2d_2} \dd\zeta_{i,1}\otimes\dd\zeta_{j,2}
\end{equation}

The $(0,2)$-part of $\ln Z$ is given by the sum over $i$ and $j$ of the following expression
\begin{align}
&
\delta_{ij}\sum_{d \geq 0} (-1)^{d}\,t_{i,d}t_{j,-1 - d} 
\nonumber \\ \notag &
+\sum_{d_1,d_2 \geq 0} \phi_{0,2}\left[\begin{smallmatrix}
i & j \\
2d_1 & 2d_2
\end{smallmatrix}\right] 
(2d_1-1)!! \frac{\d}{\d t_{i,d_1}} \left(
(-\alpha_i)^{0} \sum_{d \geq 0} (-1)^d t_{i,d}t_{i,-1-d} 
\right)
\\ \notag &
\phantom{sssssssssssssssssss}\times (2d_2-1)!! \frac{\d}{\d t_{j,d_2}} \left(
(-\alpha_j)^{0} \sum_{d \geq 0} (-1)^d t_{j,d}t_{j,-1-d} 
\right)
\\ \notag &
=\delta_{ij}\sum_{d \geq 0} (-1)^{d}\,t_{i,d}t_{j,-1 - d} 
\\ \notag &
+\sum_{d_1,d_2 \geq 0} \phi_{0,2}\left[\begin{smallmatrix}
i & j \\
2d_1 & 2d_2
\end{smallmatrix}\right] (2d_1-1)!!(2d_2-1)!! (-1)^{d_1+d_2} t_{i,-1-d_1} 
t_{j,-1-d_2} 
\end{align}
If we take into account Remark~\ref{rem:02conv}, then after the changes \eqref{substim}-\eqref{substin}, we get exactly the right hand side of Equation~\eqref{eq:omega02}.

\subsection{Special cases of Theorem~\ref{thm:DP}} In this paragraph we collect some known special cases of Theorem~\ref{thm:DP}. The one discussed in \S~\ref{sec:nontrivial01} is needed later in the general proof of this theorem.

\subsubsection{Case of $\DP_{g,n} = 0$ for all $g,n$}

On the topological recursion side, the odd part of the $(0,1)$ correlator near the point $p_i$ is
$$
\omega_{0,1}^{({\rm odd})}\sim \alpha_i\zeta_i^2\dd\zeta_i
$$
Furthermore, we have:
$$
\omega_{0,2}(\zeta_{i,1},\zeta_{i,2}) = \frac{\dd\zeta_{i,1}\otimes\dd\zeta_{i,2}}{(\zeta_{i,1} - \zeta_{i,2})^2} =
\omega_{0,2}^{{\rm even}}(\zeta_{i,1},\zeta_{i,2}) + \omega_{0,2}^{{\rm odd}}(\zeta_{i,1},\zeta_{i,2}),
$$
where
\bea
\omega_{0,2}^{{\rm even}}(\zeta_{i,1},\zeta_{i,2}) & = & \frac{2\zeta_{i,1}\zeta_{i,2}}{(\zeta_{i,1}^2 - \zeta_{i,2}^2)^2}\,\dd\zeta_{i,1}\otimes\dd\zeta_{i,2}  \\
\label{omega02kdv}\omega_{0,2}^{{\rm odd}}(\zeta_{i,1},\zeta_{i,2}) & = & \frac{\zeta_{i,1}^2 + \zeta_{i,2}^2}{(\zeta_{i,1}^2 - \zeta_{i,2}^2)^2}\,\dd\zeta_{i,1}\otimes\dd\zeta_{i,2}
\eea
Eventually, for $2g - 2 + n > 0$, it is proved in \cite{Ekappa} (see also the book \cite{Ebook}) that $\omega_{g,n}$ is odd in each variable, and:
\begin{equation}\label{eq:omegalocal}
\omega_{g,n}^{(i,\ldots,i)} \sim (-\alpha_i)^{2-2g-n} \sum_{d_1,\ldots,d_n \geq 0} \int_{\oM_{g,n}} \prod_{\ell=1}^n \psi_\ell^{d_\ell} \bigotimes_{\ell=1}^n \frac{(2d_\ell+1)!!\,\dd\zeta_i}{\zeta_i^{2d_\ell+2}}\end{equation} 
near the point $(p_i,\ldots,p_i)$, and the expansions near the points $(p_{i_1},\ldots,p_{i_n})$ are equal to zero if $i_k\neq i_\ell$ for some $1\leq k < \ell\leq n$.

On the other hand, $Z=\prod_{i=1}^s Z_i$, and by comparison with the definition of the Witten-Kontsevich partition function:
$$
\langle \begin{smallmatrix}
i_1 & \ldots & i_n \\
d_1 & \ldots & d_n
\end{smallmatrix} \rangle_g = \delta_{i_1,\ldots,i_n} (-\alpha_i)^{2 - 2g - n}\int_{\oM_{g,n}} \prod_{\ell=1}^n \psi_\ell^{d_\ell},  $$
We apply the transformation \eqref{substim} of monomials in $\ln Z$: 
\begin{align}
& \frac{ (-\alpha_i)^{2-2g-n}  \int_{\oM_{g,n}} \prod_{\ell=1}^n \psi_\ell^{d_\ell} \prod_{\ell=1}^n t_{i,d_\ell}}{|\Aut(d_1,\ldots,d_n)|} \nonumber
\\ 
\label{wgnkdv} &
\rightsquigarrow
(-\alpha_i)^{2-2g-n} \sum_{d_1,\ldots,d_n \geq 0} \int_{\oM_{g,n}} \prod_{\ell=1}^n \psi_\ell^{d_\ell} \bigotimes_{\ell=1}^n (2d_\ell+1)!! \frac{\dd\zeta_i}{\zeta_{i,\ell}^{2d_\ell+2}},
\end{align}
and this gives, for $2g - 2 + n > 0$ and each $i \in \llbracket 1,s \rrbracket$, the same expansions as Equation~\eqref{eq:omegalocal}. For $(g,n) = (0,1)$, there is only one term. For $(g,n) = (0,2)$, we have the substitution \eqref{substin}:
\beq
\sum_{d \geq 0} (-1)^d\,t_{i,d}t_{i,-d - 1} \rightsquigarrow \sum_{d \geq 0} (2d + 1)\,\frac{\zeta_{i,1}^{2d}}{\zeta_{i,2}^{2d + 2}}\,\dd\zeta_{i,1}\otimes\dd\zeta_{i,2} = \frac{\zeta_{i,1}^2 + \zeta_{i,2}^2}{(\zeta_{i,1}^2 - \zeta_{i,2}^2)^2}\,\dd\zeta_{i,1}\otimes\dd\zeta_{i,2}
\eeq
and we recognize \eqref{omega02kdv} (remind Remark~\eqref{rem:02conv} on convergence of this series).

\begin{notation}
We define $\omega_{g,n}|_{{\rm KdV}}$ to be the right-hand side of \eqref{wgnkdv} for $2g - 2 + n > 0$, \eqref{wo2sk} for $(g,n) = (0,2)$, and $\omega_{0,1}|_{{\rm KdV}}(z) = \alpha_i\zeta_i^2\dd\zeta_{i}$ when $z \in U_i$. They encode in a simple way the intersection of $\psi$-classes. For instance:
\beq
\label{eq:03KdV} \omega_{0,3}|_{{\rm KdV}} = \frac{1}{-\alpha_i}\,\frac{\dd\zeta_{i,1}\otimes\dd\zeta_{i,2}\otimes\dd\zeta_{i,3}}{\zeta^2_{i,1}\zeta_{i,2}^2\zeta_{i,3}^2}
\eeq
\end{notation}

\begin{remark}
For $2g - 2 + n > 0$, only the non-negative $d_i$ contribute in \eqref{wgnkdv}. Since the top dimension of $\overline{\mathcal{M}}_{g,n}$ is $3g - 3 + n$, $\omega_{g,n}|_{{\rm KdV}}$ has a pole of total degree $6g - 6 + 4n$, and actually is homogeneous of degree $-3(2g - 2 + n)$ in its $n$ variables.
\end{remark}

\subsubsection{Non-trivial $\DP_{0,1}$} \label{sec:nontrivial01} This case falls in the scope of the usual topological recursion of \cite{EORev} with an arbitrary local expansion of the function $y$ near the points $p_i$, $i \in \llbracket 1,s \rrbracket$, but $\omega_{0,2}$ still being the standard bidifferential. In this case the local expansion of $\omega_{g,n}$, $2g-2+n>0$, is given by the following formula:
\begin{align}\label{eq:nontrivial01}
& \omega_{g,n}^{(i,\ldots,i)} \sim \sum_{m \geq 0} \frac{ (-\alpha_i)^{2-2g-n-m}}{m!} \sum_{a_1,\ldots,a_m \geq 2}
\prod_{j=1}^m \left(\phi_{0,1}\left[\begin{smallmatrix}
i \\
2a_j
\end{smallmatrix}\right] (2a_j-1)!!\right) \cdot
\\ \notag
& \sum_{d_1,\ldots,d_n \geq 0} \int_{\oM_{g,n+m}} \prod_{\ell=1}^n \psi_\ell^{d_\ell}
 \prod_{j=1}^m \psi_{n+j}^{a_j}  \bigotimes_{\ell=1}^n (2d_\ell+1)!! \frac{\dd\zeta_{i,\ell}}{\zeta_{i,\ell}^{2d_\ell+2}}
\end{align} 
near the point $(p_i,\ldots,p_i)$. The expansions near the points \mbox{$(p_{i_1},\ldots,p_{i_n})$} are equal to zero if $i_k\neq i_\ell$ for some $1\leq k,\ell\leq n$. This is just a properly renormalized formula in~\cite{Einter}.

It is obvious (see Equation~\eqref{eq:DP01action}) that the intersection number $$\int_{\oM_{g,n+m}} \prod_{\ell=1}^n \psi_\ell^{d_\ell}
 \prod_{j=1}^m \psi_{n+j}^{a_j}$$ should come from the coefficient of the monomial $\prod_{\ell=1}^n t_{i,d_\ell}$ in the corresponding summand\footnote{Since  $\DP_{0,1}$ is an order $1$ operator, it does not matter whether we apply it to $Z$ and then take the logarithm, or if we apply it directly to $\ln Z$.} of $\big(\DP_{0,1}^{(i)}\big)^m\ln Z$. Indeed, by definition the operator $\DP_{0,1}^{(i)}$ replaces the variable $t_{i,a}$, $a\geq 2$, that controls a factor of $\psi^a$ in the intersection class with the scalar coefficient $\phi_{0,1}\left[\begin{smallmatrix} i \\ 2a \end{smallmatrix}\right] (2a-1)!!$. So we have Theorem~\ref{thm:DP}. This formula can be nicely rewritten in terms of $\kappa$-classes \cite{Einter}, as we recall in Appendix~\ref{kappaapp}.

\begin{notation}
\label{nonon}We denote the $n$-form in the right-hand side of \eqref{eq:nontrivial01} by $\omega_{g,n}|_{{\rm KdV}}^{\Box}$. It implicitly depends on the coefficients of $\DP_{0,1}$. 
\end{notation}

\subsubsection{Non-trivial $\DP_{0,1}$ and $\DP_{0,2}$} This is the general setting of the usual topological recursion with $\omega_{0,2}$ a priori different from the standard bidifferential $\omega_{0,2}|_{{\rm KdV}}$. It creates $n$-forms $\omega_{g,n}$ which can be non-zero even when the $n$ variables belong to different open sets. This is the simplest case of non-trivial coupling between the ramification points. It is obtained from the KdV by action with the exponential of a second order differential operator, see~\cite{DBOSS,Givental} for the special cases of that.  Our Theorem~\ref{thm:DP} is then equivalent to the results of~\cite{Einter}, namely~\cite[Theorem 3.1]{Einter} for one ramification point and \cite[Theorem 4.1]{Einter} in the general case. This actually allowed a proof \cite{EOBKMP} of the BKMP conjecture \cite{BKMP}, stating that open Gromov-Witten invariants of toric Calabi-Yau $3$-folds are computed by the topological recursion with initial data coming from their mirror curve.

\subsection{Diagrammatic representation of coefficients of $Z$ in terms of KdV}\label{sec:graphs1}
\label{S36}
There is a way to represent the coefficients $\langle \begin{smallmatrix}
i_1 & \ldots & i_n \\
d_1 & \ldots & d_n
\end{smallmatrix} \rangle_g$ in Equation~\eqref{eq:logZ} (or, equivalently, the $n$-forms $\omega_{g,n}^{(i_1,\ldots,i_n)}$) graphically. Indeed, every term can be represented as a sum of contributions of connected bipartite graphs $\Gamma$ with the following structure:

\begin{itemize} 
\item
There are two types of vertices: KdV-vertices and $\Phi$-vertices. 
\item 
Each vertex $v$ is labeled by a non-negative integer (called genus) $h(v)\geq 0$, and we require $\sum_{v} h(v) + b_1(\Gamma)=g$.
\item
The valency of each vertex is at least $1$.
\item
Each (internal) edge connects a KdV-vertex and a $\Phi$-vertex.
\item
There are exactly $n$ leaves (= unbounded edges), and they are labeled by the numbers $1,\ldots,n$.
\item
Each leaf is connected to a KdV-vertex of genus $0$ and of valency $2$, with one exception: if $(g,n) = (1,0)$, $\Gamma$ has just one vertex, which is the KdV vertex of genus $0$ and valency $1$, and the leaf is connected to it. 
\end{itemize}

\begin{notation}
We denote $\mathcal{G}_{g,n}$ this set of connected bipartite graphs.
\end{notation}

We now describe the weight $\omega_{g,n}^{\Gamma}$ assigned to such a graph $\Gamma$. We first assign variables to any edge of color $i$: integration variables $z_{e}$ for an internal edge $e$, and external variables $z_{k}$ for the leaf labeled by $k \in \llbracket 1,n\rrbracket$. We denote by $E=E(\Gamma)$ the set of the internal edges, and we denote by $E_0\subseteq E$ the subset of internal edges attached to the same KdV vertices as the leaves. By our last rule, provided $\Gamma$ is neither the $(0,1)$-KdV vertex with one leaf or the $(0,2)$-KdV vertex with two leaves, the set $E_0=\{e_1,\ldots,e_n\}$ is in one-to-one correspondence with the set of leaves. 
The weight of $\Gamma$ is a symmetric form in $n$ variables defined by
\begin{equation}
\label{resGamma}
\omega_{g,n}^{\Gamma}:=\frac{1}{|\Aut(\Gamma)|} \prod_{e_i\in E_0} \Res_{z_{e_i}\rightarrow z_i} \prod_{e\in E\setminus E_0} \Bigg[\sum_{i_e}\Res_{z_{e} \rightarrow p_{i_e}}\Bigg] \Big(\bigotimes_v D[v] \Big).
\end{equation}
$D[v]$ are local weights associated to a vertex $v$. Let us denote $d(v)$ its valency, $h(v)$ its genus, and $\{e(1),\ldots,e(d(v))\}$ its set of incident edges (internal edge or leaf) \begin{itemize} 
\item If $v$ is a KdV vertex, $D[v]$ is the symmetric meromorphic form in $d(v)$ variables equal to $\omega_{h(v),d(v)}|_{{\rm KdV}}(z_{e(1)},\ldots,z_{e(d(v))})$.
\item If $v$ is a $\Phi$-vertex, $D[v]$ is the symmetric holomorphic function of $n$ variables equal to $\Phi_{h(v),d(v)}(z_{e(1)},\ldots,z_{e(d(v))})$, where:
\begin{equation}
\label{dVphi}\Phi_{h,d}(z_1,\ldots,z_d) := \begin{cases} \int^{z_1}\cdots\int^{z_d} \phi_{h,d} & \mbox{if }(h,d)\neq(0,1) ; \\
\int^{z_1} \left( \phi_{0,1}- \alpha_{i_1} \zeta^2_{i,1} \dd \zeta_{i,1}\right)& \mbox{if }(g,n)=(0,1). 
\end{cases}
\end{equation}
Here we assume that the variable $z_{j}$ belongs to the open set $U_{i_j}$, and the corresponding integration starts from the ramification point $p_{i_j}$.
\end{itemize}
Since in \eqref{resGamma} we are computing residues of products of holomorphic functions by meromorphic $1$-forms, $\omega_{g,n}^{\Gamma}$ in general does not vanish. $|{\rm Aut}(\Gamma)|$ is the number of permutations of edges and vertices labels fixing the label of each leaf and preserving $\Gamma$.

\begin{theorem}\label{thm:omega-graphs} The formula:
\begin{equation}
\label{Gngi}\omega_{g,n} = \sum_{\Gamma\in \mathcal{G}_{g,n}} \omega^{\Gamma}_{g,n}.
\end{equation}
defines a solution of abstract loop equations.
\end{theorem}

We prove this theorem in \S~\ref{S37}. Using this theorem, we can prove Theorem~\ref{thm:DP}. Indeed, it is a direct consequence of the following simple Lemma:

\begin{lemma} \label{lem:connection2thm}
	In $U_{i_1}\times\cdots\times U_{i_n}$, we have the following expansion:
	$$
		\bigg(\sum_{\Gamma\in \mathcal{G}_{g,n}} \omega^{\Gamma}_{g,n}\bigg)^{\mathrm{(odd)}} \sim \omega_{g,n}^{(i_1,\ldots,i_n)}.
	$$
	where $\omega_{g,n}^{(i_1,\ldots,i_n)}$ was by definition the right-hand side of \eqref{eq:ome}.
\end{lemma}

\begin{proof} In fact, this Lemma is almost obvious. Indeed, consider an expansion of an expression of the same type as~\eqref{eq:defZ}, that is, an action of a differential operator with constant coefficients on an exponential formal power series. Then the results can always be presented as a sum of bipartite graphs, where the two types of vertices represent the coefficients of the original formal power series and the coefficients of the differential operator; leaves correspond to the variables in which we expand the result, and the edges correspond to particular differentiations in the operator. 

In our case, we re-arrange the result of this computation into a differential form via the substitution $t_{i,d} \leftrightarrow (2d + 1)!!\zeta_i^{-(2d + 2)}\dd\zeta_i$ described in \eqref{substim} and used above in the definition of the expansion of $\omega_{g,n}$ at $(p_{i_1},\ldots,p_{i_n})$. Then the only subtle thing is that we want to represent differentiations in terms of the residues. Then it is enough to observe that the identity for $d \in \mathbb{Z}$:
$$
(2d-1)!!\frac{\partial}{\partial t_{i,d}} t_{i',d'} =\delta_{i,i'}\delta_{d,d'} (2d-1)!!
$$
is reproduced by the residue pairing:
$$
\Res_{\zeta_i\to 0}\Bigg[\Big(\int^{\zeta_i}_{0} \zeta^{2d}\Big) \cdot \delta_{i,i'} \cdot (2d'+1)!! \frac{\dd\zeta_{i'}}{\zeta_{i'}^{2d'+2}}\Bigg]
= \delta_{i,i'}\delta_{d,d'} (2d-1)!!, 
$$
while the last equality is precisely the computation on the internal edges in the residue formula. 
\end{proof}

\begin{remark}\label{rem:graph01} Using this diagrammatic formalism, we can now represent Equation~\eqref{eq:nontrivial01} for $\omega_{g,n}|_{{\rm KdV}}^{\Box}$ in terms of graphs. In this case the only non-trivial $\Phi$-vertices $v$ have $h(v)=0$ and $d(v)=1$ (see Figure~\ref{omphi5}).
\end{remark}
\begin{figure}[h!]
\begin{center}
\includegraphics[width=0.85\textwidth]{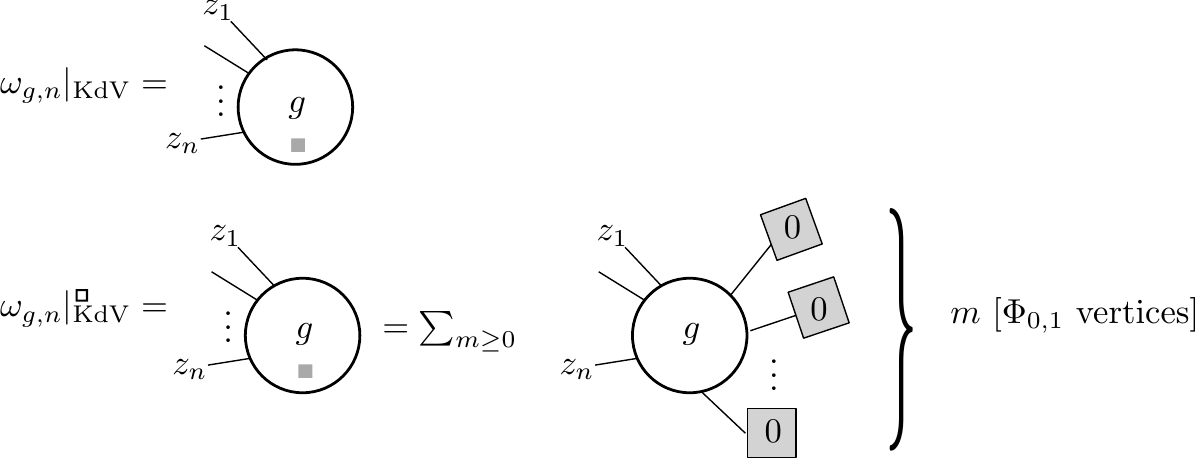}
\caption{\label{omphi5} Relation between the KdV vertices and the disk-renormalized KdV vertices. The graphs in the right-hand side have a symmetry factor $1/m!$.}
\end{center}
\end{figure}

\subsection{Graphical representations of topological recursion}
\label{S37}
In this Section we compare the sum of the graphs in Equation~\eqref{HAPB} with the one presented in \S~\ref{S36}. Since we are able to identify these two expressions, this way we will prove Theorem~\ref{thm:omega-graphs}, and, as a corollary (cf. Lemma~\ref{lem:connection2thm})  Theorem~\ref{thm:DP}. This way we also explain an explicit relation between the parameters $\phi_{g,n}$ and $\varphi_{g,n}$. 

Consider Equation~\eqref{Gngi} for $2g-2+n>0$. There are several options for the leave number $i$, $1\leq i\leq n$. It can either be attached to a stable KdV vertex, or it can be attached to a $(0,2)$-KdV vertex that is further attached to a $(0,2)$-$\Phi$ vertex, or it can be attached to a $(0,2)$-KdV vertex that is further attached to a stable $\Phi$-vertex. In the first two cases we say that the $i$-th leaf of $\mathcal{P}$-type, and in the third case we call it $\mathcal{H}$-type. 

Consider $\omega_{g,n}$ given by the sum of graphs~\eqref{Gngi}. A direct consequence of the definition of the projections $\mathcal{P}$ and $\mathcal{H}$ in \S~\ref{sec:PolarPart} is the following lemma.

\begin{lemma} \label{limim}If $\omega_{g,n}$ is given by the sum over graphs~\eqref{Gngi},
	then for any partition $A \sqcup B = \llbracket 1,n \rrbracket$,  the form $\mathcal{P}_A\mathcal{H}_B\omega_{g,n}$ is given by the sum over the subset of graphs in $\mathcal{G}_{g,n}$, where all leaves with labels in $A$ are of $\mathcal{P}$-type and all leaves with labels in $B$ are of $\mathcal{H}$-type. 
\end{lemma}

Now, we have a formula for $\varphi_{g,n}$. It is given by the restriction of the sum of graphs~\eqref{Gngi} to those graphs, where all leaves are of the $\mathcal{H}$-type. We remove the $(0,2)$-KdV vertices on the leaves using the following formula for the expansions of $\varphi_{g,n}$ (see Figure~\ref{omphi6}):
\begin{equation} 
\label{phiexi}\varphi_{g,n}^{(i_1,\ldots,i_n)}(\zeta_{i_1},\ldots,\zeta_{i_n})= \Res_{\eta_{1}\to \zeta_{i_1}} \cdots \Res_{\eta_{n}\to \zeta_{i_n}}
\Big[\prod_{\ell=1}^n \frac{\dd\zeta_{i_\ell}\otimes\dd\eta_{\ell}}{(\zeta_{i_\ell}-\eta_{\ell})^2} \Big]\,\int^{\eta_{1}}_{p_{i_1}}\cdots\int^{\eta_{n}}_{p_{i_n}} \varphi_{g,n}
\end{equation}

\begin{figure}[h!]
	\begin{center}
		\includegraphics[width=0.45\textwidth]{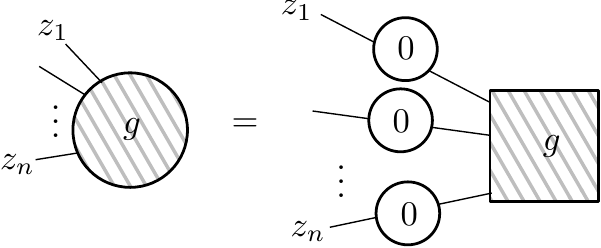}
		\caption{\label{omphi6} The hatched circle represents any symmetric holomorphic form in $n$ variables, and the hatched square the holomorphic function of $n$ variables, so that formula \eqref{phiexi} hold. We remind that the bivalent genus $0$ KdV vertices represent the standard bidifferential $\dd \zeta\otimes\dd \zeta'/(\zeta - \zeta')^2$. Edge correspond to the residue pairing between a form and a function, with the special rule~\eqref{resGamma} for the point at which the residue is taken.} 
	\end{center}
\end{figure}

Now, we can arrange the vertices of any graph in $\mathcal{G}_{g,n}$, $2g-2+n>0$, into clusters. Any cluster is a connected subgraph of positive Euler characteristic. Namely, we consider the maximal connected subgraphs that consist of the KdV-vertices (stable or unstable), unstable $\Phi$-vertices, and have at least one leaf of $\mathcal{P}$-type. These clusters are contributions to the graphical formulas of Eynard in~\cite{Einter} for the solutions of the usual topological recursion, so they form the black-white vertices in the terminology of Section~\ref{S2}. The connected components of the complement of these clusters are exactly of the type we use in the formula for $\varphi_{g,n}$ of Lemma~\ref{limim} (with the same remark that the $(0,2)$-KdV vertices on the leaves should be removed). These graphs from the $(g,n)$ $\varphi$-vertices in the terminology of Section~\ref{S2}, and indeed they can be internal, or be incident to $\mathcal{H}$-leaves. 

So, this way we represent black-white and gray vertices of Section~\ref{S2} as the sum over special subgraphs of graphs in $\mathcal{G}_{g,n}$, and this way we associate with any graph in $\mathrm{Bip}_{g,n}^0$ a subset of graphs in $\mathcal{G}_{g,n}$. The disjoint union of these subsets forms the whole set $\mathcal{G}_{g,n}$, and the definition of the weights of graphs imply that the weight of a graph in $\mathrm{Bip}_{g,n}^0$ given by Equation~\eqref{pairing0} is equal to the sum over the corresponding subset of graphs in $\mathcal{G}_{g,n}$ of the weights given by Equation~\eqref{resGamma}. 
This proves Theorem~\ref{thm:DP} and explains the relation between the parameters $\phi_{g,n}$ and $\varphi_{g,n}$. 

\begin{remark}
	In fact, we can sketch a different proof of Theorem~\ref{thm:DP}, which is easier. First, we observe that the $\omega_{g,n}|_{{\rm KdV}}$ (our KdV-vertices) solve the abstract loop equations. In this case the abstract loop equations are equivalent to the Virasoro constrains for the intersection indices of $\psi$-classes, see \cite{Ekappa} or~\cite{Einter}. Then, we observe that all dependence on the variables is through the leaves that are attached to the KdV-vertices. So, since the abstract loop equations are of the local nature, we can apply it to the piece of the graph that consists of one or two KdV vertices of fixed Euler characteristic, and this will imply the same property for the whole sum over graphs. 
	
	A big disadvantage of this approach is that we do not see that we can represent in this way \emph{any solution} of the abstract loop equations. This can be done only through a link to Theorem~\ref{bllb}, as we did in our proof.
\end{remark}

\subsection{Renormalization by $\DP_{0,1}$}
\label{S37a}
In this paragraph, we simplify a bit the diagrammatics of Section~\ref{S36}. The set of graphs $ \mathcal{G}_{g,n}$ is infinite, since without changing the topology: we can attach an arbitrary number of $\Phi_{0,1}$-vertices to each of the KdV vertex ; and we can replace a $(0,2)$-KdV vertex by an arbitrarily long sequence alternatively made of $(0,2)$-KdV vertices and $\Phi_{0,2}$ vertices. Nevertheless, the sum \eqref{Gngi} is finite because only a finite number of them have non-zero weight -- this has to do with the absence of poles at $p_i$'s in $\phi_{g,n}$ as well as $\omega_{0,1}|_{{\rm KdV}}$ and $\omega_{0,2}|_{{\rm KdV}}$, see Section~\ref{sec:structureZ}. The argument of Section~\ref{sec:structureZ} can be revisited in terms of graphs in the following way.

\begin{lemma}
\label{renom}
The weight $\omega_{g,n}^{\Gamma}$ vanishes if the graph $\Gamma$ has one of the following properties:
\begin{itemize}
\item there exists a $(0,1)$-KdV vertex incident to an internal edge.
\item there exists a KdV vertex attached to $k$ $(0,1)$-$\Phi$-vertices, with total valency $n + k$ for $n\geq 0$, $k \geq 1$, and genus $g \geq 0$, such that $2g - 2 + n+k > 0$ and $k > 3g - 3 + n$.
\item there exists an internal $(0,2)$-KdV vertex.
\end{itemize}
\end{lemma}
\begin{proof} $(0,1)$ KdV vertices that are incident to an internal edge $e$ have weight $\alpha_{i} \zeta_{i,e}^2\dd z_{i,e}$ for $z_e \in U_i$, and they are paired with $\Phi$-vertices whose weight is a holomorphic function, hence yield a zero weight. 
	
	For the second statement, we need to prove that:
$$
\Res_{\zeta_1 \rightarrow p_i} \cdots \Res_{\zeta_k \rightarrow p_i} \Bigg[\omega_{g,n + k}|_{{\rm KdV}}(w_1,\ldots,w_n,\xi_{1},\ldots,\xi_k)\,\prod_{j = 1}^k \Phi_{0,1}(\xi_j)\Bigg] = 0
$$
whenever $2g - 2 + n +k > 0$ and $k > 3g - 3 + n$. Let us consider a pairing of $\omega_{g,n + k}|_{{\rm KdV}}$ with $k$ $\Phi_{0,1}$-vertices, with $k \geq 1$. Since $2g - 2 + n + k > 0$, the pairings only involve residues at $p_i$'s. According to the definition \eqref{dVphi}, $\Phi_{0,1}$ behaves as $O(\zeta_i^5)$ around $p_i$, therefore it can only gives a non-zero result when it is paired with $\zeta_i^{-2(d_i + 1)}\,\dd\zeta_i$ for $d_i \geq 2$. Since the dimension of $\overline{\mathcal{M}}_{g,n + k}$ is $3g - 3 + n + k$, $\omega_{g,n + k}|_{{\rm KdV}}$ is a linear combination of terms $\prod_{i} \zeta_i^{-2(d_i + 1)}\,\dd\zeta_i$ with $\sum_{i} d_i = 3g - 3 + n + k$. Having the non-vanishing condition $d_i \geq 2$ for $i \in \llbracket 1,k \rrbracket$ implies:
$$
 k \leq 3g - 3 + n.
$$
In particular, any graph where $k\geq 1$ $\Phi_{0,1}$-vertices are attached to a $(0,2+k)$ vertex receives a zero weight.  

To prove the last statement, we need to consider:
$$
\Res_{z \rightarrow p_{i}} \Res_{z' \rightarrow p_{i'}} \frac{\dd \zeta_1\otimes\dd\zeta}{(\zeta_1 - \zeta)^2}\,\Phi_{h,k}(z,z',\ldots)\,\frac{\dd \zeta'\otimes\dd\zeta_2}{(\zeta' - \zeta_2)^2} 
$$
or
$$
\Res_{z \rightarrow p_{i}} \Res_{z' \rightarrow p_{i'}} \frac{\dd\zeta_1\otimes\dd\zeta}{(\zeta_1 - \zeta)^2}\,\Phi_{h,k}(\zeta,\ldots)\Phi_{h',k'}(\zeta',\cdots)\,\frac{\dd \zeta'\otimes\dd\zeta_2}{(\zeta' - \zeta_2)^2}.
$$
Since $\Phi$'s are holomorphic at $(p_i,p_{i'})$, these two residues are zero.
\end{proof}

For $(g,n) \neq (0,1)$, let us define the set of reduced bipartite graphs $ \mathcal{G}_{g,n}^{\Box}$ which do not have $(0,1)$-vertices, 
and in which each $(0,2)$-KdV vertex is attached to at least one leaf. By convention, $ \mathcal{G}_{0,1}^{\Box}$ consists of a single graph, made of a $(0,1)$-KdV vertex attached to the leaf.
  
\begin{lemma}
\label{aru}For any $g \geq 0$ and $n \geq 1$, $ \mathcal{G}_{g,n}^{\Box}$ is finite.
\end{lemma}
\begin{proof} The statement is obvious for $(g,n)=(0,1),(0,2)$. Consider $\Gamma \in  \mathcal{G}_{g,n}^{\Box}$, $2g-2+n>0$. According to our defining rules, there are no  $(0,1)$ KdV or $\Phi$-vertices, and at most $n$ $(0,2)$-KdV vertices. Since the total $g$ is fixed, the total number of KdV- or $\Phi$-vertices carrying a positive genus is bounded. We have the two basic relations:
\bea
& & \sum_{k,h} k\big(\#{\rm KdV}_{h,k} + \#\Phi_{h,k}\big) = 2\#{\rm edge} + n \nonumber  \\
\label{combre}& & 1 + \#{\rm edge} - \sum_{h,k} \big(\#{\rm KdV}_{h,k} + \#\Phi_{h,k}\big) = b_1(\Gamma) \leq g 
\eea
and we deduce:
\beq
\label{missing} \#\frac{1}{2} \sum_{k \geq 3} (\#\Phi_{0,k} + \#{\rm KdV}_{0,k}) \leq g - 1 + n
\eeq
In particular, there exists only a finite number of KdV vertices. Since the graph is bipartite, it also means that there exists a finite number of $\Phi$-vertices -- which was the piece of information missing in \eqref{missing}. We conclude that there exists a finite number of vertices, so we can only form a finite number of graphs.
\end{proof}

To a graph $\Gamma \in  \mathcal{G}^{\Box}_{g,n}$, we assign a new weight $\omega_{g,n}^{\Gamma,\Box}$. It is defined following the steps of \S~\ref{S36}, but now each $(h,k)$-KdV vertex is assigned a renormalized weight $\omega_{h,k}|^{\Box}_{{\rm KdV}}$, that incorporates the effect of blossoming $\Phi_{0,1}$ vertices (see Figure~\ref{omphi5}). We ruled out in $ \mathcal{G}_{g,n}^{\Box}$ graphs which had zero weight according to Lemma~\ref{renom}, so the result of the sum over reduced graphs is the same:
\beq
\label{re:ombox}\omega_{g,n} = \sum_{\Gamma \in  \mathcal{G}_{g,n}^{\Box}} \omega_{g,n}^{\Gamma,\Box}
\eeq
This formula holds for any $g$ and $n$. The first graphs in $\mathcal{G}_{g,n}^{\Box}$ are given in Figure~\ref{omphi7}.

\begin{remark} The proof of Lemma~\ref{renom} tells us that
\label{nonono}\beq
\label{norenom}\omega_{0,2}|_{{\rm KdV}}^{\Box} = \omega_{0,2}|_{{\rm KdV}},\qquad \omega_{0,3}|_{{\rm KdV}}^{\Box} = \omega_{0,3}|_{{\rm KdV}},\qquad \omega_{1,1}|_{{\rm KdV}}^{\Box} = \omega_{1,1}|_{{\rm KdV}}
\eeq
which is equivalent to the observation we made in \S~\ref{sec:structureZ}.
\end{remark}

\begin{figure}[h!]
\begin{center}
\includegraphics[width=0.9\textwidth]{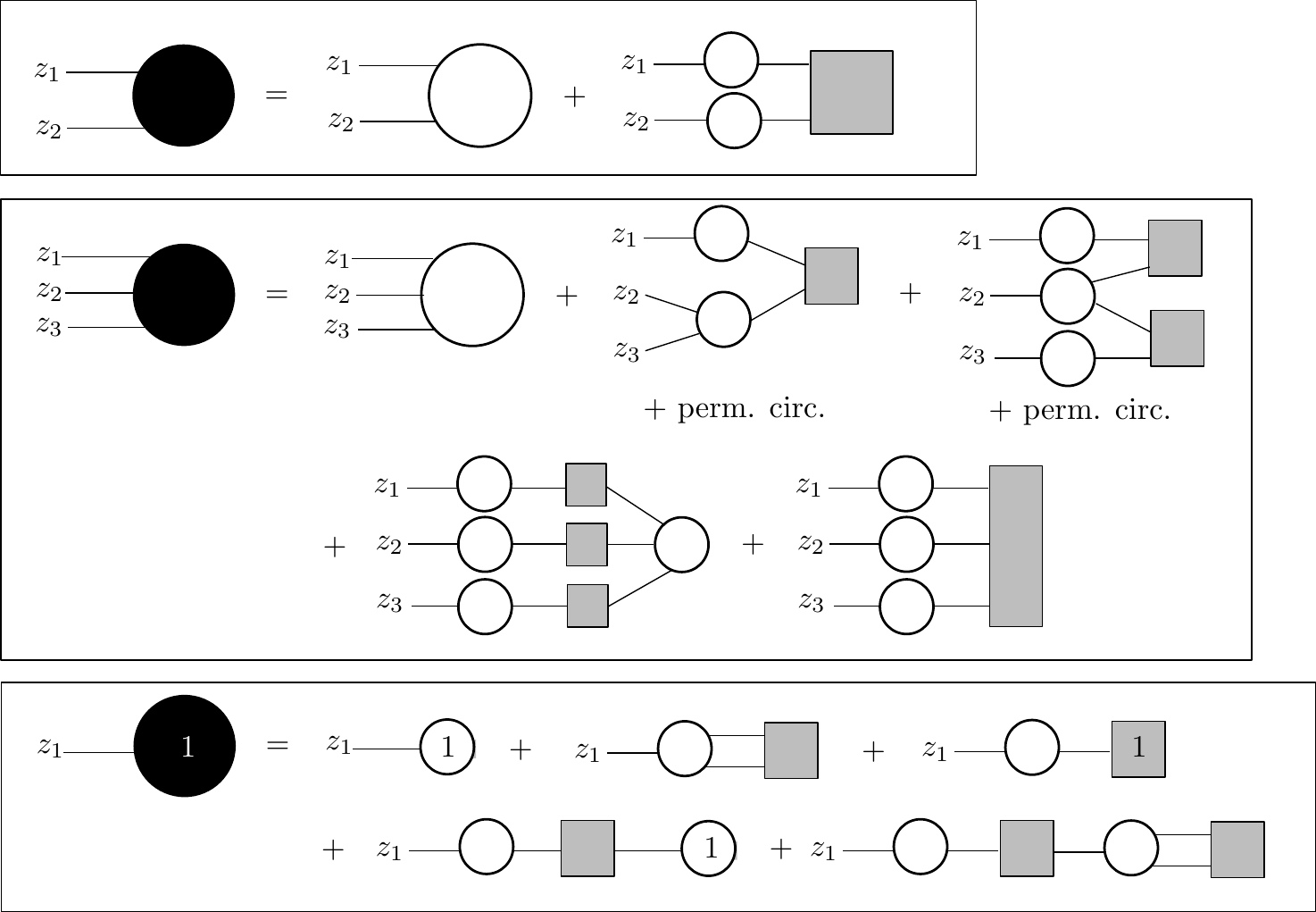}
\end{center}
\caption{\label{omphi7} Diagrammatic representation of correlators as a sum over $\mathcal{G}_{g,n}^{\Box}$ for $(g,n) = (0,2)$, $(0,3)$ and $(1,1)$. We only indicate the genus of vertices when greater or equal to $1$. It is not necessary to put a grey dot on $(0,2)$, $(0,3)$ and $(1,1)$-KdV vertices thanks to the no-renormalization Remark~\ref{nonono}.}
\end{figure}

% % % % % % % % % % % % % % % %
% % % % % % % % % % % % % % % %
% % % % % % % % % % % % % % % %
% % % % % % % % % % % % % % % %
% % % % % % % % % % % % % % % %
% % % % % % % % % % % % % % % %

\section{Some basic properties}

\subsection{Even and odd parts}\label{sec:even-odd-parts} Note that Theorem~\ref{thm:omega-graphs} describes the forms $\omega_{g,n}$ completely while Theorem~\ref{thm:DP} just gives the purely odd part of their expansion. In this section we analyze the difference between odd and even parts of $\omega_{g,n}$.

The first easy statement is the following.

\begin{proposition} \label{prop:oddness}
	If all forms $\phi_{g,n}$ are odd in each variable, then the forms $\omega_{g,n}$ defined by \eqref{Gngi} are also odd in each variable -- for $(g,n) \neq (0,2)$. 
\end{proposition}

\begin{proof}
	Indeed, consider the formula for $\omega_{g,n}(z_1,\ldots,z_n)$ given by Theorem~\ref{thm:omega-graphs}. Dependence on the variables $z_a$, $a \in \llbracket 1,n \rrbracket$, is expressed via the leaves of the underlying graphs. All leaves are attached to KdV-vertices. If a leaf is attached to a KdV-vertex of positive Euler characteristic, its contribution to the weight of the graph is purely odd, as it follows from Equation~\eqref{eq:omegalocal}. The same is true in the exceptional situation when we have one $(0,1)$-KdV vertex. 
	
	So, the only situation that we have to consider is when the $i$th leaf is attached to a $(0,2)$-KdV vertex. But in this case this $(0,2)$-KdV vertex is connected by an edge $e$ to a $\Phi_{h,k}$ vertex (we excluded the case of $(g,n)=(0,2)$).
	
	Consider the piece of the graph that consists of this leaf, the $(0,2)$-KdV vertex, and the internal edge that connects this $(0,2)$-KdV vertex to a $\Phi_{k,h}$ vertex. The contribution of this piece of the graph to the to the weight of the graph is equal 
	\begin{equation} \label{eq:evendep}
	\Res_{z_e\to z_a} \frac{\dd \zeta_{i,a} \otimes\dd \zeta_{i,e}}{(\zeta_{i,a}-\zeta_{i,e})^2} \int^{z_e}_{p_{i}} \phi_{h,k}(\ldots,z_e,\ldots)=\phi_{h,k}(\ldots,z_a,\ldots)|_{z_e=z_i}.
	\end{equation}
	for $z_a \in U_i$. Since we assume that $\phi_{h,k}$ is odd in each of its variables, the contribution of this piece of the graph is also odd in $z_i$. 
\end{proof}

Let us restrict $\omega_{g,n}$ (respectively, $\phi_{g,n}$) to $U_{i_1}\times \cdots \times U_{i,n}$. We can uniquely represent $\omega_{g,n}$ (resp., $\phi_{g,n}$) as a sum of $2^n$ forms, where each of these forms is either even or odd in each of its variable leaving in $U_{i_j}$, $j \in \llbracket 1,n \rrbracket$. Let us denote by $\omega_{g,n_o/n_e}$ (resp., $\phi_{g,n_o/n_e}$) the summand that is odd in the first $n_o$ variables and is even in the last $n_e$ variables, $n_o+n_e=n$. Since we assume the forms to be symmetric, we loose no generality when we make statements only about $\omega_{g,n_o/n_e}$ (resp., $\phi_{g,n_o/n_e}$).

The same argument that we used in the proof of Proposition~\ref{prop:oddness} implies:

\begin{lemma}
	The form $\omega_{g,n_o/n_e}$ is holomorphic in the last $n_e$ variables. 
\end{lemma}

\begin{proof}
	The only way $\omega_{g,n_o/n_e}$ depends on its last $n_e$ variables is through the formula given by Equation~\eqref{eq:evendep}, and the forms $\phi_{g.n}$ are assumed to be holomorphic.
\end{proof}

Now, in order to decouple the even parts of the forms $\omega_{g,n}$ and $\phi_{g,n}$, we introduce a new notation. Let $s_{i,\ell}$, $i \in \llbracket 1,s \rrbracket$ and $\ell \geq 0$ be formal variables. Define the operator 
$$
O:=\sum_{\ell \geq 0} \sum_{i=1}^s \varsigma_{i,\ell} \Res_{\zeta_i \to 0} \zeta_i^{-2\ell-2}. 
$$
This operator substitutes formal variables $\varsigma_{i,\ell}$ instead of the terms $\zeta_i^{2\ell+1}\dd\zeta_i$ in the expansions of (holomorphic) differential forms. 

We introduce new differential forms, 
\bea
\omega^O_{g,n}& := & \sum_{k \geq 0} \frac{1}{k!} \bigg[\prod_{i=1}^k O^{(n+a)}\bigg]\omega_{g,n/k}; \nonumber \\ 
\phi^O_{g,n}& := & \sum_{k \geq 0} \frac{1}{k!} \bigg[\prod_{i=1}^k O^{(n+a)}\bigg]\phi_{g,n/k}; \nonumber
\eea
Here by $O^{(a)}$ we denote the action of the operator $O$ on $a$-th variable. So, this way we recollect the dependence on all even variables in formal power series in $\varsigma_{a,\ell}$. One exception for this definition is the case of $\omega_{0,2}$, where we modify in this way only the non-singular part. 

An immediate consequence of the argument we used to proof Proposition~\ref{prop:oddness} is the following statement.

\begin{proposition}\label{prop:oddandeven}\label{decun}
	The topological recursion~\eqref{btoporec} applied to $\phi^O_{g,n}$ gives $\omega^O_{g,n}$, and its purely holomorphic part is $\varphi_{g,n}^O$. 
\end{proposition}

Note that in order to redefine $\phi_{0,2}$ in this proposition one needs to re-define $B$ in Equation~\eqref{btoporec}. This proposition explains completely how the even parts of the differential forms are decoupled.

\subsection{Elementary group action properties}

The diagrammatic representation of \S~\ref{S36} has two main inputs: the weights assigned to KdV vertices, that we should now call ``reference vertices", and the weights assigned to $\Phi$-vertices. 

Proposition~\ref{prop:oddandeven} allows us to restrict our attention only to the purely odd forms, and in this case the diagrammatic representation of \S~\ref{S36} is equivalent to the presentation in terms of differential operators given by Theorem~\ref{thm:DP}, and we will use it throughout this section.

\begin{lemma}
Let us give a weight $\omega_{g,n}$ to the reference vertices, and $-\Phi_{g,n}$ to the $\Phi$-vertices. Then, the sum over graphs compute $\omega_{g,n}|_{{\rm KdV}}$.
\end{lemma}
\begin{proof}
We can invert \eqref{eq:defZ}:
$$
Z = \exp\Big(\sum_{g \geq 0} \sum_{n \geq 1} \DP_{g,n}\Big) \prod_{i = 1}^s Z_i,\qquad \prod_{i  = 1}^{s} Z_i = \exp\Big(-\sum_{g \geq 0} \sum_{n \geq 1} \DP_{g,n}\Big)\,Z
$$
Since the diagrammatic representation of \S~\ref{S36} follows from the first expression, we deduce that $\omega_{g,n}|_{{\rm KdV}}$ as the same diagrammatic representation with the changes announced. And we remark that the position/degree of poles of $\omega_{g,n}|_{{\rm KdV}}$ -- that were used to perform the renormalization steps of \S~\ref{S37} -- are the same in $\omega_{g,n}$.
\end{proof}

Theorem~\ref{thm:DP} showed that any sequence of admissible correlators solution of the abstract loop equations can be expressed explicitly in terms of the KdV correlators. Instead of the KdV correlators, one could have chosen any other solution of the abstract loop equation, and obtain a similar expression.
\begin{lemma}
Let $(\omega_{g,n})_{g,n}$ and $(\omega_{g,n}|_{{\rm ref}})_{g,n}$ be two solutions of abstract loop equations. They are associated to KdV-blobs $(\phi_{g,n})_{g,n}$ and $(\phi_{g,n}|_{{\rm ref}})_{g,n}$. Let us give a weight $\omega_{g,n}|_{{\rm ref}}$ to reference vertices, and $\big[\phi_{g,n} - \phi_{g,n}|_{{\rm ref}}\big]$ to $\Phi$-vertices. Then, the sum over graphs computes $\omega_{g,n}$.
\end{lemma}
\begin{proof} We write:
$$
Z_{{\rm ref}} = \exp\Big(\sum_{g \geq 0} \sum_{n \geq 1} \DP_{g,n}|_{{\rm ref}}\Big)\,\prod_{i = 1}^{s} Z_i,\qquad Z = \exp\Big(\sum_{g \geq 0} \sum_{n \geq 1} \DP_{g,n}\Big)\prod_{i = 1}^s Z_i
$$
Since the $\DP_{g,n}$ form a commutative algebra, we have:
$$
Z = \exp\Big(\sum_{g \geq 0} \sum_{n \geq 1} \big[\DP_{g,n} - \DP_{g,n}|_{{\rm ref}}\big]\Big) Z_{{\rm ref}}
$$
hence the diagrammatic representation.
\end{proof}

\section{Variational formulae}

\subsection{Variation of $\alpha_i$}
\label{var1}
We study the action of the flow:
$$
\phi_{0,1} \rightarrow \phi_{0,1} + t\,\zeta_i^2\dd\zeta_i,\qquad t \in \mathbb{R}
$$
on the correlators. Its infinitesimal generator is $\partial_{\alpha_i}$. 

\begin{proposition}
For any $g \geq 0$ and $n \geq 1$:
\beq
\label{unub}\partial_{\alpha_i} \omega_{g,n}(z_1,\ldots,z_n) = \Res_{z \rightarrow p_i} \frac{\zeta_i^3}{3}\,\omega_{g,n + 1}(z,z_1,\ldots,z_n)
\eeq
\end{proposition}

\begin{proof}
The equation is obvious for $(g,n) = (0,1)$ and $(0,2)$, so we assume now $2g - 2 + n > 0$. The KdV vertices are homogeneous:
\beq
\label{eccehomo}\partial_{\alpha_i} \omega_{g,n}|_{{\rm KdV}}(z_1,\ldots,z_n) = (2 - 2g - n)\alpha_i^{-1}\,\omega_{g,n}|_{{\rm KdV}}(z_1,\ldots,z_n)
\eeq
for $z_1,\ldots,z_n \in U_i$, and $0$ otherwise. With the dilaton equation -- see \eqref{dilao} in Appendix -- the right-hand side can be transformed into:
$$
(2 - 2g - n) \omega_{g,n}|_{{\rm KdV}}(z_1,\ldots,z_n) = \Res_{z \rightarrow p_i} \frac{\alpha_i\zeta_{i}^3}{3}\,\omega_{g,n + 1}|_{{\rm KdV}}(z,z_1,\ldots,z_n)
$$
so \eqref{unub} is true in the pure KdV case. In general, $\omega_{g,n}$ is expressed as a sum over graphs in $\mathcal{G}_{g,n}$, and their weight depends on $\alpha_i$ only via KdV vertices. Applying $\partial_{\alpha_i}$ amounts to sum over graphs $\Gamma$ in $\mathcal{G}_{g,n}$ together with a marked KdV vertex, and replace the weight of this KdV vertex in $\omega_{g,n}^{\Gamma}$ by \eqref{eccehomo}. By attaching a new leaf labeled $n + 1$ to the marked KdV vertex, we obtain in this way all graphs $\overline{\Gamma}$ in $\mathcal{G}_{g,n + 1}$ exactly once, and represent $\partial_{\alpha_i}\omega_{g,n}$ as a sum over $\mathcal{G}_{g,n + 1}$. According to \eqref{eccehomo}, since the weights given to graph is a product of local weights, the weight given to $\overline{\Gamma}$ in this sum is:
$$
\Res_{z \rightarrow p_i} \frac{\zeta_i^3}{3}\,\omega_{g,n + 1}^{\overline{\Gamma}}
$$
hence the result by summing over $\overline{\Gamma}$.
\end{proof}

\subsection{General variations}

For any $h \geq 0$ and $k \geq 1$, we can consider the deformation of the initial data and the blobs by holomorphic $k$-forms. If $\kappa_{h,k} \in H^0(\Sigma^{k},K_{\Sigma}^{\boxtimes k})^{\mathfrak{S}_{k}}$, we define the infinitesimal generator $\delta[\kappa_{h,k}]$ of the flow:
\beq
\label{flow51}\phi_{h',k'} \rightarrow \phi_{h',k'} + t\,\delta_{h,h'}\delta_{k,k'}\,\kappa_{h,k},\qquad t \in \mathbb{R}
\eeq
In this paragraph, we restrict to $\kappa_{0,1}$ behaving like $O(\zeta_i^3\dd\zeta_i)$ in the open set $U_i$, so that the flow preserves the quadruple zero of $\Phi_{0,1}$ assumed in \eqref{dVphi}, and the diagrammatic representation of Section~\ref{S3} holds for any $t \in \mathbb{R}$. The flow changing the coefficient of the double zero of $\phi_{0,1}$ was studied separately in \S~\ref{var1}, but is at the end governed by the same formula. So we will state Theorem~\ref{vary} below in full generality.

Let us define:
\beq
\label{eq:Kapap}\mathcal{K}_{h,k}(z_1,\ldots,z_k) = \int^{z_1}\!\!\!\cdots\int^{z_k} \kappa_{h,k}
\eeq
where the integral up to $z_{\ell} \in U_{i_{\ell}}$ should starts from the point $p_{i_{\ell}}$. And conversely we have:
\beq
\label{eq:ff}\kappa_{h,k}(z_1,\ldots,z_k) = \Res_{z_1' \rightarrow z_1} \cdots \Res_{z_k' \rightarrow z_k} \Bigg[\bigotimes_{\ell = 1}^k \frac{\dd \zeta_{i_{\ell},\ell}\otimes\dd \zeta_{i_{\ell},\ell}'}{(\zeta_{i_{\ell},\ell} - \zeta_{i_{\ell},\ell}')^2}\,\,\mathcal{K}_{h,k}(z_1',\ldots,z_k')\Bigg]
\eeq
where $\zeta_{i,\ell}$ is a local coordinate in $U_{i}$ such that $x(z_{\ell}) = \zeta_{i,\ell}^2/2 + x(p_{i})$.

Let us compute the variation of $\omega_{g,n}$ using the diagrammatic representation of \eqref{re:ombox}. By construction, we always have:
$$
\delta[\kappa_{h,k}]\cdot\omega_{g,n}|_{{\rm KdV}} = 0
$$
For $(h,k) \neq (0,1)$, we deduce:
$$
\delta[\kappa_{h,k}]\cdot \omega_{g,n}|_{{\rm KdV}}^{\Box} = 0
$$
but since the $\Box$ is a renormalization by $\Phi_{0,1}$-vertices, it is affected by $(0,1)$ flows (see Figure~\ref{omphi9}). We have for $2g - 2 + n > 0$,
\beq
\label{eq:var01}\delta[\kappa_{0,1}] \cdot \omega_{g,n}|_{{\rm KdV}}^{\Box}(z_{[1,n]}) = \sum_{i} \Res_{w \rightarrow p_i} \mathcal{K}_{0,1}(w)\cdot \omega_{g,n + 1}|_{{\rm KdV}}(w,z_{[1,n]})
\eeq
where for a set $I$, we put $z_{I} = (z_i)_{i \in I}$. This formula is also valid for $(g,n) = (0,2)$ and $(0,1)$. Indeed, the $(0,2)$-KdV vertices are not renormalized so the left-hand side vanishes, and in the right-hand side, $\omega_{0,3}|_{{\rm KdV}}$ has a double pole, so its pairing with $\Phi_{0,1}$ which has atmost a quadruple zero gives zero. For $(g,n) = (0,1)$, \eqref{eq:var01} coincides with formula \eqref{eq:ff}. 

\begin{figure}[h!]
\begin{center}
\includegraphics[width=0.6\textwidth]{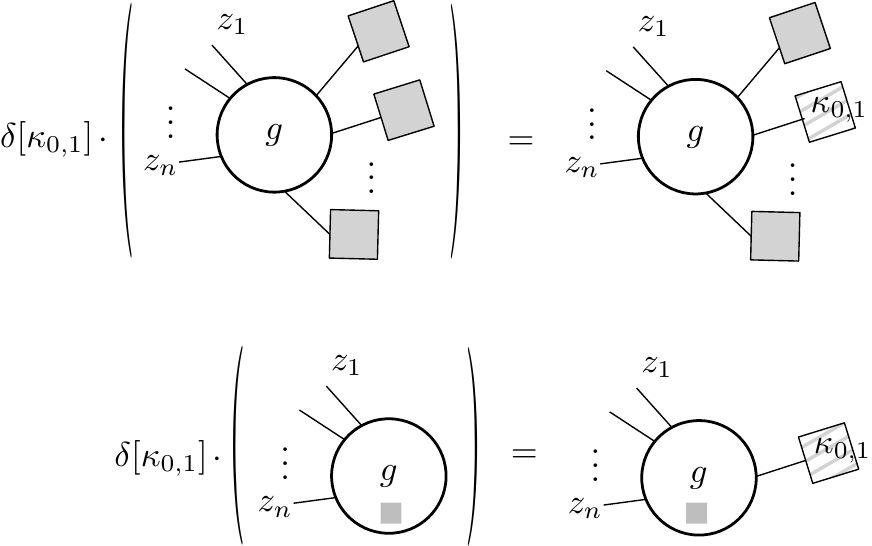}
\caption{\label{omphi9} If there are $m$ $\Phi_{0,1}$-vertices in the left-hand side, the graph comes with symmetry factor $1/m!$ ; since there are $m$ ways to replace one $\Phi_{0,1}$-vertex by a $\mathcal{K}_{0,1}$-vertex, it becomes a $1/(m - 1)!$ in the right-hand side, which is also the symmetry factor of the graph. This is summarized by the second graphical equation.}
\end{center}
\end{figure}

Let us define:
\beq
\label{eq:Edef}\mathcal{E}_{g,n;h,k}(z_{[1,n]};w_{[1,k]}) := \sum_{\substack{J \vdash [1,k] \\ I_1\sqcup\cdots \sqcup I_{[J]} = [1,n]}} \sum_{\substack{h_1,\ldots,h_{[L]} \geq 0 \\ g = h + (\sum_{i} h_i) + k - [J] \\ (h_i,|J_i| + |I_i|) \neq (0,1)}} \frac{1}{{\rm s}_{J,I,\mathbf{h}}} \bigotimes_{i = 1}^{[J]} \omega_{h_i,|I_i| + |J_i|}(z_{I_i},w_{J_i})
\eeq
In this sum, $J$ is a partition of $\llbracket 1,k \rrbracket$ into $[J]$ non-empty subsets, while $I$ is a partition of $\llbracket 1,n \rrbracket$ into possibly empty subsets. $s_{J,I,\mathbf{h}}$ is the symmetry factor of this data.

We remark that the sum itself is empty -- and the corresponding $\mathcal{E}$ equal to zero by convention -- if $2h - 2 + k > 2g - 2 + n$, or if $2h -2 + k = 2g - 2 + n$ but $h > g$. In the maximal case $(h,k) = (g,n)$, we have:
$$
\mathcal{E}_{g,n;g,n}(z_{[1,n]};w_{[1,n]}) = \bigotimes_{i = 1}^n \omega_{0,2}(z_i,w_i)
$$
so its pairing with $\mathcal{K}_{g,n}$ produces $\kappa_{g,n}$ itself, as expected. In the minimal case $(h,k) = (0,1)$, we have:
\beq
\label{E01}\mathcal{E}_{g,n;0,1}(z_{[1,n]};w) = \omega_{g,n + 1}(z_{[1,n]},w)
\eeq
$(h,k) = (0,2)$ is an example with non-trivial symmetry factor:
\bea
\mathcal{E}_{g,n;0,2}(z_{[1,n]};w_1,w_2) & = & \frac{1}{2}\Big[\omega_{g - 1,n + 2}(z_{[1,n]},w_1,w_2) \nonumber \\
& & + \sum_{\substack{J \subseteq [1,n] \\ 0 \leq g' \leq g \\ (J,g') \neq (\emptyset,0),(I,g)}} \omega_{g',|J| + 1}(z_{J},w_1) \otimes \omega_{g-g',n + 1 - |J|}(z_{I\setminus J},w_2)\Big] \nonumber
\eea

We find the variational formula:
\begin{theorem}
For $n \geq 1$ and $g \geq 0$:
\label{vary}\beq
\label{eq:var}\delta[\kappa_{h,k}]\cdot\omega_{g,n}(z_{[1,n]}) = \sum_{1 \leq i_1,\ldots,i_k \leq s} \Res_{w_1 \rightarrow p_{i_1}} \cdots \Res_{w_k \rightarrow p_{i_k}} \mathcal{E}_{g,n;h,k}(z_{[1,n]};w_{[1,k]})\,\mathcal{K}_{h,k}(w_{[1,k]})
\eeq
\end{theorem}
In Section~\ref{Freenerg}, we define $\omega_{g,0} := F_{g}$ in a way such that \eqref{eq:var} also holds for $n = 0$, see Corollary~\ref{varaF}.

\begin{proof}
In a graph $\Gamma \in \mathcal{G}_{g,n}$, the variation $\delta[\kappa_{h,k}]$ amounts to replacing one of the $\Phi_{h,k}$-vertices by a vertex with a weight $\mathcal{K}_{h,k}$, and sum over all possible ways to do this substitution. When working with reduced graphs $\Gamma \in \mathcal{G}_{g,n}^{\Box}$, one has to distinguish whether $(h,k) = (0,1)$ or not. For $(h,k) = (0,1)$, we have seen that the variation is equivalent to adding an edge to one of the renormalized KdV vertex, and pair it with a $\mathcal{K}_{0,1}$-vertex (one can check that the symmetry factors are automatically accounted for). So, we exactly get \eqref{eq:var} with \eqref{E01}. If $(h,k) \neq (0,1)$, let us consider the graph $\Gamma'$ obtained after removing of the vertex targeted by the substitution, and considering as new leaves the edges that we had to cut. This graph will in general have $r \geq 1$ connected components $\Gamma'_1,\ldots,\Gamma'_r$, which have their own genera $h_1,\ldots,h_r$. The initial leaves are distributed in possibly empty subsets $I_i$ belonging to $\Gamma_i'$, while the new leaves are distributing among non-empty subsets $J_1,\ldots,J_r$ of $\Gamma_i'$, for $i \in \llbracket 1,r \rrbracket$. All graphs of topology $(h_i,n_i)$ with $n_i = |I_i| + |J_i|$ may appear as $\Gamma_i'$. The only constraints is that $(h_i,n_i) \neq (0,1)$, since in the initial graph $\Gamma$, there were not $(0,1)$-KdV vertex. Translating this decomposition into weights, we see that the weight of $\delta[\kappa_{h,k}]\cdot\omega_{g,n}^{\Gamma,\Box}$ is the sum over all possible choice of a $\Phi_{h,k}$-vertex of $\bigotimes_{i = 1}^r \omega_{h_i,n_i}^{\Gamma'_i,\Box}$, paired a $\mathcal{K}_{h,k}$-vertex. Summing over all initial graphs $\Gamma \in \mathcal{G}_{g,n}$ and thus all $\Gamma_i'$, one recognizes \eqref{eq:var} with $\mathcal{E}$'s defined in \eqref{eq:Edef} (see Figure~\ref{omphi10}).
\end{proof}

\begin{figure}[h!]
\begin{center}
\includegraphics[width=0.8\textwidth]{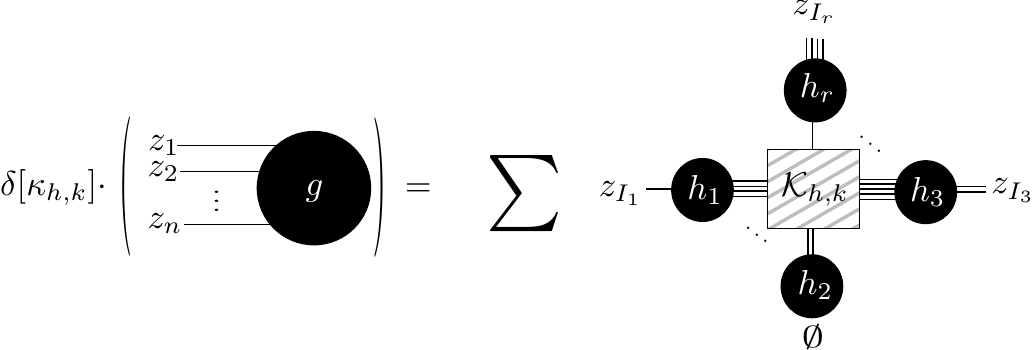}
\caption{\label{omphi10} The variational equation for $(h,k) \neq (0,1)$.}
\end{center}
\end{figure}

\section{Free energies}
\label{Freenerg}
\subsection{Definition of $F_g$}

The graphical formalism developed in \S~\ref{S3} allows also to define numbers $(F_g)_{g \geq 0}$, called the ``free energies", as was done for the usual topological recursion in \cite{EOFg}. Indeed, we can extend the definition of the weight of the graph given by Equation~\eqref{resGamma} to graphs without leaves. At this occasion we also waive the condition that vertices have positive valency. We should then complete the definition by giving weight to $(g,0)$ KdV and $\Phi$-vertices. For the KdV vertices, a natural convention is to take the orbifold Euler characteristic of $\mathcal{M}_{g,0}$ \cite{HarerZagier}
$$
F_{g}|_{{\rm KdV}} := \omega_{g,0}|_{{\rm KdV}} := (-\alpha_i)^{2 - 2g}\,\chi_{{\rm orb}}(\mathcal{M}_{g,0}) = \frac{(-\alpha_i)^{2 - 2g}\,B_{2g}}{2g(2g - 2)}
$$
in terms of the Bernoulli numbers $B_0 = 1$, $B_2 = -1/6$, $B_4 = 1/30$, etc. And for $\Phi$-vertices, we choose arbitrary numbers $\Phi_{g,0}$. Then, we extend Equation~\eqref{Gngi} for $n = 0$:
\begin{definition}
\begin{equation}
\label{Fgdef} F_g:= \omega_{g,0} := \sum_{\Gamma\in\mathcal{G}_{g,0}} \omega_{g,0}^\Gamma = \sum_{\Gamma \in \mathcal{G}_{g,0}^{\Box}} \omega_{g,0}^{\Gamma,\Box}
\end{equation}
\end{definition}

With this definition, $F_g$ always contains the two terms $F_{g}|_{{\rm KdV}} + \Phi_{g,0}$, and this is the only contribution where $0$-valent vertices are involved. The combinatorial relations \eqref{combre} imply that $F_0$ and $F_1$ contain no other terms. The first non trivial graphs appear for $g = 2$ (see Figure~\ref{f2graphl}).  The combinatorial arguments leading to the variational formula of Theorem~\ref{vary} work in the same way for $F_{g}$:
\begin{theorem}
\label{varaF} For any $g \geq 0$
$$
\delta[\kappa_{h,k}]\cdot F_{g} = \sum_{1 \leq i_1,\ldots,i_k \leq s}\Res_{z_1 \rightarrow p_{i_1}} \cdots \Res_{z_k \rightarrow p_{i_k}} \int^{z_1}\cdots\int^{z_1}\mathcal{E}_{g,0;h,k}(z_1,\ldots,z_k)\,\mathcal{K}_{h,k}(z_1,\ldots,z_k)
$$
where $\dd_{1}\cdots\dd_{k}\mathcal{K}_{h,k} = \kappa_{h,k}$ as in \eqref{eq:Kapap}.
\end{theorem}

To summarize, to any solution $(\omega_{g,n})_{g,n}$ of abstract loop equations, and any sequence of integration constants $\Phi_{g,0} \in \mathbb{C}$, we have defined numbers $\omega_{g,0}:=F_{g}$ so that the variational formula of Theorem~\ref{vary} is valid for any $g,n$.

\begin{figure}[h!]
\includegraphics[width=0.9\textwidth]{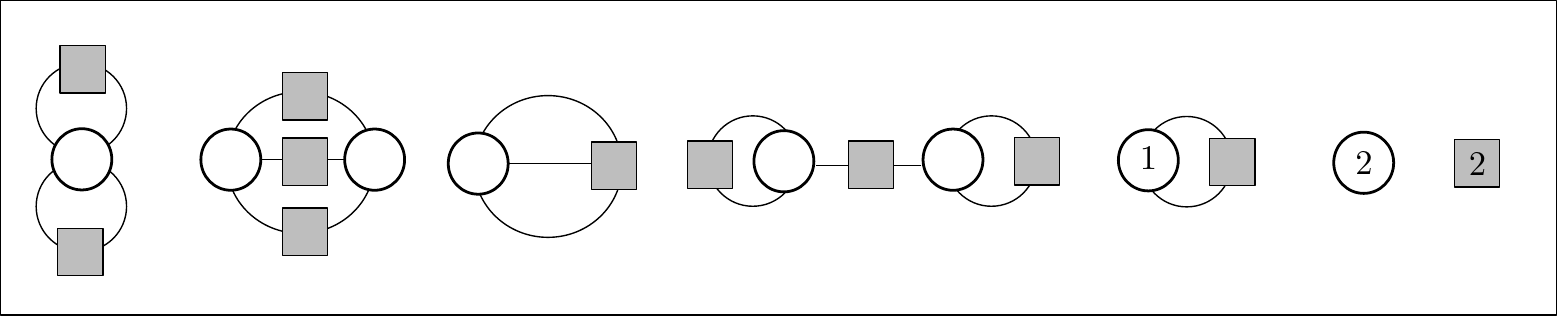}
\caption{\label{f2graphl} The graphs in $\mathcal{G}_{g = 2,0}^{\Box}$.}
\end{figure}

\subsection{Expression in terms of graphs with leaves}

We are going to prove another representation of the $F_g$'s
\begin{theorem}
\label{popore} For $g \geq 2$:
\bea
F_{g} & = & \!\!\!\Phi_{g,0} + F_{g}|_{{\rm KdV}} \nonumber \\
& & \!\!\!+ \frac{1}{2 - 2g}\Bigg\{\sum_{\substack{h \geq 0 \\ k \geq 1}} (2 - 2h - k) \!\!\!\!\!\!\sum_{1 \leq i_1,\ldots,i_k \leq s}\! \Big[\prod_{\ell = 1}^k \Res_{z_{\ell} \rightarrow p_{i_{\ell}}}\Big] \Big(\int^{z_1}\!\!\!\!\cdots\int^{z_k} \!\!\!\phi_{h,k}\Big) \mathcal{E}_{g,0;h,k}(z_1,\ldots,z_k)\Bigg\} \nonumber
\eea
\end{theorem}
The expressions for $\mathcal{E}$'s are given in Equation~\eqref{eq:Edef}. The proof shows that the sum truncates at $h \leq g$ and involves a finite number of terms. In particular we have $\mathcal{E}_{g,1;0,1}(z) = \omega_{g,1}(z)$. For the usual topological recursion, $\phi_{h,k} \neq 0$ for $(h,k) \neq (0,1)$ and $(0,2)$, and we retrieve the definition of $F_g$'s given in \cite{EOFg}:
\beq
\label{usualin} F_{g} = \Phi_{g,0} + \frac{1}{2 - 2g} \sum_{i = 1}^s \Res_{z \rightarrow p_i} \Big(\int^{z} \phi_{0,1}\Big) \omega_{g,1}(z)
\eeq

The starting point for the proof is to establish a dilaton-type equation for renormalized KdV vertices:

\begin{lemma}
\label{22gn}For any $g \geq 0$ and $n \geq 1$:
$$
(2 - 2g - n)\omega_{g,n}|_{{\rm KdV}}^{\Box}(z_1,\ldots,z_n) = \sum_{i = 1}^{s} \Res_{z \rightarrow p_i} \Big(\int^{z} \phi_{0,1}\Big)\omega_{g,1}|_{{\rm KdV}}^{\Box}(z_1,\ldots,z_n,z)
$$
\end{lemma}
\begin{proof} This is Theorem 4.7 in \cite{EOFg}, but to be self-contained, we give a proof here. For $(g,n) = (0,2)$ and $(0,1)$, this follows directly from the expression \eqref{eq:03KdV} of $\omega_{0,3}|^{\Box}_{{\rm KdV}} = \omega_{0,3}|_{{\rm KdV}}$. Now we assume $2g - 2 + n > 0$. We also assume all variables $z_1,\ldots,z_n$ to be in the same $U_i$, otherwise the result is trivial since both sides are $0$. In the pure KdV case, $\int^{z}_{p_i} \phi_{0,1} = \alpha_i \zeta_i^3/3$, and the equation is a consequence of the dilaton equation -- see \eqref{dilao} in Appendix. In general, we have by definition $\int^{z}_{p_i} \phi_{0,1}(z) = \alpha_i\zeta_i^3/3 + \Phi_{0,1}(z)$ and
$$
\omega_{g,n + 1}|_{{\rm KdV}}^{\Box}(z,I) = \omega_{g,n}|_{{\rm KdV}}(z,I) + \sum_{m \geq 1} \frac{1}{m!} \Big[\prod_{\ell = 1}^m \Res_{z_{\ell} \rightarrow p_i} \Phi_{0,1}(z_{\ell})\Big]\omega_{g,n + 1 + m}|_{{\rm KdV}}(I,z,z_1,\ldots,z_m)
$$
Then:
\bea
& & \Res_{z \rightarrow p_i} \Big(\int^{z}\phi_{0,1}\Big)\omega_{g,n + 1}|_{{\rm KdV}}^{\Box}(z,I) \nonumber \\
& & = (2 - 2g - n)\omega_{g,n}|_{{\rm KdV}}(I) + \Res_{z \rightarrow p_i} \Phi_{0,1}(z) \omega_{g,n + 1}|_{{\rm KdV}}(z,I) \nonumber \\
& & + \sum_{m \geq 1} \frac{(2 - 2g - n - m)}{m!} \prod_{\ell = 1}^m \Big[\Res_{z_{\ell} \rightarrow p_i} \Phi_{0,1}(z_{\ell})\Big]\,\omega_{g,n + m}(I,z_1,\ldots,z_{m}) \nonumber \\
& & + \sum_{m \geq 1} \frac{1}{m!} \Big[\prod_{\ell = 1}^{m + 1} \Res_{z_{\ell} \rightarrow p_i} \Phi_{0,1}(z_{\ell})\Big]\,\omega_{g,n + m + 1}|_{{\rm KdV}}(z_{m + 1},z_1,\ldots,z_{m}) \nonumber
\eea
The second, third term coupled to $-m/m!$ and fourth term cancel each other, and we find:
\bea
&& \Res_{z \rightarrow p_i} \Big(\int^{z}\phi_{0,1}\Big)\omega_{g,n + 1}|_{{\rm KdV}}^{\Box}(z,I) \nonumber \\
&& = (2 - 2g - n)\bigg\{\omega_{g,n}|_{{\rm KdV}}(I) + \sum_{m \geq 1} \frac{1}{m!} \Big[\prod_{\ell = 1}^m \Res_{z_{\ell} \rightarrow p_i} \Phi_{0,1}(z_{\ell}) \Big]\omega_{g,n + m}|_{{\rm KdV}}(I,z_1,\ldots,z_{m})\bigg\} \nonumber
\eea
which proves the result.
\end{proof}

\begin{notation}
If $\Gamma \in \mathcal{G}_{g,0}$, we denote $\chi_{{\rm KdV}}(\Gamma)$ the sum of $\sum_{a} (2 - 2h_{a} - k_{a})$ running over all KdV vertices -- $h_a$ and $k_a$ denote their genus and valency. $\chi_{\Phi}(\Gamma)$ is defined similarly taking into account only $\Phi$-vertices with type $(h_a,k_a) \neq (0,1)$. By consistency:
$$
\forall \Gamma \in \mathcal{G}_{g,0},\qquad 2g - 2 + \chi_{\rm KdV}(\Gamma) + \chi_{\Phi}(\Gamma) = 0
$$
\end{notation}

To continue the proof of Theorem~\ref{popore}, we observe that applying the operator:
$$\widehat{E}_{{\rm KdV}} := \sum_{i = 1}^s \alpha_i\partial_{\alpha_i}
$$
amounts to marking a KdV-vertex and count it with an extra weight $(2 - 2h_a - k_a)$ if it has type $(h_a,k_a)$.  Therefore, it can be computed by summing of graphs $\overline{\Gamma}$ with one leaf, and replacing in $\omega_{g,n}^{\overline{\Gamma}}$ the local weight $\omega_{h_a,k_a + 1}|_{{\rm KdV}}(z,I)$ attached to the marked KdV-vertex by $(2 - 2h_{a} - k_{a})\omega_{h_a,k_a}|_{{\rm KdV}}(I)$. After Lemma~\ref{22gn}, this operation can be written in terms of residues. So:
\beq
\label{EKdVvar0}\widehat{E}_{{\rm KdV}}\cdot F_{g} = \sum_{i = 1}^s \Res_{z \rightarrow p_i}\Big(\int^{z}\phi_{0,1}\Big)\omega_{g,1}(z)
\eeq
Similarly, the operator:
\beq
\label{EPhivar0} \widehat{E}_{\Phi} = \sum_{\substack{h \geq 0 \\ k \geq 1}} (2 - 2h - k)\,\delta[\phi_{h,k}]
\eeq
amounts to marking a $\Phi_{h,k}$-vertex and counting it with an extra weight $(2 - 2h - k)$. By the notation $\delta[\phi_{h,k}]$, we mean the flow \eqref{flow51} obtained by choosing $\kappa_{h,k} = \phi_{h,k}$. It is computed by Proposition~\ref{varaF}:
\bea
& &\widehat{E}_{\Phi}\cdot F_{g} = \!\!\!\!\!\sum_{\substack{h \geq 0 \\ k \geq 1 \\ 2 - 2h - k > 0}} \!\!\!\!(2 - 2h - k)\!\!\! \sum_{1 \leq i_1,\ldots,i_k \leq s} \Big[\prod_{\ell = 1}^k \Res_{z_{\ell} \rightarrow p_{i_{\ell}}}\Big] \Big(\int^{z_1}\!\!\!\!\cdots\int^{z_k} \phi_{h,k}\Big)\mathcal{E}_{g,0;h,k}(z_1,\ldots,z_k) \nonumber \\
 \label{phisum}
\eea
The terms $(h,k) = (0,1)$ are excluded from the sum, and $(0,2)$ does not appear because of the Euler characteristic in prefactor. By the  argument already used in Lemma~\ref{aru}, we know that for a given $g$, the genus and valency of the $\Phi$ vertices involved in a graph of $\mathcal{G}_{g,0}$ is bounded, the sum in \eqref{phisum} has only finitely many terms.

On the other hand, since all KdV vertices of a given graph $\Gamma \in \mathcal{G}_{g,0}$ are marked in the expression $\widehat{E}_{{\rm KdV}}\cdot F_{g}$, we also have:
\beq
\label{EKdVvar}\widehat{E}_{{\rm KdV}}\cdot F_{g} = \sum_{\Gamma \in \mathcal{G}_{g,0}'} \chi_{{\rm KdV}}(\Gamma)\,\omega_{g,0}^{\Gamma} = \sum_{\Gamma \in \mathcal{G}_{g,0}'} \big(2 - 2g - \chi_{\Phi}(\Gamma)\big)\,\omega_{g,0}^{\Gamma}
\eeq
where $\mathcal{G}_{g,0}'$ is $\mathcal{G}_{g,0}$ minus the graph with a $0$-valent $\Phi$-vertex and no KdV vertices. Likewise, since $\widehat{E}_{\Phi}$ is just marking all $\Phi$-vertices:
\beq
\label{Ephivar} \widehat{E}_{\Phi}\cdot F_{g} = \sum_{\Gamma \in \mathcal{G}_{g,0}} \chi_{{\rm \Phi}}(\Gamma)\,\omega_{g,0}^{\Gamma}
\eeq
Putting together \eqref{EKdVvar} and \eqref{Ephivar}:
$$
\big(\widehat{E}_{{\rm KdV}} + \widehat{E}_{{\rm KdV}}\big)\cdot F_{g} = (2 - 2g)F_{g}
$$
Since we have computed the left-hand side independently in \eqref{EKdVvar0} and \eqref{EPhivar0}, we obtain a formula for $F_g$ if $g \neq 2$ (the case $g = 0$ is trivial). To get the claim of Theorem~\ref{popore}, we remark that \eqref{EKdVvar0} is nothing but the $(0,1)$ term missing in \eqref{EPhivar0}. \hfill $\Box$

\begin{remark} For the usual topological recursion, the proof stops at \eqref{EKdVvar}: only $(0,1)$- and $(0,2)$-$\Phi$ vertices can occur, so $\chi_{\Phi}(\Gamma)$ is always $0$, and by comparison with \eqref{EKdVvar0} we find \eqref{usualin}.
\end{remark}

% % % % % % % % % % % % % % % %
% % % % % % % % % % % % % % % %
% % % % % % % % % % % % % % % %
% % % % % % % % % % % % % % % %
% % % % % % % % % % % % % % % %
% % % % % % % % % % % % % % % %

% % % % % % % % % % % % % % % %
% % % % % % % % % % % % % % % %
% % % % % % % % % % % % % % % %
% % % % % % % % % % % % % % % %
% % % % % % % % % % % % % % % %
% % % % % % % % % % % % % % % %

\section{Multi-trace matrix model}
\label{motiv}

\subsection{Definition}

Consider the partition function of the one hermitian matrix model
\beq
\label{eq0}Z_{\hbar}(\mathbf{t}) = \int_{\mathcal{H}_N} \dd M\,\exp\Big(\sum_{k \geq 1} \frac{1}{k!}\,\mathrm{Tr}\,T_k(M^{(1)},\ldots,M^{(k)})\Big)
\eeq
$\mathcal{H}_{N}$ is the space of hermitian matrices of size $N$, and $\dd M$ its canonical Lebesgue measure:
$$
\dd M = \prod_{i = 1}^N \dd M_{i,i} \prod_{1 \leq i < j \leq N} \dd\mathrm{Re}\,M_{i,j}\cdot\dd\mathrm{Im}\,M_{i,j}
$$
The $k$-th linear potential is given by:
$$
T_k(x_1,\ldots,x_k) = \sum_{p_1,\ldots,p_k \geq 1} \frac{t_{p_1,\ldots,p_k}}{p_1\cdots p_k}\,x_1^{p_1}\cdots x_k^{p_k}
$$
$M^{(i)}$ is a $k$-th tensor product consisting of identity matrices, except for the $i$-th factor which is a matrix $M$. In other words, in terms of the eigenvalues $\mu_1,\ldots,\mu_N$ of $M$:
$$
\mathrm{Tr}\,T_k(M^{(1)},\ldots,M^{(k)}) = \sum_{1 \leq i_1,\ldots,i_k \leq N} T_k(\mu_{i_1},\ldots,\mu_{i_k})
$$

Let us introduce the disconnected correlators:
\bea
\overline{W}_n(x_1,\ldots,x_n) & = & \Big\langle\prod_{i = 1}^n \Tr\,\frac{1}{x_i - M} \Big\rangle \nonumber \\
& = & \frac{1}{Z_{\hbar}(\mathbf{t})} \sum_{\ell_1,\ldots,\ell_n \geq 0} \hbar^n\,\ell_1\cdots\ell_n \frac{\partial^n Z_{\hbar}(\mathbf{t})}{\partial t_{\ell_1}\cdots \partial t_{\ell_n}}\,\prod_{i = 1}^n x_i^{-(\ell_i + 1)} \nonumber
\eea
Again, we replace by convention any factor $\ell_i = 0$ by $1$ in this formula. The connected correlators are defined as:
\bea
W_n(x_1,\ldots,x_n) & = & {\rm Cum}_{n}\Big(\mathrm{Tr}\,\frac{1}{x_1 - M},\ldots,{\rm Tr}\,\frac{1}{x_n - M}\Big) \nonumber \\
& = & \sum_{\ell_1,\ldots,\ell_n \geq 0} \hbar^{n}\,\ell_1\ldots \ell_n \frac{\partial^n \ln Z_{\hbar}(\mathbf{t})}{\partial t_{\ell_1}\cdots \partial t_{\ell_n}}\,\prod_{i = 1}^n x_i^{-(\ell_i + 1)} \nonumber
\eea
${\rm Cum}_{n}(O_1,\ldots,O_n)$ is the cumulant expectation value of the observables $O_1,\ldots,O_n$.

\subsection{$\hbar$ expansions} 
\label{fomoim}
We keep the product $N\hbar = u$ fixed. The formula \eqref{eq0} can be considered either:
\begin{itemize}
\item[$(i)$] as a convergent matrix integral, producing a function $Z_{\hbar}(u,\mathbf{t})$. This puts restrictions on the choice of $T_k$ so that the integral converges for all sizes $N$.
\item[$(ii)$] as a formal matrix integral near a convergent data. We set $\mathbf{t} = \mathbf{t}^{\mathrm{ini}} + \bs{\tau}$ where $\mathbf{t}^{\mathrm{ini}}$ is chosen such that the matrix integral for $\mathbf{t}^{\mathrm{ini}}$ is convergent. Then, we consider $Z_{\hbar}(\mathbf{t})/Z_{\hbar}(\mathbf{t}^{{\rm ini}})$ as a formal series in $\bs{\tau}$: we expand the exponential as a power series in $\bs{\tau}$, and exchange the sum with the integral over $\mathcal{H}_{N}$. We obtain a formal series in $\bs{\tau}$ whose coefficients are proportional to moments for the probability measure with $k$-linear potentials $T_{k}^{{\rm ini}}$ on $\mathcal{H}_{N}$.
\end{itemize}
A particular case $(ii)$-G occurs when $t^{\mathrm{ini}}_{2}$ is the only non-zero time, i.e. we expand around a Gaussian measure on $\mathcal{H}_{N}$. In cases $(ii)$, we choose in general $\mathbf{t}^{{\rm ini}}$ independent of $\hbar$, but $\bs{\tau}$ itself could depend on $\hbar$, i.e. we rather introduce a collection $(\bs{\tau}^{(h)})_{h \geq 0}$ of formal variables, such that in total we have an equality of formal series:
$$
T_k = \sum_{h \geq 0} \hbar^{2h - 2 + k}\,T_{h,k}
$$
This allows a combinatorial interpretation of the model in terms of maps (discrete surfaces), where $\hbar$ is coupled to minus their Euler characteristics, and $u$ coupled to the number of vertices, see \S~\ref{revi}.

\begin{definition}
We say that the correlators have an expansion of topological type (TT property) if:
\beq
\label{TTexp}W_n =  \sum_{g \geq 0} \hbar^{2g - 2 + n}\,W_{g,n}
\eeq
\end{definition}

In case $(i)$, under a few extra assumptions, some non-trivial analysis is necessary to study the $\hbar \rightarrow 0$ all-order asymptotic expansion of the partition function. It is proved in \cite{BG11,BGK} in the off-critical, one-cut case, the TT property holds and \eqref{TTexp} is an asymptotic (in general non-convergent) series. In the multi-cut case, the TT property does not hold, for instance one has $W_n \in O(1)$ when $\hbar \rightarrow 0$ for any $n \geq 2$. This case is rather interesting and also related to abstract loop equations, but we defer to future work its study in light of the present article. $(ii)$ is basically a perturbation theory around $Z_{\hbar}(\mathbf{t}^{\mathrm{ini}})$, and the case $(ii)$-G where we perturb around the Gaussian weight is the most commonly studied. As reviewed in \cite{Bstuff}, in the case $(ii)$-G, the correlators $W_{n}$ can a priori defined as elements of $R:= \hbar^{-1}\cdot\mathbb{Q}[\bs{\tau}][[u]][[x_1^{-1},\ldots,x_n^{-1}]][[\hbar]]$, and there exists formal series $W_{g,n} \in \mathbb{Q}[\bs{\tau}][[u]][[x_1^{-1},\ldots,x_n^{-1}]] \subseteq R$ such that \eqref{TTexp} holds as an equality in $R$. It is indeed a fact following from Euler characteristic counting that, to a given order in $\bs{\tau}$, only finitely many powers of $\hbar$ contribute.
%G-R1: Corrected order of [[...]] in formal series
\vspace{0.2cm}

\begin{remark} In case $(i)$, the TT property for the partition function should be formulated as follows: there is an expansion of the form $\ln Z_{\hbar}(\mathbf{t}) = C_{\hbar} + \sum_{g \geq 0} \hbar^{2g - 2}\,F_{g}(\mathbf{t})$ where $C_{\hbar}$ is a locally constant function of $\mathbf{t}$ in the domain where the asymptotic expansion holds.
\end{remark}

\subsection{Virasoro constraints}

The first Schwinger-Dyson equation (ignoring the boundary terms) for this model is
\beq
\label{SD1}\Big\langle \Big(\Tr\,\frac{1}{x - M}\Big)^2 + \sum_{k \geq 1} \frac{\hbar^{k - 2}}{(k - 1)!} \mathrm{Tr}\,\frac{\partial_{1} T_k(M^{(1)},\ldots,M^{(k)})}{(k - 1)!\,(x - M^{(1)})}\Big\rangle = 0
\eeq
It can be proved by integration by parts, see e.g. \cite{Bstuff}. In Laurent expansion at $x \rightarrow \infty$, if we collect the terms of order $x^{-(m + 2)}$ for $m \geq -1$, we find that:
$$
\forall m \geq -1,\qquad L_m\cdot Z = 0
$$
with:
\bea
L_m & = & L_m^{(0)} +  \sum_{k \geq 2} \sum_{p_1,p_2,\ldots,p_k \geq 0} \frac{(p_1 + m)\,t_{p_1,\ldots,p_k}}{(k - 2)!}\,\frac{\partial^k}{\partial t_{p_1 + m}\partial t_{p_2}\cdots \partial t_{p_k}} \\
L_m^{(0)} & = & \hbar^2\,\sum_{p = 0}^{m} p(m - p)\frac{\partial^2}{\partial t_{p}\partial t_{m - p}} + \hbar^{-2}\sum_{p_1 \geq 0} (p_1 + m)\,t_{p_1}\,\frac{\partial}{\partial t_{p_1 + m}} 
\eea
By convention, we set $\partial t_{l} = 0$ for $l < 0$. $L_{m}^{(0)}$ are the usual representation Virasoro operators in the context of matrix models, and satisfy the commutation relations:
$$
[L_{m}^{(0)},L_{n}^{(0)}] = (m - n)L_{m + n}^{(0)}
$$
\begin{lemma}
$(L_m)_{m \geq 1}$ also form a representation of the Virasoro commutation relations: $[L_m,L_n] = (m - n)L_{m + n}$.
\end{lemma}
\begin{proof}
Let us compute the full commutation relations:
\bea
& & [L_m,L_n] \nonumber \\
& = & (m - n)L_{m + n}^{(0)} \nonumber \\
& & + \hbar^{-2}\sum_{p_0 \geq 0} \sum_{k \geq 2} \sum_{p_1,\ldots,p_k \geq 0} \frac{(p_0 + m)(p_1 + n)\,t_{p_1,\ldots,t_{p_k}}}{(k - 2)!} \Bigg[t_{p_0}\,\frac{\partial}{\partial t_{p_0 + m}},\frac{\partial^k}{\partial t_{p_1 + n}\cdots \partial t_{p_k}}\Bigg] \nonumber \\
& & - (m \leftrightarrow n) \nonumber \\
& = & (m - n)L_{m + n}^{(0)} \nonumber \\
& & + \Bigg(- \sum_{k \geq 2} \sum_{p_1,\ldots,p_k \geq 0} \frac{(p_1 + m + n)(p_1 + n)\,t_{p_1,\ldots,p_k}}{(k - 1)!}\,\frac{\partial^{k}}{\partial t_{p_1 + m + n}\partial t_{p_2}\cdots \partial t_{p_k}} \nonumber \\
& & - \sum_{k \geq 2} \sum_{p_1,\ldots,p_k \geq 0} \frac{(p_1 + m)(p_2 + n)\,t_{p_1,\ldots,p_k}}{(k - 2)!}\,\frac{\partial^k}{\partial t_{p_1 + m} \partial t_{p_2 + n} \partial t_{p_3} \cdots \partial t_{p_k}}\Bigg) \nonumber \\
& & - (m \leftrightarrow n) \nonumber
\eea
We observe that the third line is symmetric in $m \leftrightarrow n$, hence does not contribute to the commutator. And the second line can be combined with the first to find $[L_m,L_n] = (m - n)L_{m + n}$.
\end{proof}

Schwinger-Dyson equations involving $n \geq 2$ variables can be derived from \eqref{SD1} by infinitesimal deformations of the potential:
\bea
& & {\rm Cum}_{n}\bigg[\Big\{\Big(\mathrm{Tr}\,\frac{1}{x - M}\Big)^2 + \sum_{k \geq 1} \hbar^{k - 2} \mathrm{Tr}\,\frac{\partial_{1} T_k(M^{(1)},\ldots,M^{(k)})}{(k - 1)!\,(x - M^{(1)})}\Big\},O_{x_2}(M),\ldots,O_{x_n}(M)\bigg]\nonumber \\ %\mathrm{Tr}\,\frac{1}{x_2 - M}, \cdots, \mathrm{Tr}\,\frac{1}{x_n - M}\bigg] \nonumber \\
& & + \sum_{2 \leq i \leq n} {\rm Cum}_{n - 1}\bigg[\Big\{\mathrm{Tr}\,\frac{1}{(x - M)(x_i - M)^2}\Big\}, O_{x_2}(M),\ldots,\hat{i},\ldots, O_{x_2}(M)\bigg] = 0 \nonumber
\eea
where $O_{x_i}(M) = \mathrm{Tr}\,\frac{1}{x_i - M}$, and $\hat{i}$ means omitting the factor with label $i$.

\subsection{Spectral curve}

In case $(i)$ and $(ii)$-G, the assumptions considered respectively in \cite{BGK} and \cite{Bstuff} imply that:
\begin{itemize}
\item[$\bullet$] $W_{g,n}(x_1,\ldots,x_n)$ exists as a holomorphic function on $(\mathbb{C}\setminus\Gamma)^n$ for some segment $\Gamma = [a,b] \subseteq \mathbb{R}$ determined by the model
\item[$\bullet$] $T_{h,k}(x_1,\ldots,x_n)$ exists as a holomorphic function in an open neighborhood $\underline{V}$ of $\Gamma^{k}$.
\end{itemize}
$\mathbb{C}\setminus\Gamma$ can be mapped conformally to the exterior of the unit disk with:
$$
x(z) = \frac{a + b}{2} + \frac{a - b}{4}\Big(z + \frac{1}{z}\Big)
$$
Then, it is known that $y(z) := W_{0,1}(x(z))$ can be analytically continued in some neighborhood of $\{|z| = 1\}$ inside the unit disk. We can take as spectral curve of the model the domain $\underline{\Sigma}$ including the point at $\infty$ (Figure~\ref{Fspcurv}) and as morphism of Riemann surfaces $x\,:\,\Sigma \rightarrow \widehat{\mathbb{C}}$. The involution is $\sigma(z) = 1/z$. There are two simple ramification points $z = \pm 1$, corresponding to the simple branchpoints $x = a$ and $b$. This spectral curve has the topology of a disk with $2$ marked points.

\begin{figure}[h!]
\includegraphics[width=0.55\textwidth]{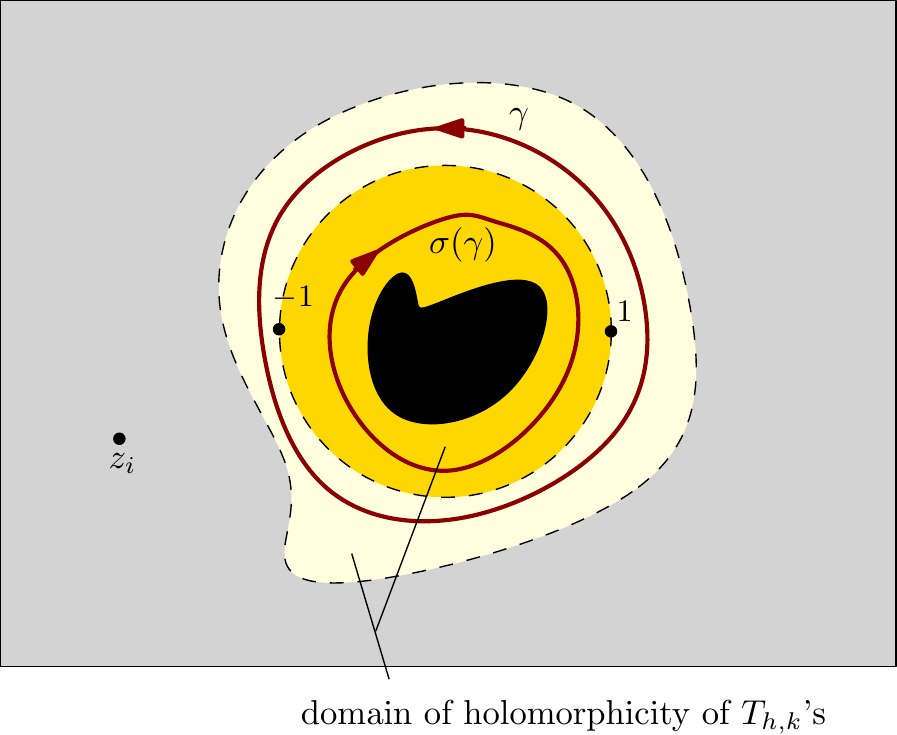}
\caption{\label{Fspcurv} Spectral curve of the matrix model, uniformized by the $z$-coordinate. $\underline{\Sigma}$ is the Riemann sphere minus the black region. The involution $\sigma(z) = 1/z$ exchanges the light-gold and gold colored domains.}
\end{figure}

\subsection{Blobbed topological recursion}
\label{btrstuff}
Let us summarize the results of \cite{Bstuff}. The differential forms:
\beq
\label{woewo}\omega_{g,n}(z_1,\ldots,z_n) = \bigg(W_{g,n}(x_1,\ldots,x_n) + \frac{\delta_{g,0}\delta_{n,2}}{(x(z_1) - x(z_2))^2}\bigg)\dd x(z_1)\cdots\dd x(z_n)
\eeq
are initially defined for $z_i$ outside the unit disk. They actually extend to meromorphic forms on $\underline{\Sigma}$, with poles only at the ramification points $z = \pm 1$ -- and a simple pole at $\infty$ for $\omega_{0,1}$, which satisfy abstract loop equations. This is proved as a consequence of the Schwinger-Dyson equations.

The term we identify with $\mathcal{H}_{1}\omega_{g,n}$ in light of Section~\ref{S2}, is computed in \cite[Equations 5.5 and 4.13]{Bstuff} in terms of the $k$-linear potentials\footnote{We underline that $\sum'$ in \cite[Equation 4.13]{Bstuff} excludes the term $T_{0,2}$, so only $T_{h,k}$ with $2h - 2 + k > 0$ are involved. In this equation, there is a misprint in the condition on genera, which should read $h + \sum f_i + k - 1 - [K] = g$. We also point that the prefactor $1/4{\rm i}\pi$ in \cite[Equation 5.5]{Bstuff} is erroneous and should be replaced with $1/2{\rm i}\pi$.}. It takes the form of a sum over graphs, described hereafter:
\beq
\label{HUHT}\mathcal{H}_{1}\omega_{g,n}(z_1,z_2,\ldots,z_n) = \sum_{\Gamma \in {\rm Plant}_{g,n}^{T}(1)} \varpi_{\Gamma}^{T}(z_1;z_2,\ldots,z_n)
\eeq
On the other hand, we know from Theorem~\ref{bllb} that $\omega_{g,n}$ is determined in terms of its purely holomorphic part $\varphi_{g,n} = \mathcal{H}_{1}\cdots\mathcal{H}_{n}\omega_{g,n}$. To compute it in terms of the $k$-linear potentials, we need to project \eqref{HUHT} to the holomorphic part in the variables $z_2,\ldots,z_n$. The final result is:

\begin{proposition}
\label{mimfsu} For $2g - 2 + n > 0$, we have:
$$
\varphi_{g,n}(z_1,\ldots,z_n) = \sum_{\Gamma \in {\rm Bip}^{\tau}_{g,n}} \varpi^{\tau}_{\Gamma}(z_1,\ldots,z_n)
$$
\end{proposition}
\noindent ${\rm Bip}_{g,n}^{\tau}$ is the set of bipartite graphs $\Gamma$ with the following properties (Figure~\ref{PhiT-1tau}):
\begin{itemize}
\item[$\bullet$] vertices $v$ are of type $\omega$ or $\tau$, and carry a genus $h(v)$, such that $2h(v) - 2 + d(v) > 0$.
\item[$\bullet$] dashed edges can only connect $\omega$ vertices to $\tau$-vertices.
\item[$\bullet$] $\Gamma$ is connected and $b_1(\Gamma) + \sum_{v} h(v) = g$.
\end{itemize}
According to its type, a vertex $v$ has local weight $\omega_{h(v),d(v)}$, or $\tau_{h(v),d(v)}$ given by \eqref{TbulletN}-\eqref{taudefN} in terms of the matrix potentials and the $\omega_{h',k'}$ with $2h' - 2 + k' < 2h(v) - 2 + d(v)$. The internal edge variables are integrated out with the pairing \eqref{pairing0}.

\begin{proof} ${\rm Plant}^{T}_{g,n}(1)$ is a set of bipartite graphs $\Gamma$ with the following properties:
\begin{itemize}
\item[$\bullet$] the set of vertices consists of one root vertex, and a set of $\omega$ vertices. Each vertex $v$ carries a genus $h(v)$. The root vertex must have $2h(v) - 2 + d(v) > 0$, but we do not impose such conditions for $\omega$ vertices.
\item[$\bullet$] Edges are dashed, and can only connect the root vertex to $\omega$ vertices.
\item[$\bullet$] There are $n$ leaves labeled from $1$ to $n$, and they must be incident to $\omega$ vertices. Besides, the leaf labeled $1$ is incident to a $\omega_{0,2}$-vertex, which is itself incident to the root vertex.
\item[$\bullet$] The leaf labeled $1$ is incident to a $\omega_{0,2}$-vertex, which is incident to the root vertex.
\item[$\bullet$] There are $n - 1$ other leaves labeled from $2$ to $n$. The leaves with label $a \in A$ must be incident to the root vertex, and the leaves with label $b \in B$ must be incident to a $\omega$-vertex. There is no restriction for leaves labeled by $c \notin A \sqcup B$.
\item[$\bullet$] All $\omega$ vertices must be incident to at least one leaf.
\item[$\bullet$] $\Gamma$ is connected and $b_1(\Gamma) + \sum_{v} h(v) = g$.
\end{itemize}
We attach leaf variables $z_1,\ldots,z_n$ outside $\gamma$, and integration variables $z_{e}$. To a vertex $v$ with set of incident variables $Z(v)$, we assign a local weight $T_{h(v),d(v)}(Z(v))$ if it is the root, and $\omega_{h(v),d(v)}(Z(v))$ otherwise. $\varpi_{\Gamma}^{T}$ is computed by multiplying the local weights, and integrating the dashed edge variable $z_{e}$ on $\frac{1}{2{\rm i}\pi}\oint_{\gamma}$.

The root vertex can be attached to $\omega$-vertices with no external legs. It is convenient to resum all these contributions by defining a ``dressed" root vertex of genus $h$ and valency $k$, with local weight:
\bea
& & \!\!\!\!\!\! T^{\bullet}_{h,k}(z_1,\ldots,z_k) \nonumber \\
& & \!\!\!\!\!\! = \!\!\!\sum_{\substack{h_0 \geq 0,\,\,r \geq 0 \\ 2h_0 - 2 + k > 0}} \frac{1}{r!} \!\!\!\!\!\!\!\! \sum_{\substack{h_1,\ldots,h_r \geq 0 \\ \ell_1,\ldots,\ell_r \geq 1 \\ h_0 + \sum_{j} (h_j + \ell_j - 1) = h}}\!\!\!\!\!\!\!\! \prod_{\substack{1 \leq j \leq r \\ 1 \leq m \leq \ell_j}} \oint_{\gamma} \frac{\dd z'_{j,m}}{2{\rm i}\pi}\cdot\frac{\omega_{h_j,\ell_j}(Z'_{j})}{\ell_j!} T_{h_0,(k + \sum_{j} \ell_j)}(z_1,\ldots,z_k,z'_{1,1},\ldots,z'_{r,\ell_{r}}) \nonumber \\
\label{TbulletN} & &
\eea
Here, $Z'_j:= \{z'_{j,m},\,\, m \in \llbracket 1,\ell_j \rrbracket\}$, the $z_1,\ldots,z_k$ are outside the contour, and to define properly the integral in case some $\omega_{0,2}$ appear, we choose the contours such that  $(|z'_{q,i}|)_{q,i}$ is decreasing with lexicographic order on $(q,i)$. By symmetry of $\omega_{0,2}$ and $T$, the result does not depend on this order. Remark that, for a given $h_0$, we can attach an arbitrary number of $\omega_{0,1}$'s to the $T$ without changing the topology. For instance, in genus $0$ we have for $k \geq 3$:
$$
T^{\bullet}_{0,k}(z_1,\ldots,z_k) = \sum_{r \geq 0} \frac{1}{r!} \,\,\oint_{\gamma^{r}} T_{0,k + r}(z_1,\ldots,z_k,z'_1,\ldots,z'_r) \prod_{j = 1}^{r} \omega_{0,1}(z_{j}')
$$

\begin{figure}[h!]
\includegraphics[width=0.85\textwidth]{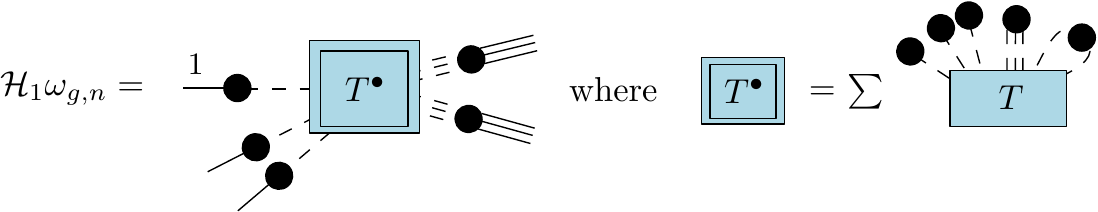}
\caption{\label{FPhiT-0} $\mathcal{H}_{1}\omega_{g,n}$ according to \cite{Bstuff}, in terms of the dressed $T$-vertices. The black vertices are $\omega$'s.}
\end{figure}

$\mathcal{H}_{1}\omega_{g,n}$ is then represented as the sum over the graphs $\Gamma$ in  ${\rm Plant}^{T}_{g,n}(1)$ such that all $\omega$ vertices are incident to at least a leaf. In this new sum, the weight of $\Gamma$ is produced like $\varpi_{\Gamma}^{T}$, but assigning $T^{\bullet}$ as local weight for the root vertex. In particular, the first condition ensure that no $\omega_{0,1}$-vertex, and thus there are finitely many graphs. We now compute $\varphi_{g,n}$ by projecting to the holomorphic part the variables carried by the leaves $2,\ldots,n$. It amounts to replacing the weight of a $\omega$-vertex $v$ by $\mathcal{H}_{A(v)}\omega_{h(v),d(v)}$, where $A(v)$ is its (non-empty) set of incident leaves. We can write:
$$
\mathcal{H}_{A}\omega_{h,d} = \sum_{A' \sqcup B = \llbracket 1,d \rrbracket \setminus B} \mathcal{H}_{A\sqcup A'}\mathcal{P}_{B}\omega_{h,d}
$$
and we compute the summands with \eqref{WPgn}. They only involve $\varphi_{h',d'}$'s and $\omega_{h',d'}^{\mathcal{P}}$'s with $0 < 2h' - 2 + d' < 2g - 2 + n$ (Figure~\ref{FPhiT-1}), and the latter can be replaced with its definition (Figure~\ref{OmP}):
$$
\omega_{h',d'}^{\mathcal{P}}(z_1,\ldots,z_n) = \sum_{1 \leq i_1,\ldots,i_{d'} \leq s} \Res_{z_1' \rightarrow p_{i_1}} \cdots \Res_{z_{d'}' \rightarrow p_{i_{d'}}} \omega_{h',d'}(z_{1}',\ldots,z'_{d'})\,\prod_{j = 1}^{d'} \int_{p_{i_j}}^{z_{j}'} \omega_{0,2}(\cdot,z_j)
$$
This results in a graphical recursion for $\varphi_{g,n}$, involving lower $\varphi$'s and $\omega$'s. Solving this recursion, we arrive to a first representation:
$$
\varphi_{g,n}(z_1,\ldots,z_n) = \sum_{\Gamma \in {\rm PreBip}^{T^{\bullet}}_{g,n}} \varpi^{T^{\bullet}}_{\Gamma}(z_1,\ldots,z_n)
$$
${\rm PreBip}^{T^{\bullet}}_{g,n}$ is the set of graphs $\Gamma$ with the properties:
\begin{itemize}
\item[$\bullet$] Vertices $v$ are of type $\omega$ or $T^{\bullet}$, carry a genus $h(v)$, and the valency $d(v)$ should satisfy $2h(v) - 2 + d(v) > 0$ for type $T^{\bullet}$, and $2h(v) - 2 + d(v) \geq 0$ for type $\omega$.
\item[$\bullet$] Edges are either dashed or plain. Dashed edges can connect $\omega$ and $T^{\bullet}$ vertices, while plain edges can only connect a $\omega_{0,2}$-vertex to a $\omega_{h,k}$-vertex with $2h - 2 + k > 0$.
\item[$\bullet$] Cutting a plain edge cannot disconnect the graph.
\item[$\bullet$] There are $n$ labeled leaves: they must be incident to a $\omega_{0,2}$-vertex which is itself incident to a $T^{\bullet}$ vertex.
\item[$\bullet$] $\Gamma$ is connected and $b_1(\Gamma) + \sum_{v} h(v) = g$.
\end{itemize}
The weight $\varpi^{T^{\bullet}}_{\Gamma}$ is computed by integrating out the edge variables with $\frac{1}{2{\rm i}\pi}\oint_{\gamma}$ for dashed edges, and with the pairing \eqref{pairing0} for plain edges, i.e. $$\sum_{i = 1}^{s} \Res_{z \rightarrow p_i} \big(\int^{z}_{p_i} \omega_{0,2}(z',\cdots)\big) \omega_{h,k}(z,\cdots).$$

The last step is to get rid of the internal $\omega_{0,2}$-vertices. They are incident either to a $\omega_{h,k}$-vertex with $2h - 2 + k > 0$ and a $T^{\bullet}$-vertex, or to two $T^{\bullet}$-vertices. If we remove all $\omega$ vertices but the $(0,2)$, and contract the internal $(0,2)$, we get connected components which are can be arbitrary graphs formed on $T^{\bullet}$-vertices. Therefore, it suggests to define:
\beq
\label{taudefN}\tau_{g,n}(z_1,\ldots,z_k) = \sum_{\mathcal{T} \in {\rm Graph}_{g,n}} \Big\langle\prod_{\ell = 1}^{n} \omega_{0,2}(z_{\ell},z_{\ell}') \prod_{v = {\rm vertex}} T_{h(v),d(v)}(Z'(v)) \prod_{\substack{e = \{v,w\} \\ {\rm edge}}} \omega_{0,2}(z'_{v},z'_{w})\Big\rangle_{\gamma}
\eeq
%G-R1: graphs instead of trees
\begin{itemize}
\item[$\bullet$] $\mathcal{G}$ is a graph made of $T^{\bullet}$-vertices with a genus $h(v)$ and a valency $d(v)$ such that $2h(v) - 2 + d(v) > 0$, and of $n$ labeled univalent vertices (by convention, they have genus $0$).
\item[$\bullet$] $\sum_{v} h(v) = g$.
\end{itemize}
In formula \eqref{taudefN}, we have distributed integration variables $z_{e}'$ on the edges of the trees, and $Z'(v)$ is the set of variables incident to a vertex $v$. The bracket indicates that the internal variables $z'$ should be integrated on $\frac{1}{2{\rm i}\pi}\oint_{\gamma}$. Resumming the graphs of $T^{\bullet}$-vertices (Figure~\ref{PhiT-1tau}), we obtain the representation of Proposition~\ref{mimfsu} that contains only $\tau$- and $\omega$-vertices with $2h(v) - 2 + d(v) > 0$.
\end{proof}

\begin{figure}[h!]
\includegraphics[width=0.95\textwidth]{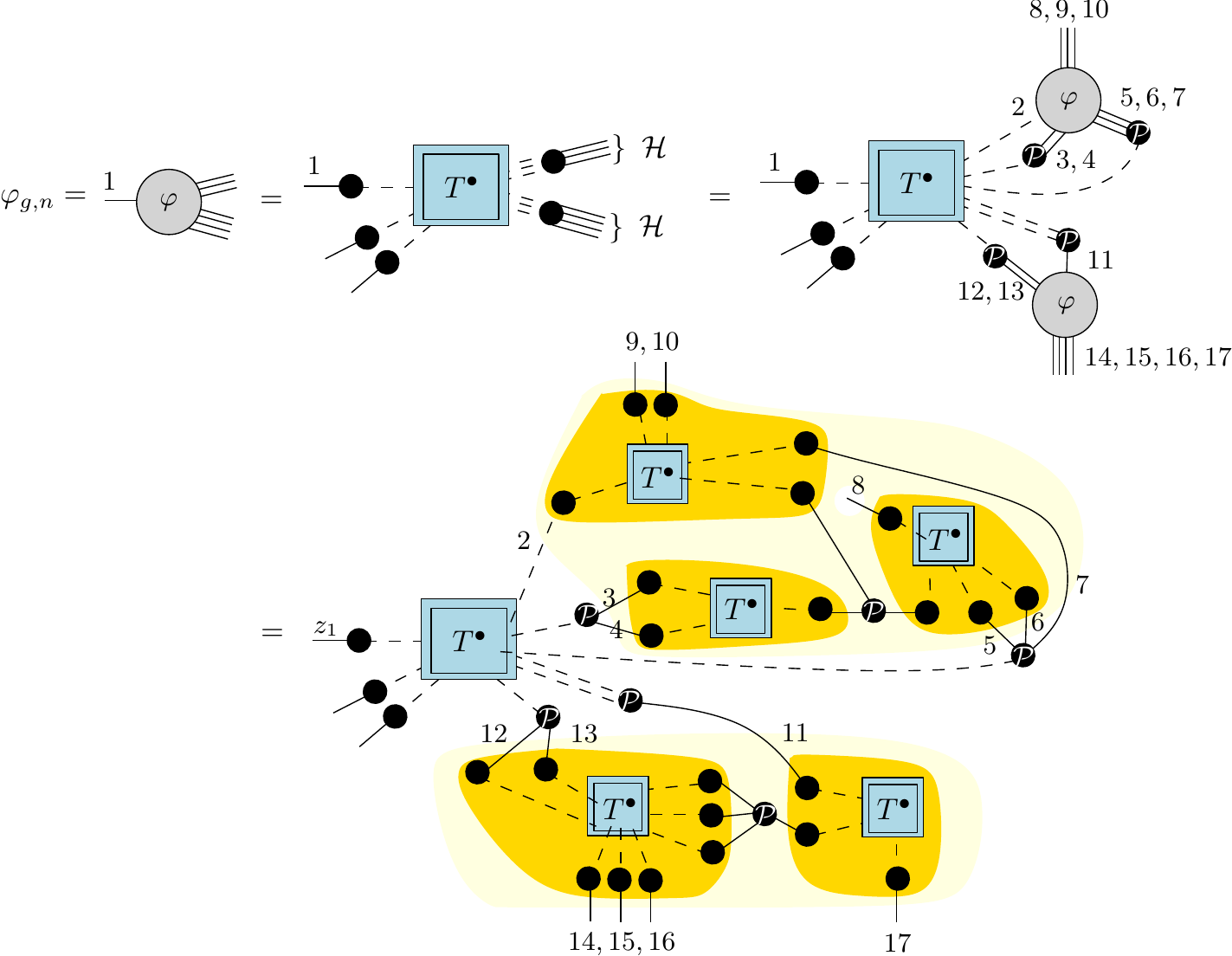}
\caption{\label{FPhiT-1} Recursive relation for $\varphi$ in terms of $T^{\bullet}$. The black vertices with a label $\mathcal{P}$ have local weight $\omega^{\mathcal{P}}$. To exemplify the properties of the graphs solving this recursion, we have displayed in the second line the second step of the recursion in a case where two steps is enough to convert all $\varphi$'s to $T^{\bullet}$. The light-gold (resp. gold) colored domain is the $\varphi$'s revealed at the first step (resp. at the second step). The labels $2,\ldots,17$ have no other meaning than facilitating the comparison between the first and second line.}
\end{figure}

\begin{figure}[h!]
\includegraphics[width=0.65\textwidth]{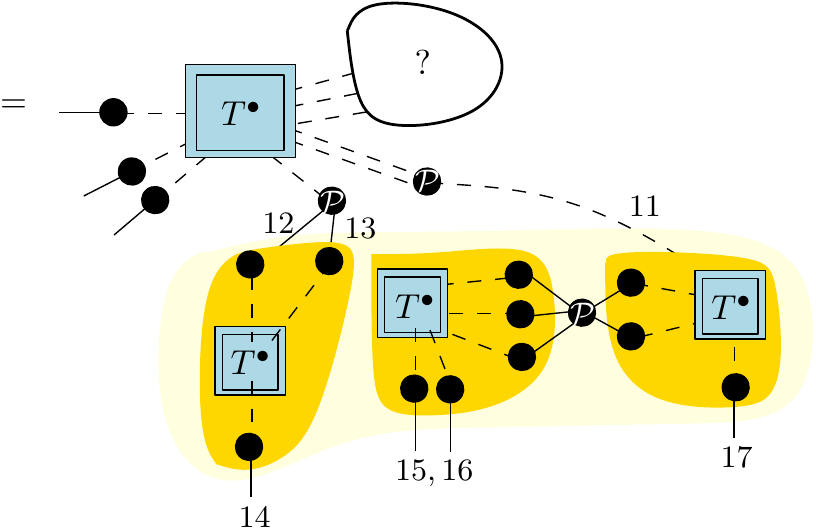}
\caption{\label{FPhiT-1err} This graph cannot occur at the second step of the recursion from Figure~\ref{FPhiT-1}, because the graph in the light-gold colored domain (supposedly a $\varphi$) is not connected. This is the origin of the condition ``plain edges are non-separating".}
\end{figure}

\begin{figure}[h!]
\includegraphics[width=0.2\textwidth]{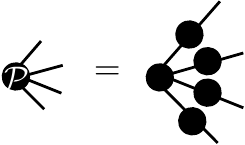}
\caption{\label{OmP}}
\end{figure}

\begin{figure}[h!]
\includegraphics[width=0.6\textwidth]{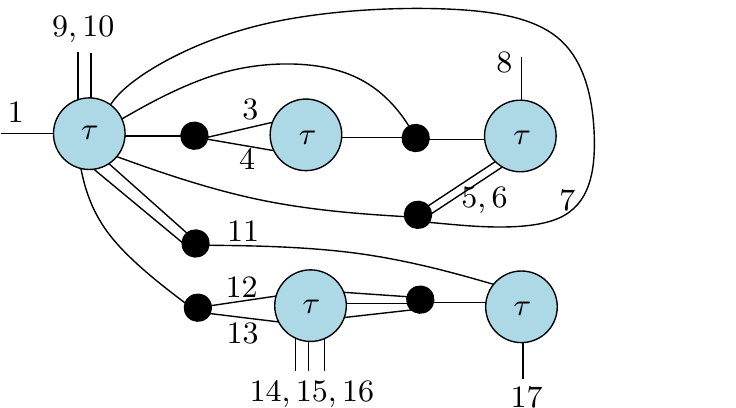}
\caption{\label{PhiT-1tau} We contracted the $\omega_{0,2}$-vertices in the example of Figure~\ref{FPhiT-1}: the two $T^{\bullet}$-vertices connected by the edge $2$ are resummed in the first $\tau$-vertex. The result is an example of graph in ${\rm Bip}^{\tau}_{g,n}$.}
\end{figure}

\subsection{Combinatorial interpretation (stuffed maps)}
\label{revi}

We agree that an elementary $2$-cell is an (isomorphism class of) oriented topological surface of genus $h$ with $k$ labeled boundaries $\mathscr{B}_1,\ldots,\mathscr{B}_k$, with a set $\mathscr{O}_i \subseteq \mathscr{B}_i$ of marked points labeled from $1$ to $\ell_i = |\mathscr{B}_i|$ (the labeling should respect the cyclic order along the boundary). We call ``edge" the closure of the connected components of $\mathscr{B}_i\setminus \mathscr{O}_i$ in $\mathscr{B}_i$. $\ell_i$ is the perimeter of $\mathscr{B}_i$. A $2$-cell with topology of a disk is called ``face".

We consider (isomorphism classes of) graphs embedded on oriented surfaces, built with the following rules:
\begin{itemize}
\item[$\bullet$] take a finite collection of elementary $2$-cells of arbitrary topology and perimeter lengths ; for each $i \in \llbracket 1,n \rrbracket$, add to this collection a marked face with label $i$.
\item[$\bullet$] take a pairing of edges with opposite orientations, and glue the elementary $2$-cells according to this pairing. The result is an oriented topological surface $\mathscr{S}$, with $n$ marked faces with perimeters $\ell_1,\ldots,\ell_n$. The union of boundaries of the elementary $2$-cells is an embedded graph $\mathscr{G}$ in this surface.
\end{itemize}
If all elementary $2$-cells have the topology of a disk, $(\mathscr{G},\mathscr{S})$ is called a ``map". The model with elementary $2$-cells of arbitrary topology are allowed was introduced in \cite{Bstuff}, and $(\mathscr{G},\mathscr{S})$ was called a ``stuffed map", as a reference to its nested structure. The weight of a stuffed map $(\mathscr{G},\mathscr{S})$ is a product of Boltzmann weights:
\begin{itemize}
\item[$\bullet$] for each elementary $2$-cell of topology $(h,k)$ with perimeter lengths $\ell_1,\ldots,\ell_k$ and which is not a marked face, we put a weight $t_{h;\ell_1,\ldots,\ell_n}$ ;
\item[$\bullet$] for the $i$-th marked face with perimeter $\ell_i$, we put a weight $x_i^{-(\ell_i + 1)}$ ;
\item[$\bullet$] for each vertex, we put a weight $u$ ;
\end{itemize}
and we divide by the number of automorphisms of $(\mathscr{G},\mathscr{S})$. Then, the techniques of \cite{BIPZ} show that $W_{g,n}(x_1,\ldots,x_n)$ defined in \S~\ref{fomoim} is the generating series of all connected stuffed maps of genus $g$ with $n$ boundaries. In $W_{0,1}(x)$, we add by convention the term $u/x$, corresponding to a single vertex embedded in the sphere -- we considered that the marked face has perimeter $0$. The $k$-linear potentials in the matrix model:
$$
T_{h,k}(x_1,\ldots,x_k) = -\delta_{h,0}\delta_{k,2}\,\frac{x^2}{2} + \sum_{\ell_1,\ldots,\ell_k \geq 1} \frac{t_{h;\ell_1,\ldots,\ell_k}}{\ell_1\cdots\ell_k}\,x^{\ell_1}\cdots x^{\ell_k}
$$
can be interpreted as the generating series of elementary $2$-cells of genus $h$ with $k$ boundaries. If $\sum_{\ell \geq 0} A_{n}\dd x/x^{\ell + 1}$ is the Laurent expansion at $x = \infty$ of a holomorphic function $A$ in $\{|z| > 1\} \subseteq \underline{\Sigma}$, and $\sum_{\ell \geq 1} B_\ell\,x^{\ell}/\ell$ is the Laurent expansion at $x = 0$ of a holomorphic function $B$ in a neighborhood of $\{|z| = 1\} \subseteq \underline{\Sigma}$, then the contour integral is:
$$
\frac{1}{2{\rm i}\pi}\oint_{\gamma} A(z)\,B(z) = \sum_{\ell \geq 1} \frac{A_{\ell}B_{\ell}}{\ell}
$$
and it computes the number of objects that are obtained by gluing (rooted) an object counted in $A$ to a (rooted) object counted in $B$ of same (arbitrary) size $\ell$. 

Then, many equations in \S~\ref{btrstuff} acquire a clear combinatorial interpretation\footnote{The shift in \eqref{woewo} for $(0,2)$ does not play a role here, because it is even with respect to the involution $\sigma$, while $T$ is holomorphic and even in a neighborhood of $\gamma$, and $\sigma(\gamma)$ is homologically equivalent to $-\sigma(\gamma)$ in this neighborhood. So, replacing $\omega_{0,2}(z_1,z_2)$ by $W_{0,2}(x_1,x_2)\dd x_1\dd x_2$ does not change the weight of the graphs in \S~\ref{btrstuff}.}: the series expansion (in the variable $x_i = x(z_i)$) of $T^{\bullet}_{g,n}$ at $0$ and of $\tau_{g,n}$ and $\varphi_{g,n}$ at $\infty$ count stuffed maps with particular properties read from the graphs. For instance:
\begin{itemize}
\item[$\bullet$] $\varphi_{g,n}$ is the generating series of stuffed maps of genus $g$ with $n$ marked faces in which, after removing all elementary $2$-cells with stable topology (i.e. $2h - 2 + k > 0$), the marked faces are in disjoint connected components, which have the topology of a cylinder (Figure~\ref{PhiT-1tau}).
\item[$\bullet$] $\omega^{\mathcal{P}}_{g,n}$ is the generating series of stuffed maps of genus $g$ with $n$ marked faces in which, for each $i = 1,\ldots,n$, a cycle homologous to the boundary of the $i$-th marked face has been marked.
\end{itemize}
Since we have shown in Section~\ref{S3} that $\omega^{\mathcal{P}}_{g,n}$ is expressed in terms of intersection numbers on $\overline{\mathcal{M}}_{g,n}$, it would be interesting to know if such a representation has a combinatorial interpretation in terms of combinatorics of maps.

% % % % % % % % % % % % % % % %
% % % % % % % % % % % % % % % %
% % % % % % % % % % % % % % % %
% % % % % % % % % % % % % % % %
% % % % % % % % % % % % % % % %
% % % % % % % % % % % % % % % %

\appendix

\section{Kappa class formula}
\label{kappaapp}

The effect of the forgetful map $\pi_{k}\,:\,\overline{\mathcal{M}}_{g,n + k} \rightarrow \overline{\mathcal{M}}_{g,n}$ can be expressed via Mumford classes $(\kappa_{a})_{a \geq 0}$ \cite{ACmoduli}:
$$
(\pi_k)_*\Big[\prod_{\ell = 1}^n \psi_{\ell}^{d_{\ell}} \cdot \prod_{m = 1}^{k} \psi_{n + m}^{b_m} \Big] = \prod_{\ell = 1}^n \psi_{\ell}^{d_{\ell}}\cdot\bigg(\sum_{I_1 \sqcup \cdots \sqcup I_r = \llbracket 1,k \rrbracket} \prod_{p = 1}^r \kappa_{\sum_{m \in I_p} (b_m - 1)}\bigg)
$$
In particular, the dilaton equation:
\beq
\label{dilao}\int_{\overline{\mathcal{M}}_{g,n + 1}} \prod_{\ell = 1}^n \psi_{\ell}^{d_{\ell}} \psi_{n + 1} = (2g - 2 + n) \int_{\overline{\mathcal{M}}_{g,n}} \prod_{\ell = 1}^{n} \psi_{\ell}^{d_{\ell}}
\eeq
gives $\kappa_0 = 2g - 2 + n$. Therefore, we can rearrange Equation~\eqref{eq:DP01action}:
\bea
& & \ln\big(\exp(\widehat{\phi}_{0,1})Z_i\big)  \nonumber \\
& = & \sum_{\substack{g \geq 0 \\ n \geq 1}} \frac{(\hbar/\alpha_i^2)^{g - 1}}{n!} \sum_{k \geq 0} \frac{1}{k!} \!\!\!\!\!\!\sum_{\substack{d_1,\ldots,d_n \in \mathbb{Z} \\ b_1,\ldots,b_k \geq 2 \\ I_1 \sqcup \cdots \sqcup I_r = \llbracket 1,k \rrbracket}} \!\!\!\! \int_{\overline{\mathcal{M}}_{g,n}} \prod_{\ell = 1}^{n} \frac{t_{d_{\ell}}\psi_{\ell}^{d_{\ell}}}{-\alpha_i} \prod_{p = 1}^r \Big\{ \kappa_{\sum_{m \in I_p} (b_m - 1)} \prod_{m \in I_{p}} \frac{\phi_{0,1}\left[\begin{smallmatrix} i \\
2b_{m}
\end{smallmatrix}\right]}{-\alpha_i} (2b_{m} - 1)!!\Big\}
\nonumber \\
& = & \sum_{\substack{g \geq 0 \\ n \geq 1}}  \frac{(\hbar/\alpha_i^2)^{g - 1}}{n!} \sum_{\substack{r \geq 1 \\ c_1,\ldots,c_r \geq 1 \\ d_1,\ldots,d_n \in \mathbb{Z}}}\frac{1}{r!} \int_{\overline{\mathcal{M}}_{g,n}} \prod_{\ell = 1}^{n} \frac{t_{i,d_{\ell}}\psi_{\ell}^{d_{\ell}}}{-\alpha_i} \prod_{p = 1}^r \kappa_{c_p} \widehat{t}_{i,c_p} \nonumber \\
& = & \sum_{\substack{g \geq 0 \\ n \geq 1}} \frac{\hbar^{g - 1}}{n!} \sum_{d_1,\ldots,d_n \in \mathbb{Z}} \int_{\overline{\mathcal{M}}_{g,n}} \prod_{\ell = 1}^n \psi_{\ell}^{d_{\ell}}\,\exp\Big(\sum_{c \geq 0} \widehat{t}_{i,c}\kappa_{c}\Big) \prod_{\ell = 1}^n t_{i,d_{\ell}} \nonumber
\eea
It was convenient to introduce new parameters $(\widehat{t}_{i,c})_{c \geq 0}$ to parametrize the coefficients in $\phi_{0,1}$. We have set:
$$
\widehat{t}_{i,0} = -\ln(-\alpha_i) = -\ln\big(-\phi_{0,1}\left[\begin{smallmatrix} i \\
2
\end{smallmatrix}\right]\big)
$$
Together with $\kappa_0 = (2g - 2 + n)$, it absorbs the scaling in $\alpha_i$'s. And for $c \geq 1$, we defined:
$$
\widehat{t}_{i,c} = \sum_{q = 1}^{c} \frac{1}{q!} \sum_{\substack{b_1,\ldots,b_q \geq 2 \\ \sum_{m} (b_m - 1) = c}} \prod_{m = 1}^q \frac{\phi_{0,1}\left[\begin{smallmatrix} i \\
2b_{m}
\end{smallmatrix}\right]}{-\alpha_i} (2b_{m} - 1)!!
$$
This change of variables is nicely expressed in terms of generating series:
\bea
1 + \sum_{c \geq 1} \widehat{t}_{i,c}\,u^{c} & = & \sum_{q \geq 0} \frac{1}{q!} \sum_{b_1,\ldots,b_q \geq 2} \prod_{m = 1}^q \frac{\phi_{0,1}\left[\begin{smallmatrix} i \\ 2b_m \end{smallmatrix}\right]}{-\alpha_i} (2b_m - 1)!!\,u^{b_m - 1} \nonumber \\
& = & \exp\bigg\{\sum_{b \geq 2} \frac{\phi_{0,1}\left[\begin{smallmatrix} i \\ 2b \end{smallmatrix}\right]}{-\alpha_i}\,(2b - 1)!!\,u^{b - 1}\bigg\} \nonumber
\eea
Besides, the series in the exponential coincides with the $O(u)$ part of the formal Laplace transform:
$$
\sum_{b \geq 1} \phi_{0,1}\left[\begin{smallmatrix} i \\ 2b \end{smallmatrix}\right]\,(2b - 1)!!\,u^{b - 1} = \frac{1}{(2\pi)^{1/2}u^{3/2}}\int_{\gamma_i} \phi_{0,1}\,e^{-[x - x(p_i)]/u}
$$
$\gamma_i$ is the steepest contour lifting ${\rm Re}\,x \leq 0$ in $U_i$, which is just the real line in the coordinate $\zeta_i$. The previous equation means equality of formal power series in $u$, obtained in the right-hand side by integrating term by term the formal Taylor expansion of $\phi_{0,1}$ near $p_i$. 
We see that the effect of $\exp(\widehat{\phi}_{0,1}\big)$ is to insert the class:
$$
\Lambda_i = \exp\Big(\sum_{c} \widehat{t}_{i,c}\,\kappa_{c}\Big)
$$
in integrals over $\overline{\mathcal{M}}_{g,n}$. In the particular case of $\phi_{h,k} = 0$ for $(h,k) \neq (0,1)$, the formula for the expansion of the correlators when all $z$'s tend to $p_i$ is found by applying the substitution \eqref{substim}:
\beq
\omega_{g,n} \sim \sum_{d_1,\ldots,d_n \in \mathbb{Z}} \int_{\overline{\mathcal{M}}_{g,n}} \exp\Big(\sum_{c \geq 0} \widehat{t}_{i,c}\,\kappa_{c}\Big)\prod_{\ell = 1}^{n} \psi_{\ell}^{d_{\ell}}\,\prod_{\ell = 1}^n \frac{(2d_{\ell} + 1)!!\dd \zeta_{i,\ell}}{\zeta_{i,\ell}^{2d_{\ell} + 2}}
\eeq
for $2g - 2 + n > 0$. This is the final form of the answer given in \cite{Ekappa}, and it is our $\omega_{g,n}^{\Box}|_{{\rm KdV}}$ in \eqref{nonon}. For general $(\phi_{h,k})_{h,k}$, it is also possible to rewrite our general relation between $\omega_{g,n}$ and intersection numbers in the style of \cite[Theorem 4.1]{Einter}, by inserting suitable boundary divisors.

\end{document}